%% file: main.tex
\documentclass[11pt]{article}

\input{preamble}

\title{Solving the Correlation Cluster LP in Sublinear Time}
\author{}

\renewcommand{\epsilon}{\eps}

\newcounter{cphfn}
\newcommand{\cphthanks}[1]{\thanks{#1}\setcounter{cphfn}{\value{footnote}}}
\newcommand{\cphmark}{\footnotemark[\value{cphfn}]}

\author{
Nairen Cao\thanks{New York University, Email: \texttt{nc1827@nyu.edu}}
\and
Vincent Cohen-Addad\thanks{Google Research. Email: \texttt{cohenaddad@google.com}.}
\and
Euiwoong Lee\thanks{University of Michigan. Email: \texttt{euiwoong@umich.edu}. Supported in part by NSF grant CCF-2236669 and Google.}
\and
Shi Li\thanks{Nanjing University. Email: \texttt{shili@nju.edu.cn}. Affiliated with the School of Computer Science in Nanjing University, and supported by the State Key Laboratory for Novel Software Technology and the New Cornerstone Science Laboratory.}
\and
David Rasmussen Lolck\cphthanks{University of Copenhagen (BARC). Emails: \texttt{dalo@di.ku.dk, mthorup@di.ku.dk, shya@di.ku.dk}. Supported by VILLUM Foundation Grant 54451 and Basic Algorithms Research Copenhagen (BARC).}
\and
Alantha Newman\thanks{Université Grenoble Alpes. Email: \texttt{alantha.newman@grenoble-inp.fr}.}
\and
Mikkel Thorup\cphmark
\and
Lukas Vogl\thanks{EPFL. Email: \texttt{lukas.vogl@epfl.ch}. Supported by the Swiss National Science Foundation project 200021-184656 ``Randomness in Problem Instances and Randomized Algorithms''.}
\and
Shuyi Yan\cphmark
\and
Hanwen Zhang\thanks{University of Copenhagen. Email: \texttt{hazh@di.ku.dk}.
Supported by VILLUM Foundation Grant 54451, Basic Algorithms Research Copenhagen (BARC),
and Starting Grant 1054-00032B from the Independent Research Fund Denmark (Sapere Aude).}
}

\begin{document}
\maketitle

\begin{abstract}
\input{abstract}

\end{abstract}

\clearpage 

\tableofcontents 
\clearpage

\section{Introduction}
\input{intro}

\section{Preclustering}

\input{preliminaries}

\section{Multiplicative Weights Update Framework}
\input{mwuframework}

\section{Finding a Partial Clustering with Small Ratio in Polynomial Time}
\label{sec:family-poly}
\input{familyclusters-poly}

\section{Finding One Small Ratio Cluster}
\label{sec:goodratio}
\input{findgoodcluster}
\section{Refinements to Reach Nearly Linear Time}
\label{sec:nearlylinear}
\input{nearlylineartime}

\section{Finding a Partial Clustering with Small Ratio in Sublinear Time}
\label{sec:sublinear}
\input{familyclusters}

\section{MPC Implementation}
\label{sec:mpcimplementation}
\input{mpcalgorithm}

\section{Rounding Algorithms}
\input{rounding-algorithms}





\bibliographystyle{alpha}
\bibliography{main}

\end{document}

%% file: preamble.tex
\usepackage{mathpazo}

\usepackage[english]{babel}

\usepackage[dvipsnames,usenames]{xcolor}
\usepackage[colorlinks=true,pdfpagemode=UseNone,urlcolor=RoyalBlue,linkcolor=RoyalBlue,citecolor=OliveGreen,pdfstartview=FitH]{hyperref}

\usepackage{tikz}
\usetikzlibrary{patterns}
\usetikzlibrary{shapes.geometric}

\usepackage{subcaption}
\usepackage{mathtools}

\usepackage{amsfonts,amsmath,amsthm,thmtools}
\usepackage{nicefrac}
\usepackage{thm-restate}

\usepackage{todonotes}

\usepackage{booktabs}
\usepackage{caption}
\usepackage[numbers, sort]{natbib}

\usepackage{tcolorbox}

\usepackage{enumitem}
\usepackage[capitalise]{cleveref}

\usepackage[multiple]{footmisc}

\usepackage[letterpaper,top=2cm,bottom=2cm,left=3cm,right=3cm,marginparwidth=1.75cm]{geometry}

\usepackage{amsmath}
\usepackage{comment}
\usepackage{mdframed}
\usepackage{xspace}
\usepackage{algorithmic}
\usepackage{algorithm}
\usepackage{enumitem}
\usepackage{graphicx}
\usepackage{soul}

\newtheorem{theorem}{Theorem}
\newtheorem{lemma}[theorem]{Lemma}
\newtheorem{definition}[theorem]{Definition}
\newtheorem{claim}[theorem]{Claim}
\newtheorem{coro}[theorem]{Corollary}

\newcommand{\eps}{\varepsilon}
\newcommand{\covereps}{\gamma}
\newcommand{\esteps}{\beta}

\newcommand{\calC}{\mathcal{C}}

\newcommand{\clcost}{\mathrm{cost}}

\newcommand{\optc}{\clcost(\opt)}
\newcommand{\genclus}{\text{GenerateCluster}}
\newcommand{\genclusbysam}{\text{GenerateClusterBySampling}}

\newcommand{\adm}{\text{adm}}
\newcommand{\calK}{{\mathcal{K}}}

\newcommand{\rmEstCost}{\mathrm{EstMarg}}
\newcommand{\cost}{\mathrm{cost}}
\newcommand{\obj}{\mathrm{obj}}
\newcommand{\EX}{{\mathbb{E}}}

\newcommand{\opt}{\mathrm{OPT}}

\newcommand{\zmwu}{z^{\star}}

\newcommand{\poly}{{\mathrm{poly}}}

\newcommand{\cc}{{Correlation Clustering}\xspace}

\newcommand{\pedge}{{$+$edge}\xspace}

\newcommand{\pedges}{{$+$edges}\xspace}
\newcommand{\medge}{{$-$edge}\xspace}
\newcommand{\medges}{{$-$edges}\xspace}

\newcommand{\costs}{\mathrm{cost}}
\newcommand{\hatcosts}{\widehat{\mathrm{cost}}}
\newcommand{\covers}{{\mathrm{cover}}}
\newcommand{\clusterlp}{\textbf{cluster LP}\xspace}

\newcommand{\clusterLP}{{cluster LP}\xspace}
\newcommand{\coverClusterLP}{{covering cluster LP}\xspace}

\newcommand{\dadm}{d_{\mathrm{adm}}}
\newcommand{\Nadm}{N_{\mathrm{adm}}}
\newcommand{\Ncand}{N_{\mathrm{cand}}}
\newcommand{\eadm}{E_{\mathrm{adm}}}
\newcommand{\dc}{d_{\mathrm{cross}}}
\newcommand{\hatdc}{\hat{d}_{\mathrm{cross}}}
\newcommand{\hatNcand}{\hat{N}_{\mathrm{cand}}}
\newcommand{\Ec}{E_{\mathrm{cross}}}
\newcommand{\supp}{\mathrm{supp}}
\newcommand{\hatt}{\hat{T}}
\newcommand{\hatp}{\hat{p}}
\newcommand{\zt}{z^{(t)}}
\newcommand{\pt}{p^{(t)}}
\newcommand{\marginal}{\mathrm{Marginal}}

\newcommand{\dpos}{d^{+}}

\newcommand{\pureclusterlpratio}{1.485}

\newcommand{\tmwu}{T_{\mathrm{MW}}}
\newcommand{\tmpc}{T_{\mathrm{MPC}}}
\newcommand{\tilt}{\tilde{t}}

\newcommand{\epslarge}{$\eps$-large~}
\newcommand{\epssimilar}{$\eps$-similar~}

\newenvironment{cproof}
{\begin{proof}
 [Proof.]
 \vspace{-1.5\parsep}
}
{ \end{proof}}

\newif\ifshowcomments
\showcommentsfalse  

\newcommand{\lukas}[1]{\ifshowcomments \todo{\color{white} Lukas: #1} \fi}

\newcommand{\nairen}[1]{\ifshowcomments {\color{orange} Nairen: #1} \fi}


%% file: abstract.tex
\cc is a fundamental and widely-studied problem in unsupervised learning and data mining.
The input is a graph and the goal is to construct a clustering minimizing the number of inter-cluster edges plus the number of missing intra-cluster edges.

\cite{cao2024understanding} introduced the {\em cluster LP} for \cc, which they argued 
captures the problem much more succinctly than previous linear
programming formulations.  However, the \clusterLP has exponential
size, with a variable for every possible set of vertices in the input
graph.  Nevertheless, \cite{cao2024understanding} showed how to
find a feasible solution for the \clusterLP in time
$O(n^{\text{poly}(1/\eps)})$ with objective value 
at most $(1+\epsilon)$ times the value of an optimal
solution for the respective \cc instance.  Furthermore, they showed
how to round a solution to the \clusterLP, yielding a
$(\pureclusterlpratio+\eps)$-approximation algorithm for the \cc problem.\footnote{\cite{cao2024understanding} claimed that there is a rounding for the \clusterLP yielding a 1.437-approximation algorithm for \cc.  However, the proof has a bug, and the best bound we can currently prove for the \clusterLP is 1.485.  Proof details will appear shortly in an updated arxiv version of \cite{cao2024understanding}.}

The main technical result of this paper is a new approach to find a
feasible solution for the \clusterLP with objective value at most
$(1+\epsilon)$ of the optimum in time $\widetilde
O(2^{\text{poly}(1/\eps)} n)$, where $n$ is the number of vertices in
the graph.  We also show how to implement the rounding within the same
time bounds, thus achieving a fast $(\pureclusterlpratio+\eps)$-approximation
algorithm for the \cc problem.  This bridges the gap between
state-of-the-art methods for approximating \cc and the recent focus on
fast algorithms.


%% file: intro.tex
\cc, introduced by Bansal, Blum, and Chawla~\cite{BBC04}, is a fundamental problem in unsupervised machine learning that neatly captures the inherent tension between grouping similar data elements and separating dissimilar ones.  Given a complete graph where each
edge receives either a positive or a negative label, representing the
pairwise relationship of its endpoints, the objective is to partition
the vertices into clusters that minimize the number of ``unsatisfied''
edges: positive edges across clusters and negative edges within
clusters. This framework naturally arises in diverse applications,
including clustering ensembles~\cite{bonchi2013overlapping}, duplicate
detection~
\cite{arasu2009large}, community mining~ \cite{chen2012clustering}, link prediction~\cite{yaroslavtsev2018massively}
disambiguation tasks~\cite{kalashnikov2008web}, image segmentation~\cite{kim2014image}, and automated labeling~\cite{agrawal2009generating, chakrabarti2008graph}.

Despite its widespread applicability, Correlation Clustering is
APX-hard~\cite{CGW05}, motivating a rich line of research focused on
approximation algorithms.  Early work by~\cite{BBC04} provided an
$O(1)$-approximation, which was subsequently improved to a
4-approximation by~\cite{CGW05}. The influential Pivot
algorithm~\cite{ACN08} achieved a 3-approximation, and further
refinements using LP-based techniques culminated in a
2.06-approximation~\cite{CMSY15}, nearly matching the integrality gap
of 2.

Recent work has shown how to surpass this barrier.  Cohen-Addad, Lee and Newman employed the Sherali-Adams hierarchy to obtain a $(1.994 + \epsilon)$-approximation~\cite{CLN22}, which was later improved to 
$(1.73 + \epsilon)$ by Cohen-Addad, Lee, Li and Newman, using a new so-called \emph{preclustering} technique to preprocess the instance~\cite{CLLN23}. Most recently, Cao, Cohen-Addad, Lee, Li, Newman and Vogl introduced the \clusterLP framework which generalizes all previously known formulations, and showed how to round it to obtain a $(\pureclusterlpratio + \epsilon)$-approximation in polynomial time~\cite{cao2024understanding}. \footnote{See the footnote in the abstract for a remark regarding the \pureclusterlpratio\ approximation ratio for \clusterLP.}
However, the \clusterLP (see Section~\ref{sec:results} for details) has exponential size: It contains a variable for each subset of the vertices indicating whether this subset is a cluster or not. The recent work of \cite{cao2024understanding} leverages the preclustering method to compute a solution to the \clusterLP in time $O(n^{\text{poly}(1/\eps)})$ whose value is within a $(1+\epsilon)$-factor of the cost an optimal clustering, hence leading to a polynomial-time approximation algorithm.

\cc is a versatile, fundamental model for clustering and has thus received a lot of attention from practitioners. Therefore, a large body of work studying \cc in popular practical computation models has emerged. Since 2018, researchers have shown how to obtain a
$(3+\eps)$-approximation to \cc in
streaming~\cite{DBLP:conf/soda/BehnezhadCMT23,makarychev2024single},
sublinear time~\cite{DBLP:conf/innovations/Assadi022}, the Massively
Parallel Computation (MPC)
model~\cite{DBLP:conf/focs/BehnezhadCMT22,cao2023breaking,prunpivot},
or vertex or edge fully dynamic
setting~\cite{behnezhad2024fullydynamiccorrelationclustering,prunpivot}.
For all but the dynamic setting, a very recent work by Cohen-Addad,
Lolck, Pilipczuk, Thorup, Yan and Zhang~\cite{cohen2024combinatorial}
has provided a unified approach achieving a
$(1.847+\eps)$-approximation to \cc in all the above models ($1.875$
for the MPC model) using a new local search algorithm.

In general, the approximation guarantees obtained in these various
restricted settings have remained significantly higher than the best
known polynomial time approximation of $(\pureclusterlpratio + \epsilon)$, which
employed computationally expensive solutions to Sherali-Adams
relaxations of \cc to solve the \clusterLP.

Thus, it is natural to consider what we can achieve when we are allowed limited computation time.  For instance, \textbf{How well can Correlation Clustering be approximated in (sub)linear time?}  Here, by linear time, we mean time linear in the number of \pedges, and by sublinear time, we mean time linear in the number of vertices.

In this work, we answer this question by showing how to solve and
round the \clusterLP in sublinear time, matching the state-of-the-art
approximation achieved in~\cite{cao2024understanding}.  This opens up
a new path to study \cc in other computational models.

\paragraph{Notation.}
Before we explain our results in more detail, we introduce some basic
notation.  The input to the correlation clustering problem is a
complete graph where each edge is labeled either as a $+$edge or a
$-$edge.  For a graph \( G \), we often denote its vertex set by \(
V(G) \) and its edge set by \( E(G) \).  We let \( G = (V, E) \)
represent the subgraph induced by the \pedges (i.e., $E^+ = E(G)$),
while the set of \medges is given by \( E^{-}
= \binom{V(G)}{2} \setminus E(G) \). Let \( n = |V| \) and \( m =
|E| \) represent the number of vertices and \pedges in \( G \),
respectively. Given a graph \( G \) and a subset \( V' \subseteq
V(G) \), \( G[V'] \) denotes the subgraph of \( G \) induced by \(
V' \).  For technical reasons, we assume \( E \) contains all \( n \)
self-loops \( \{uu : u \in V\} \).

For two sets \( A \) and \( B \), we use \( A \bigoplus B = (A \setminus B) \cup (B \setminus A) \) to denote their symmetric difference.
For simplicity, we use the following shorthand. Given a function or vector $f$ and a set $S$, if $f$ outputs reals, then $f(S) = \sum_{e \in S} f(e)$ or $f(S) = \sum_{e \in S}f_e$.

Finally, we emphasize that we always use $\opt$ to denote an optimal clustering for a given Correlation Clustering instance, and we use $\cost(\opt)$ to denote its cost, which we will formally define shortly in Section \ref{sec:tech-overview}.

\vspace{2mm}

\subsection{Our Results}
\label{sec:results}

The \textbf{cluster LP} was introduced in~\cite{cao2024understanding}.
In this formulation of \cc, we have a variable \( z_S \) for every
non-empty subset \( S \subseteq V \), where \( z_S \) indicates
whether \( S \) forms a cluster in the output clustering.
Additionally, for every pair of vertices \( uv \in {V \choose 2} \),
the variable \( x_{uv} \) indicates whether \( u \) and \( v \) are
separated in the clustering.

\begin{equation}
	\min \qquad \obj(x) := \sum_{uv \in E^+} x_{uv} + \sum_{uv \in E^-} (1-x_{uv})  \qquad \text{s.t.} \tag{cluster LP} \label{LP:clusterlp}
\end{equation} \vspace{-15pt}
\begin{align*}
	\sum_{S \ni u} z_S &= 1 \quad \forall u \in V \\
	\sum_{S \supseteq \{u, v\}} z_S &= 1 - x_{uv} \quad \forall uv \in {V \choose 2}.
\end{align*}

With this formal definition of the cluster LP, our main results are stated in the following
theorems. 

\begin{restatable}[Efficient~\ref{LP:clusterlp}]{theorem}{thmsolveclusterLPsequential} \label{thm:solving-cluster-LP-sequential}
Let $\eps, \delta > 0$ be small enough constants and let $\opt$ be the cost of the optimum solution to the given \cc instance. Then there is a small $\Delta = \poly(\epsilon)$ such that the following statement holds. One can output a solution $(z_S)_{S \subseteq V}$ to the cluster LP with $\obj(x) \leq (1+\eps)\opt$ in expectation, described using a list of non-zero coordinates such that each coordinate of $z$ is either $0$ or at least $\Delta$.   In the various models, the respective procedure has
the following attributes. 
\begin{itemize}    
    \item \label{itm:sublinearmain}(Sublinear model) The running time to compute $z$ is $\widetilde{O}(2^{\poly (1/\eps )}n)$.
\item \label{itm:mpcmain}(MPC model) It takes $2^{\poly (1/ \eps )}$ rounds with $O(n^{\delta})$ memory per machine and total memory $\tilde O(\poly (\frac{1}{\eps}) m)$, or takes $\poly (\frac{1}{\eps})$ rounds with $O(n^{\delta})$ memory per machine and total memory $\tilde O(2^{\poly (1/ \eps )} m)$.

\end{itemize}
\end{restatable}

We call the non-zero entries of a solution $(z_S)_{S \subseteq V}$ its {\em support} and refer to it as $\supp(z)$.

\paragraph{Rounding Algorithms.}
Given a solution $z$ to the \clusterLP, we can round it to a solution for \cc efficiently as well. 
Concretely, we prove the following theorem.
\begin{restatable}[Efficient Rounding Algorithm]{theorem}{thmroundclusterLPsequential} \label{thm:rounding-cluster-LP-main}
    Let $\eps> 0$ be a small enough constant and $\Delta = \poly(1/\epsilon)$.
    Given a solution to the cluster LP $(z_S)_{S \subseteq V}$ where each coordinate of $z$ is either $0$ or at least $\Delta$, one can output a clustering with expected cost at most $(\pureclusterlpratio+\eps) \obj(x)$ in time $\tilde O(n / \Delta^2 )$.
\end{restatable}

Combining Theorem~\ref{thm:solving-cluster-LP-sequential} and Theorem~\ref{thm:rounding-cluster-LP-main} gives us our main conclusion for \cc.

\begin{coro}
\label{coro:mainalgorithm}
There exists a $(\pureclusterlpratio+\eps)$-approximation algorithm for \cc that runs in time $\Tilde{O}(2^{\poly(1/\eps)}n)$.
\end{coro}

\paragraph{Remark} 
We emphasize that the sublinear-time algorithm outputs a clustering. However, we do not know the number of disagreements for the clustering.

\subsection{Technical Overview}\label{sec:tech-overview}

We first discuss the cluster LP in more depth.

\paragraph{Cluster LP.}
Given a set $S$, let $E^+(S, V \setminus S) = \{ uv \in E^{+} \mid u \in S, v \not\in S \textmd{ or } u \not\in S, v \in S\}$ denote the set of \pedges with exactly one endpoint in $S$, and let $E^-(S) = \{ uv \in E^{-} \mid u, v \in S \}$ denote the set of \medges with both endpoints in $S$. We can rewrite the objective value of cluster LP using $z$ values instead of $x$ values as
\[
\cost(z) := \sum_{S \subseteq V} \costs (S) \cdot z_S, 
\]
where $\costs(S)$ is the cost contribution of the cluster $S$ and 
\begin{equation}
\costs(S) := \frac{1}{2} \cdot |E^+(S, V \setminus S)| + |E^-(S)|. \label{def:cc_cost}
\end{equation}
The~\ref{LP:clusterlp} is equivalent to minimizing $\cost(z)$, subject to the following constraints.
\begin{equation}
	\min \qquad \cost(z)  \qquad \text{s.t.} \tag{cluster LP} \label{LP:clusterlp2}
\end{equation} \vspace{-15pt}
\begin{align*}
	\sum_{S \ni u} z_S &= 1 \quad \forall u \in V\\
	z_S &\geq 0 \quad \forall S \subseteq V, S\neq \emptyset. 
\end{align*}

\paragraph{Difficulty of Solving~\ref{LP:clusterlp}.}

A natural approach to solving a linear program with an exponential number of variables is to solve the dual, which has an exponentially many constraints, but polynomially many variables and could potentially be solved via a polynomial-time separation oracle.  In the case of the \clusterLP, we have the following dual linear program containing a variable for each vertex, which can take on a possibly non-positive value.
\begin{align*}
  \max  & \sum_{u \in V} q_u   ~\tag{dual cluster LP}\label{lp:cluster-dual}\\
  \sum_{u \in S} q_u & \leq \cost(S) \quad \forall S. 
  \end{align*}
 In this case, it is not obvious how to find an efficient separation oracle. However, a subroutine used in the multiplicative weights algorithm (discussed next) does yield an approximate separation oracle.
Previous approaches to finding a solution for the \clusterLP 
with objective value at most $(1+\eps)\cost(\opt)$
involved solving a Sherali-Adams relaxation containing $n^{\poly{(1/\eps)}}$ constraints~\cite{cao2024understanding}. 

\paragraph{Multiplicative Weights Update Framework.}
We will use the multiplicative weights update (MWU) framework by Plotkin, Shmoys, and Tardos \cite{plotkin1995fast} to solve the \clusterLP approximately. Actually, we will solve a covering linear program that agrees with the \clusterLP on its optimal solutions, but is better adapted for the MWU framework. The framework maintains a weight $w_v$ for each vertex constraint $\sum_{S \ni v} z_S \geq 1$. 
During each step $t$, the vertex constraints are collapsed into a single constraint, each scaled according to the normalized vertex weight. The resulting optimization problem is the following.
\begin{equation}
	\min \qquad \cost(z)  \qquad \text{s.t.}
\end{equation} \vspace{-15pt}
\begin{align}
	\sum_{S} p(S) z_S & \geq 1 \\
	z_S &\geq 0 \qquad \forall S \subseteq V, S\neq \emptyset. 
\end{align}
Here, $p(S)$ is the normalized weight of vertices in $S$.
The optimal solution to this problem is the set $S \subseteq V$ that has the smallest cost to vertex weight ratio. However, instead of solving the problem optimally we will find a family of several sets that have a small cost to vertex weight ratio compared to the cost of the optimal clustering. To be more precise, we constructing a partial clustering $\mathcal{F}$ that covers at least a constant fraction $P$ of the vertex weight. We set $z_{C} = \frac{1}{P}$ for each cluster $C \in \mathcal{F}$ as the solution for step $t$. 
Then, each vertex weight is scaled depending on the margin by which $z$ violates the covering constraint of that vertex. We can show that a constant number of rounds of this framework is enough to produce an approximately optimal solution to the \clusterLP.

\paragraph{Finding the Partial Clustering.}
We will find a family $\mathcal{F}$ of disjoint clusters where each cluster $C \in \mathcal{F}$ has a small cost to vertex weight ratio. Moreover, this family will cover a constant fraction of the total vertex weight.
Given a vertex $r$ and a guess $R$ for the cost of the  optimal solution, we can efficiently find a single cluster $C_r \ni r$ with a cost to vertex weight ratio at most $R$ if such a cluster exists.  In particular, if the cluster $C \in \opt$ that contains the vertex $r$ achieves the ratio $R$, then we will find a cluster $C_r$ that achieves a ratio close to $R$.  
There has to exist a cluster $C \in \opt$ that achieves the ratio $R$ since the normalized vertex weight gives a probability distribution. Thus, we can find a cluster $C_r$ for each vertex $r \in V$ and add the cluster with the best ratio (which is at most $R$) to the partial clustering $\mathcal{F}$. We then remove the vertices already covered by $\mathcal{F}$ and update all clusters $C_r$ that contain covered/lost vertices.  We repeat the process until $\mathcal{F}$ covers a constant fraction of the total vertex weight. 

\paragraph{Finding a Small Ratio Cluster.} The problem we aim to solve is as follows: given a weight function \( w \) for each vertex, we seek a cluster \( C^* \) that minimizes the expression \( \cost(C^*) - \sum_{v \in C^*} w(v) \), where \( \cost(C^*) \) represents the \cc cost associated with selecting \( C^* \) as a cluster (see \eqref{def:cc_cost}).  A small ratio cluster will yield a negative value for this expression.

Inspired by the local search algorithm in \cite{cohen2024combinatorial}, we use a similar approach to find such a \( C^* \). Let’s focus on how to find \( C^* \) to minimize this difference. We start by assuming we know \( C^* \), but we cannot directly query whether a vertex \( v \in C^* \) or not. In this case, our key observation is that if \( C^* \) is the cluster that minimizes \( \cost(C^*) - \sum_{v \in C^*} w(v) \), then  
according to the optimality of \( C^* \), we have:
\begin{enumerate}
    \item \( \marginal(C^*, v) - w(v) \leq 0 \) for all \( v \in C^* \), and
    \item \( \marginal(C^*, v) - w(v) \geq 0 \) for all \( v \notin C^* \),
\end{enumerate}
where \( \marginal(C^*, v) = \cost(C^* \cup \{v\}) - \cost(C^* \setminus \{v\}) \) is the marginal cost of adding \( v \) to \( C^* \) or removing \( v \) from \( C^* \). 
At first glance, it seems we don’t make progress, as we don’t know how to compute \( \marginal \). The key idea here is that if we can sample some vertices from \( C^* \), we can obtain a good estimate for \( \marginal \). However, if \( C^* \) has no structure, we have no way to sample nodes from \( C^* \).

Thanks to the preclustering step \cite{CLLN23}, which roughly identifies almost-cliques that should be clustered together and prohibits pairs that are impossible to place in the same cluster, we can make a good guess about the size of \( C^* \). After preclustering, we know a candidate set \( \Ncand \), and \( C^* \) is a subset of \( \Ncand \) with \( |C^*| = \Theta(|\Ncand|) \). Now we can sample nodes from \( C^* \): if we uniformly sample \( \Theta(1) \) nodes from \( \Ncand \), each sample hits \( C^* \) with constant probability, and our sample set has size \( \Theta(1) \), hitting \( \Theta(1) \) nodes from \( C^* \). By enumerating all possible subsets of the sampled set, we obtain a sample for \( C^* \), and by a standard concentration bound, we have a good estimation for $\marginal(C^*, v)$.

One last problem remains. Since we no longer obtain \( \marginal(C^*, v) \) exactly, we rely on an estimated \( \marginal(C^*, v) \), meaning we may make some errors in selecting \( C^* \): we might miss some vertices in \( C^* \) and add others incorrectly. If we make too many such errors, our final set will differ significantly from \( C^* \). The key observation is that as long as \( C^* \) and \( \tilde{C}^* \) (the approximated cluster) do not differ too much, our estimate for \( C^* \) remains reasonable. To leverage this, we split \( \Ncand \) into many small chunks. The estimation error in each round is bounded by the size of each chunk. After processing one chunk, we update our sample vertices to match the nodes we already found. One might argue that we no longer obtain \( C^* \) exactly, which is true due to errors in each chunk. To bound the error affecting \( \cost(C^*) - \sum_{v \in C^*} w(v) \), we will use new sample sets based on previous choices. As long as there is some slackness in \( \cost(C^*) - \sum_{v \in C^*} w(v) \), we can still obtain a final cluster with a reasonably small ratio.

\paragraph{Achieving Nearly Linear Time.}
We have to address two problems in order to find the partial clustering $\mathcal{F}$ in nearly linear time. First, we cannot afford to compute all the clusters $C_r$, one for each vertex $r \in V$. Second, we cannot afford to update each cluster $C_r$ as soon as it loses one vertex $v \in C_r$. To solve the first problem, we will sample a subset of vertices $U \subseteq V$ such that $U$ is not too large and we can compute a cluster $C_r$ for each vertex $r \in U$. Moreover, with high probability, we will hit every cluster $C \in \opt$ with a vertex, $U \cap C \neq \emptyset$. Because of this, there will still be a cluster among $\{C_r\}_{r \in U}$ that achieves the ratio $R$. 
To address the second problem, we will update a cluster $C_r$ only if it lost a constant fraction of its weight. If a cluster $C_r$ did not lose this constant fraction, the ratio did not change by too much. We show that a cluster $C_r$ has to lose at least a constant fraction of vertices in order to lose a constant fraction of weight. We can charge the cost of updating the cluster $C_r$ to the number of vertices the cluster $C_r$ loses.

\paragraph{Achieving Sublinear Time.}
To achieve a sublinear-time algorithm, we must accelerate the process of finding a small-ratio cluster. The bottleneck is the estimation of the cost for a given cluster. In general, this estimation is impossible in the sublinear-time model. Our key insight is that if the cost is small, the cluster should have been split at the beginning. Thus, in the remaining graph, only clusters with large costs persist, allowing us to estimate the cost efficiently.

\paragraph{MPC Algorithm.}
To parallelize the algorithm, the main bottleneck is that we need to find a set of small-ratio clusters instead of identifying them one by one. We can select multiple nodes and run the small-ratio cluster-finding algorithm simultaneously. However, their candidate sets may overlap, so we remove nodes that appear in multiple candidate sets. As long as we use a sufficiently small constant probability to select nodes, we ensure a large enough candidate set, allowing each round to produce a sufficiently large set of small-ratio clusters.

%% file: preliminaries.tex
We first introduce some more necessary notation.
Define \( d^{+}(v) \) as the degree of vertex \( v \) with respect to $+$edges in \( G \). For any subset \( C \subseteq V \), let \( d^+(v, C) \) and \( d^{-}(v, C) \) denote the number of $+$edges and $-$edges from $v$ to $C$, respectively, connecting \( v \) to vertices in \( C \). When the context is clear, we omit the \( + \) symbol, so \( d(v) \) and \( d(v, C) \) mean \( d^{+}(v) \) and \( d^{+}(v, C) \), respectively.


\subsection{Computational Models.}
We consider two computational models in this paper: the sublinear model and the MPC model.

\paragraph{The Sublinear model.}
In the sublinear model, the algorithm can query the following information in \( O(1) \) time:
\begin{itemize}
    \item \textit{Degree queries}: What is the degree of vertex \( v \)?
    \item \textit{Neighbor queries}: What is the \( i \)-th neighbor of \( v \in V \) for \( i \leq d^+(v) \)?
    \item \textit{Edge queries}: Is $uv \in E^+$?
\end{itemize}
One can think of this model as storing an adjacency list and an adjacency matrix of \pedges for \( G \).  The matrix allows us to check whether \( uv \in E^+ \) in \( O(1) \) time.

\paragraph{The MPC model.} In the MPC model, the set of edges is distributed across a set of machines, and computation proceeds in synchronous rounds. During each round, each machine first receives messages from other machines, then performs computations based on this information and its own allocated memory, and finally sends messages to other machines to be received at the start of the next round. Each machine has limited local memory, restricting the total number of messages it can receive or send in a round. 
The efficiency of the algorithm is measured by the number of rounds,
the memory used by each machine, the total memory used by all
machines.

In this paper, we consider the MPC model in the \emph{strictly sublinear regime}: Each machine has $O(n^\delta)$ local memory, where $n$ is the input size and $\delta > 0$ is a constant that can be made arbitrarily small. Under this model, we assume the input received by each machine has size $O(n^\delta)$. 

\subsection{Preclustering}
Our results crucially depend on the usage of the 
following {\em preclustering} subroutine.  Intuitively, preclustering is a preprocessing step that will identify almost-cliques that should be together and prohibit all pairs that are impossible to place in the same cluster. There will be some remaining pairs for which we cannot decide whether or not they should be clustered together or split apart. A major advantage of this procedure is that the number of uncertain pairs can be bounded by the optimal solution. More precisely, we use the definition of a preclustered instance from \cite{CLLN23}, which is also applied in \cite{cohen2024combinatorial}.

\begin{definition}[Preclustering]\label{definition:prepro}
Given a \cc instance $(V, E^+ \uplus E^-)$, a {\em preclustered
instance} is defined by a pair $(\calK, E_{\adm})$.  Here, $\calK$ is
a family of disjoint subsets of V (not necessarily a partition) Each
set $K \in \calK$ has $|K| \geq 2$ and is called a {\em
(non-singleton) atom}. We use $V_{\calK} := \bigcup_{K \in \calK} K$ to denote the set of all vertices in non-singleton atoms.  A vertex in $V \setminus \calK$ is called a {\em singleton atom}.

$\eadm \subseteq \binom{V}{2}$ is the set of {\em admissible} edges (sometimes called {\em admissible pairs}), which are defined to be pairs of vertices $(u,v) \in \eadm$ with at least one of $u$ and $v$ not in $V_{\calK}$. A pair of vertices in a same $K \in \calK$ is called an {\em atomic} edge. A pair that is neither an atomic nor an admissible edge is called a \emph{non-admissible} edge.
\end{definition}

For a vertex \( u \), let \( \dadm(u) \) denote the number of vertices \( v \in V \) such that \( (u, v) \) is an admissible pair. Note that admissible, atomic, and non-admissible edges can each be either \pedges or \medges. 
Let \( \Nadm(v) \) to be the set of vertices \( u \) such that \( (u, v) \in \eadm \). We use \( K(v) \) to represent the set \( K \in \calK \) that contains \( v \); to make the notation more general, when \( v \) is a singleton atom, \( K(v) = \{ v \} \).

We also need the definition of \epssimilar preclustered instance and \epssimilar cluster to bound the value of $\dadm(v)$ for each node $v \in V$. 

\begin{definition}[\epssimilar Preclustering]
\label{def:epsgoodpreclustering}
Given a \cc instance $(V, E^+ \uplus E^-)$, let $(\calK, \eadm)$ be a {\em preclustered instance} for $G$, and let $\eps > 0$ be some parameter. We say $(\calK, \eadm)$ is an {\em \epssimilar preclustering} if 
\begin{itemize}
    \item for any $v \in V$, we have $ \dadm(v) \leq 2\eps^{-3} d(v)$,
    \item for any $uv \in E_{\adm}$, we have $d(u) \leq 2\eps^{-1} d(v)$ and the number of common neighbors that are degree similar to $u$ and $v$ is at least $\epsilon \cdot \min\{d(u), d(v)\}$, where a pair $w\Tilde{w}$ is degree similar if $\eps d(w) \leq d(\Tilde{w}) \leq d(w) / \eps$,
    \item  for every atom $K \in \calK$, for every vertex $v \in K$, we have that $v$ is adjacent to at least a $(1 - O(\eps))$-fraction of the vertices in $K$ and has at most $O(\eps |K|)$ neighbors not in $K$.
\end{itemize}
\end{definition}
\begin{definition}[\epslarge Cluster]
Given a preclustered instance $(\calK, \eadm)$ for some \cc instance $(V, E^+ \uplus E^-)$,  a set $C \subset V$ is called {\em \epslarge} with respect to $(\calK, \eadm)$ if 
	\begin{itemize}
	    \item $C$ does not break any atomic edge,
	    \item $C$ does not contain any non-admissible edge, and
        \item for any vertex $v$ in $C$, if $|C| > 1$, then $|C| \geq \eps d(v)$.
	\end{itemize}

Moreover, a clustering scheme $\calC$ is called \epslarge with respect to $(\calK, E_\adm)$ if all clusters
$C \in \calC$ are \epslarge clusters with respect to $(\calK, \eadm)$.
\end{definition}

The next theorem is stated in \cite{cohen2024combinatorial}. 
The construction follows the framework of Assadi and Wang \cite{DBLP:conf/innovations/Assadi022, CLMNP21, CLLN23}.

\begin{theorem}[Preclustering Procedures \cite{cohen2024combinatorial}]
\label{thm:preclustering-proc}
For a \cc instance $(V, E^+ \uplus E^-)$ with optimal value $\optc$ (which is not known to us), and for any sufficiently small $\eps > 0$, there are algorithms that produce an \epssimilar preclustered instance $(\calK, E_{\adm})$ that admits an \epslarge clustering scheme $\calC^{*}_{(\calK, E_\adm)}$ such that
\begin{enumerate}
    \item $\cost(\calC^{*}_{(\calK, E_\adm)}) \leq (1 + \eps) \optc$,
    \item and $|E_{\adm}| \leq O\big(\eps^{-12} \optc \big)$.
\end{enumerate}
n the various models, the respective procedure has the following attributes.
\begin{itemize}
    \item (MPC model) The algorithm takes $O(1)$ rounds succeeds with probability at least $1 - 1/n$ and requires $O(n^{\delta})$ memory per machine. Moreover, the algorithm uses a total memory of $O(m \log n)$, where $m$ is the number of the $+$ edges.
    \item (Sublinear model)In time $O(n \log{n})$ one can compute the partition $\calK$ and a data structure such that with success probability at least $1-1/n^2$, for any pair of vertices, the data structure can answer in $O(\log n)$ time whether the pair is in $E_{\adm}$ or not. Moreover, the data structure can list all vertices admissible to $v$ in $O(d(v) \log^2 n)$ time.
\end{itemize}

\end{theorem}

\paragraph{Assumptions.} We will use \( \opt \) to denote an optimal clustering for the respective \cc input instance. We assume that \( \opt \) is an \epslarge clustering and satisfies \( |E_{\adm}| \leq O\left(\epsilon^{-12} \clcost(\opt)\right) \) because of Theorem~\ref{thm:preclustering-proc}.

%% file: mwuframework.tex
In this section, we prove the following: there exists an MWU algorithm that, in a constant number of rounds, finds a solution to the \clusterLP that is feasible and has a cost of at most \( (1 + O(\epsilon)) \optc \).

The framework is provided in Algorithm~\ref{alg:solveclusterlpframework}. To solve~\ref{LP:clusterlp2}, we first run the preclustering subroutine, which provides some structure to the graph. Then, for each non-singleton atom of the preclustering, if the cost of isolating this atom 
is zero, it means that this atom is a clique with no outgoing edges, and we will simply make it a cluster.
Next, we preprocess~\ref{LP:clusterlp2} by adding a fixed value to the objective and reformulating it as a covering LP, so that our final LP has certain desirable properties. We will discuss the details of~\ref{LP:coverclusterlp} (Line 3) and how to convert our solution back (Line 5) in Subsection~\ref{sec:covertingclptocclp}. We use the MWU algorithm to solve the \coverClusterLP. The description of the MWU algorithm (Lines 3-4) is given in Subsection~\ref{sec:mwualgorithm}.

\subsection{Converting \clusterlp to \coverClusterLP}
\label{sec:covertingclptocclp}
In order to run a multiplicative weights algorithm to solve the cluster LP, we first transform it into a covering LP. We make two key changes to the~\ref{LP:clusterlp2}. 
First, for each set $S \subseteq V$, we will add $\sum_{v \in S} \dc(v)$ to the cost function.
\[
\covers(S) := \costs(S) + \dc(S),
\]
where \( \dc(v) \) is twice the cost attributable to vertex $v$ if we choose to make \( K(v) \) a cluster.  If $v$ is a singleton atom, then $\dc(v) = d(v) - 1$ (because of the self-loop $vv$). If $v$ is an atom, then
\begin{align}
    \dc(v) = \frac{2 \cdot \cost( K(v))}{|K(v)|} = \frac{1}{|K(v)|} \sum_{u \in K(v)} \left( d(u) + |K(v)| - 2 d(u, K(v)) \right ).
\end{align}
The new objective function will be  
\[
\covers(z) := \sum_{S \subseteq V} \covers(S) \cdot z_S.
\]
Secondly, instead of an equality constraint, we relax the constraint to an inequality. Using the new objective function, we can rewrite the~\ref{LP:clusterlp2} as a~\ref{LP:coverclusterlp}.

\begin{equation}
	\min \qquad \covers(z) \qquad \text{s.t.} \tag{covering cluster LP} \label{LP:coverclusterlp}
\end{equation} \vspace{-15pt}
\begin{align}
	\sum_{S \ni u} z_S &\geq 1   &\quad &\forall u \in V \label{LPC:set-covered}\\
	z_S &\geq 0, &\qquad &\forall S \subseteq V S\neq \emptyset.
\end{align}

\begin{algorithm}[t]
	\caption{Solving~\ref{LP:clusterlp2}}
	\label{alg:solveclusterlpframework}
	\begin{algorithmic}[1]
 \STATE \textbf{Input:} Graph $G$ 
\STATE \textbf{Output:} a feasible solution to~\ref{LP:clusterlp} satisfying Theorem~\ref{thm:solving-cluster-LP-sequential}
        \STATE Find a preclustering $G$ via Theorem~\ref{thm:preclustering-proc}; let $(\calK, \eadm)$ be the preclustered instance.
        \STATE $V' \gets V$
        \FOR{ all atoms $K \in \calK$ such that if $\dc(K) = 0$}
        \STATE make $K$ a cluster and $V' \gets V' \setminus K$
        \ENDFOR
        \STATE Solve \ref{LP:coverclusterlp} using MWU Algorithm~\ref{alg:mw} on $G[V']$ to within $(1 + O(\eps^{13}))$ of the optimum solution.
        \STATE Let $\{z_S \}_{S}$ be the solution. 
        \STATE Compute a solution $\{\tilde{z}_S \}_{S}$ to the \ref{LP:clusterlp} using Lemma~\ref{lem:cover-to-cluster}. 
        \RETURN $\{\tilde{z}_S \}_{S}$.
	\end{algorithmic}
\end{algorithm}

Note that the cross degrees \( \dc(v) \) can be computed in time \( \widetilde{O}(m) \) by accessing each edge and checking if it belongs to the same atom. 
This is sufficient if we aim for nearly linear time.

One may ask whether this modification changes the solution in terms of the approximation algorithm. Let \( \Ec \) be the set of disagreement edges if we isolate all singleton and non-singleton atoms as the final clusters (i.e., \( 2|\Ec| = \sum_{v \in V} \dc(v) \)).  
If \( \Ec \) is unbounded, we are unable to obtain any guarantee for~\ref{LP:clusterlp2} when solving~\ref{LP:coverclusterlp} approximately. Fortunately, we have the following lemma to help us bound the increase in the objective when adding \( \dc \) to the cost.

\begin{lemma} 
\label{lem:dc}
For any non-singleton atom $K \in \calK$,
\[
\dc(K) = 2 \cdot \cost(K) = \Omega( \eps^{3} \dadm(K) ).
\]
Moreover, 
\[
    \dc(V) = \sum_{v \in V} \dc(v) = O\left(\frac{1}{\eps^{12}}\right)\clcost(\opt).
\]
\end{lemma}
\begin{proof}

We first show the upper bound on $\sum_{v \in V} \dc(v)$. Each edge counted by $\sum_{v \in V} \dc(v)$ is either admissible or contributes to the cost of $\opt$. To finish the upper bound, remember that $|E_{\adm}| = O(\epsilon^{-12}\clcost(\opt))$ by \cref{thm:preclustering-proc}.

Then, we will prove the lower bound on \( \dc(K) \).
Consider a vertex \( v \in K \) and an admissible edge \( (v, u) \). If $uv$ is a \pedge, then \( \dc(K) \) already covers it, so we only need to consider \medges.

Let $A(v)$ be the set of all vertices that are connected to $v$ by an admissible \medge.
Let \( u \) be one such neighbor in \( A(v) \). By Definition~\ref{def:epsgoodpreclustering}, \( v \) and \( u \) are degree similar and must share at least \( \epsilon \cdot \min\{d(v), d(u)\} \geq \epsilon^2 d(v) \) \( + \)neighbors, which are degree similar to both \( u \) and \( v \). Recall that a pair \( w \Tilde{w} \) is degree similar if \( \epsilon d(w) \leq d(\Tilde{w}) \leq d(w) / \epsilon \). We will distinguish between two cases.

Let $A_1(v) \subseteq A(v)$ be the set of vertices $u \in A(v)$ such that
at least half of the  degree-similar \( + \)neighbors of \( u \) and \( v \) are not in \( K \). Note that  \( v \in K \),  has at most \( d(v) - d(v, K)\) degree similar \( + \)neighbors outside \( K \). Each \( u\in A_1(v) \) has to be adjacent to at least \( \frac{\epsilon^2}{2} d(v) \) of these vertices outside of \( K \) which have degree at most $d(v)/\epsilon$. Therefore, we have at most
\begin{align*}
    \frac{(d(v) - d(v, K)) \cdot d(v) / \epsilon}{\epsilon^2 d(v) / 2} \leq 2\epsilon^{-3} (d(v) - d(v, K))
\end{align*}
vertices in $A_1(v)$. Thus, we can bound $\sum_{v \in K} |A_1(v)|$ by \( 2\epsilon^{-3} \sum_{v \in K} (d(v) - d(v,K) ) \leq 4\epsilon^{-3} \cost(K) \leq 4\epsilon^{-3} \dc(K) \).

Let $A_2(v) \subseteq A(v)$ be the set of vertices $u \in A(v)$ such that that at least half of the degree-similar \( + \)neighbors of \( u \) and \( v \) are in \( K \). Each \( u \in A_2(v) \) must connect to at least \( \frac{\epsilon^2}{2} d(v) \) vertices in \( K \), so we have at most
\begin{align*}
    \frac{\dc(K)}{\epsilon^2 d(v) / 2} \leq \frac{4\dc(K)}{\epsilon^2 |K|}
\end{align*}
vertices that are in $A_2(v)$. Thus, we can bound $\sum_{v \in K}
|A_1(v)|$ by \( |K| \cdot \frac{4\dc(K)}{\epsilon^2 |K|} =
4\epsilon^{-2} \dc(K) \).

Combining the two cases, we find that the \( - \)admissible neighbors of \( K \) are at most \( \sum_{v \in K} |A(v)| = O(\epsilon^{-3} \dc(K)) \).
\end{proof}

We say that a set of vertices \( S \subseteq V \) does not split atoms if \( K(v) \subseteq S \) for all \( v \in S \). Adding \( \dc \) is useful for us because the new objective function \( \covers(\cdot) \) remains monotone as long as the involved sets do not split atoms.

\begin{lemma}
    \label{lem:monotone}
    Let $U,W \subseteq V$ such that neither $U$ nor $W$ splits an atom. If $U \subset W$ then, $\covers(U) < \covers(W)$. 
\end{lemma}
\begin{proof}
We need to show:
\begin{equation}
\cost(U) + \dc(U) < \cost(W) + \dc(W) = \cost(W) + \dc(U) + \dc(W\setminus{U}).
\end{equation}
So we want to show:
\begin{equation}
\cost(U) < \cost(W) + \dc(W\setminus{U}),
\end{equation}
which is true since the \medges in $U$ also belong to $W$ and \pedges contributing to $\cost(U)$ but not $\cost(W)$ are at most those in $E^+(U,W\setminus{U})$, whose contribution to $\cost(U)$ is $|E^+(U,W\setminus{U})|/2 < \dc(W\setminus{U})$.  
\end{proof}

Using the above observation, we can show that an optimal solution of the \ref{LP:coverclusterlp} is feasible for the \ref{LP:clusterlp} as long as it does not split atoms.
\begin{lemma}
\label{lem:cover-cluster}
    Let $z$ be an optimal solution to \ref{LP:coverclusterlp} where each set $S$ in the support of $z$ does not split atoms. Then $z$ satisfies all constraints with equality.
\end{lemma}
\begin{proof}
    Assume for contradiction that there exists an atom $K(v)$ that is not tight, $\sum_{S \supseteq K(v)} z_S > 1$. Let $W \subseteq V$ be such that $z_W > 0$.
    Let $U = W \setminus K(v)$.  We modify $z$ by decreasing $z_W$ and increasing $z_U$. 
    Let $a = (\sum_{S: K(v) \subseteq S} z_S) - 1$, and let $b = \min\{z_W, a\}$.  We define a new vector $\widetilde z$.  Let $\widetilde z_U = z_U + b$ and let $\widetilde z_W = z_W - b$.  For all other $S$ such that $S \neq W,U$, let $\widetilde z_S = z_S$.
   
    We obtain a contradiction by showing that $\widetilde z$ remains feasible but has a decreased objective value. The only constraints affected by the change are those corresponding to vertices in $W$. 
    For a vertex $u \in U= W\setminus K(v) $,
    $$\sum _{S \ni u} \widetilde z_S = \widetilde z_W + \widetilde z_U + \sum_{\substack{S \ni u \\ S \neq W, U}} z_S = 
        (z_W - b) + (b + z_U) + \sum _{\substack{S \ni u \\ S\neq W,U}} z_S = \sum _{S \ni u} z_S \geq 1.$$ 
    For the atom $K(v)$, we have
    $$\sum _{S \supseteq K(v)} \widetilde z_S = \widetilde z_W + \sum _{\substack{S \supseteq K(v)\\ S\neq W}} z_S  = z_W - b + \sum _{\substack{S \ni v\\ S\neq W}} z_S
    \geq 1.$$ The last inequality holds since $b \leq a$.
    To finish the proof, we have to show that the objective value has decreased.  In other words, we want to show the following quantity is positive.
     $$\covers(W) \cdot (z_W - \widetilde z_W) + \covers(U) \cdot (z_{U} - \widetilde z_{U}) = 
    \covers(W) \cdot b + \covers(U) \cdot (-b).$$
    To finish the proof, remember that $\covers(\cdot)$ is monotone and $U \subseteq W$ does not split atoms. The lemma follows from Lemma \ref{lem:monotone}.\end{proof}
We can show a more general version of Lemma \ref{lem:cover-cluster}. In particular, we can transform a suitable solution of the~\ref{LP:coverclusterlp}
into a solution to the~\ref{LP:clusterlp2} that satisfies the additional condition in Theorem \ref{thm:solving-cluster-LP-sequential}, which stipulates that all values in the support are lower bounded by a small constant.
\begin{lemma}
\label{lem:cover-to-cluster}
Let $\epsilon > 0$ be a small enough constant, $\covereps = O(\eps^{13})$
    and $c = \lceil \frac{1}{\covereps} \rceil$. Given a constant $\tmwu \in \mathbb{N}$ and a solution $z$ to the~\ref{LP:coverclusterlp} where,
    \begin{itemize}
        \item $z_S \geq \frac{1}{\tmwu}$ for each $S \in \supp(z)$,
        \item all $S \in \supp(z)$ do not split atoms,
        \item for each vertex $v$ the number of sets $S \in \supp(z)$ with $v \in S$ is at most $\tmwu$.
     \end{itemize}
    We can find a solution $\widetilde z$ to the~\ref{LP:clusterlp2} where $\widetilde z_S \geq \frac{1}{c \tmwu}$ for all $S \in \supp(z)$ and $\covers(\widetilde z) \leq (1 + \covereps) \covers(z)$ and the solution $\widetilde z$ can be found in time $O(n)$.
\end{lemma}
\begin{proof}
    If $z$ is a solution to the \ref{LP:coverclusterlp}, we can assume that $z_S \leq 1$ for all $S \in V$. Otherwise, we can set $z_S = 1$ only decreasing the objective.
    We start by rounding each coordinate $z_S$ to a multiple of $\frac{1}{c \tmwu}$. In particular, set $z^{(1)}_S$ to $\frac{k}{c \tmwu}$ where $k \in \mathbb{N}$ such that $\frac{k-1}{c \tmwu} \leq z_S \leq \frac{k}{c \tmwu}$.
    Since we do not decrease the value of any coordinate, $z^{(1)}$ remains feasible. Furthermore, the objective increases by at most a $\covereps$ factor. Indeed, for any set $S\subseteq V$,
    \[
        z^{(1)}_S - z_S \leq \frac{1}{c \tmwu} \leq \frac{\covereps}{\tmwu} \leq \covereps z_S.
    \]
    The second inequality holds since $c \geq \frac{1}{\covereps}$ and the last inequality follows from $z_S \geq \frac{1}{\tmwu}$. 
    Next, scale the variable $z^{(1)}$ and the constraints by $c \tmwu$. Note that by the above, we have that $z^{(1)}_S = \frac{k^{(1)}_S}{c \tmwu}$ for some $k^{(1)}_S \in \mathbb{N}$. After scaling we have variables $\{k^{(1)}_S\}_{S \subseteq V}$ and constraints, $\sum_{S \ni v} k^{(1)}_S \geq c \tmwu$.
    We can transform $k^{(1)}$ to a scaled up solution of the cluster LP similar to the proof of Lemma \ref{lem:cover-cluster}.  Pick an atom $K(v)$ that is not tight, $\sum_{S \supseteq K(v)} k^{(1)}_S > c \tmwu$. Let $W$ be a set with $k^{(1)}_W > 0$ and $K(v) \subseteq W$. Let $U = W \setminus \{K(v)\}$.  We modify $k^{(1)}$ by decreasing $k^{(1)}_W$ and increasing $k^{(1)}_U$. 
    Let $a = (\sum_{S: K(v) \subseteq S} k^{(1)}_S) - c \tmwu$, and let $b = \min\{k^{(1)}_W, a\}$.  We define a new vector $k^{(2)}$.  Let $k^{(2)}_U = b + k^{(1)}_U$ and let $k^{(2)}_W = k^{(1)}_W - b$.  For all other $S$ such that $S \neq W,U$, let $k^{(2)}_S = k^{(1)}_S$. We will show that $k^{(2)}$ remains feasible while the objective value decreases. The only constraints affected by the change are those corresponding to sets containing vertices in $W$.
    For a vertex $u \in U= W\setminus \{K(v)\}$,
    $$\sum _{S \ni u} k^{(2)}_S = k^{(2)}_W + k^{(2)}_U + \sum_{\substack{S \ni u \\ S \neq W, U}} k^{(2)}_S = 
        (k^{(1)}_W - b) + (b + k^{(1)}_U) + \sum _{\substack{S \ni u \\ S\neq W,U}} k^{(1)}_S = \sum _{S \ni u} k^{(1)}_S \geq c \tmwu.$$
    For the atom $K(v)$, we have
    $$\sum _{S \supseteq K(v)} k^{(2)}_S = k^{(2)}_W + \sum _{\substack{S \supseteq K(v)\\ S\neq W}} k^{(2)}_S  = k^{(1)}_W - b + \sum _{\substack{S \ni K(v)\\ S\neq W}} k^{(1)}_S
    \geq c \tmwu.$$ The last inequality holds since $b \leq a$.
    To finish the proof, we have to show that the objective value has decreased.  In other words, we want to show the following quantity is positive.
     $$\covers(W) \cdot (k^{(1)}_W - k^{(2)}_W) + \covers(U) \cdot (k^{(1)}_{U} - k^{(2)}_{U}) = \covers(W) \cdot b + \covers(U) \cdot (-b).$$
    Observe that $U \subset W$ does not split atoms. Thus, the objective decreases by Lemma \ref{lem:monotone}.
    We can repeat this process until $z^{(2)} = k^{(2)}/(c \tmwu)$ is feasible for the \clusterLP. This process terminates after $O(n)$ iterations, which follows from the fact that in each iteration $\sum_{v \in V} \sum_{S \ni v} k^{(2)}_S$ decreases by at least one. 
  To see this, observe that
    $$ \sum_{v \in V} \sum_{S \ni v} k^{(2)}_S = \sum_{S} \sum_{v \in S} k^{(2)}_S = \sum_{S} |S| ~k^{(2)}_S.$$
This last sum decreases by at least one, since we have ``shifted'' weight from $W$ to a smaller set $U$.  Moreover, by the third property of $z$ and the assumption that $z_S \leq 1$ for all $S \subseteq V$, we have that $\sum_{S \ni v} k^{(2)}_S \leq c \tmwu^2 = O(1)$.
\end{proof}

Throughout the following sections, we will fix the parameter \( \covereps = O(\epsilon^{13}) \) since we want to obtain a \( (1 + \covereps) \)-approximate algorithm for~\ref{LP:coverclusterlp}. Assume \( z \) is a \( (1 + \covereps) \)-approximately feasible solution for~\ref{LP:coverclusterlp}, let $\tilde{z}$ be the solution we apply Lemma~\ref{lem:cover-to-cluster} to obtain a feasible solution for~\ref{LP:clusterlp2}.
Observe that we have $\covers(\widetilde z) = \cost(\widetilde z) + \dc(V)$.
Thus,
\begin{align*}
    \cost(\widetilde z) &= \covers(\widetilde z) - \dc(V) \leq \covers(z) - \dc(V) \\
    &\leq (1 + \covereps) \covers(\opt) - \dc(V) \\
    & \leq (1 + \covereps)\cost(\opt) + (1 + \covereps) \dc(V) - \dc(V) \\
    & \leq (1 + 2\epsilon)\cost(\opt). 
\end{align*}

The last inequality holds because \( \dc(V) = O(\epsilon^{-12} \cost(\opt)) \). In the following section, we present the MWU algorithm to solve~\ref{LP:coverclusterlp} within a \( (1 + O(\covereps)) \)-approximate ratio.

\subsection{The MWU Algorithm}
\label{sec:mwualgorithm}
\input{mwualgorithm}

%% file: mwualgorithm.tex
In this section, we show that a constant number of rounds of the MWU Algorithm suffices to compute a $(1 + O(\covereps))$-approximate solution to the~\ref{LP:coverclusterlp}. The MWU algorithm is given in Algorithm~\ref{alg:mw}.

\paragraph{Overview of Algorithm~\ref{alg:mw}.} Algorithm~\ref{alg:mw} solves the~\ref{LP:coverclusterlp} within a $(1+\covereps)$ factor of the optimum. Algorithm~\ref{alg:mw} maintains a weight $w_v$ for each vertex.  During each step $t=1,\dots \tmwu$, we scale the vertex constraints by the normalized vertex weights and aggregate the scaled constraints into a single constraint. We can find a solution to the now simplified optimization problem by Lemma \ref{lemma:update-point}, which we will prove in Section \ref{sec:family-poly}. 
The weights are then updated depending on the margin by which we violate or satisfy the corresponding vertex constraint. The final solution for the covering \clusterLP is obtained by taking the average of the solutions to the simplified optimization problems solved at each round.  In order to ensure feasibility, we need to increase the values of sets corresponding to atoms containing insufficiently covered vertices.  Finally, when all elements are almost covered (to an extent $1-2\gamma$), we can scale up the value by a small amount to ensure sufficient coverage of all vertices.  The key step in the analysis is showing that there were few insufficiently covered vertices and we do not need to increase the values of too many sets corresponding to their atoms.  
   
\begin{algorithm}[t]
	\caption{MWU algorithm for the \coverClusterLP}
	\label{alg:mw}
	\begin{algorithmic}[1]
    \STATE Initialize the weights $w_v^{(1)} = \dc(v)$ for each vertex $v \in V$.    \FOR{$t=1,\dots,\tmwu$}
        \STATE Normalize the weights $p^{(t)} = \frac{w^{(t)}}{\sum_v w_v^{(t)}}$. 
        \STATE Aggregate all constraints into single constraint: $\sum_{S} p^{(t)}(S) \cdot z_S = \sum_{v} p^{(t)}_v \left (\sum_{S: v \in S} z_S \right) \geq 1$. 
        \STATE Find the point $z^{(t)} \in \left[0, 1/{\covereps}\right]^{2^V}$ from Lemma \ref{lemma:update-point} that, \label{line:findPoint}
        \begin{itemize}
            \setlength\itemsep{0.0em}
            \item satisfies the single constraint $\sum_{S} \pt(S) \cdot z^{(t)}_S \geq 1$,
            \item has objective value $\covers(z^{(t)}) \leq (1+5\covereps)~\covers(\opt)$,
            \item does not split atoms (i.e., if $\zt_S > 0$ then $K(v) \subseteq S$ for all vertices $v \in S$), and
            \item has disjoint support (i.e., if $\zt_S, \zt_{T} > 0$ for two distinct $S,T \subseteq V$, then $S \cap T = \emptyset$).
        \end{itemize}
        \STATE The cost of a constraint corresponds to the margin by which it is satisfied or violated, $m_v^{(t)} = \sum_{S: v \in S} \zt_{S} - 1$.
        \STATE Update the weights $w^{(t+1)}_v =  w^{(t)}_v e^{-\covereps^3 m^{(t)}_v}$.\label{line:update-rule}
    \ENDFOR
    \STATE Let $\hat z$ be the average $\frac{1}{\tmwu} \sum_{t=1}^{\tmwu} \zt$.
    \FOR{each $v$ with $\sum_{S \supseteq K(v)} \hat z_S \leq 1 - 2\covereps$} \label{line:loop_feasible}
    \STATE Set the atom entry $\hat z_{K(v)}$ to 1.
    \ENDFOR
       \STATE $\zmwu \gets \frac{\hat z}{1-2\covereps}$.
 \RETURN $\zmwu$.
	\end{algorithmic}
\end{algorithm}
\begin{lemma}
\label{lem:mw}
    Let $\eps > 0$ be a small enough constant and $\covereps = O(\eps^{13})$. After $\tmwu=\frac{\log(1/\covereps)}{\covereps^4}$ rounds, Algorithm \ref{alg:mw} returns a solution $z$ to the \coverClusterLP with objective $\covers(z) \leq (1+O(\covereps))~\covers(\opt)$. 
\end{lemma}
\begin{proof}
    First, we will prove that $z$ is feasible for the~\ref{LP:coverclusterlp}. Note that at the end of the {\bf for} loop on Line \ref{line:loop_feasible}, we have that $\sum_{S \ni v} \hat z_S \geq 1-2\covereps$ for each vertex $v$. Feasibility follows since we scale by $1/(1-2\covereps)$. 
    Next, we will prove the bound on the objective of $z$.
    Consider the potential
    \[
    \Phi^{(t)} = \sum_{v \in V} w_v^{(t)}.
    \]
    The starting potential is 
    \[ 
    \Phi^{(1)} = \sum_{v \in V} \dc(v) = \dc(V).
    \]
    Because of the way we update the weights, the potential at the end of the execution is
    \[
    \Phi^{(\tmwu+1)} = \sum_{v \in V} w^{(1)}_v \cdot \exp\left(-\covereps^3 \sum_{t \leq \tmwu} m^{(t)}_v\right).
    \]
    \begin{claim} \cite{AHK}
    \label{clm:potential}
    We can relate the potential at the start of the execution and at the end as follows.
        \[
        \Phi^{(\tmwu +1)} \leq \Phi^{(1)} \cdot \exp \left (\covereps^4 \tmwu - \covereps^3 \sum_{t \leq \tmwu} \langle p^{(t)}, m^{(t)} \rangle \right).
        \]
    \end{claim}
    \begin{proof}
    Again, by our update rule in Line \ref{line:update-rule},
    \[
    \Phi^{(t+1)} = \sum_{v \in V} w^{(t+1)}_v = \sum_{v \in V} w^{(t)}_v \cdot \exp(-\covereps^3 m^{(t)}_v).
    \]
    Remember that $0 \leq z^{(t)}_S \leq \frac{1}{\covereps}$ for all $S \subseteq V$. Moreover, the support of $z^{(t)}$ consists of disjoint sets. This implies that $m^{(t)}_v \in [-1,\frac{1}{\gamma}-1]$ for all vertices $v$ and steps $t$.
    \begin{align*}
        \sum_{v \in V} w^{(t)}_v \cdot \exp(-\covereps^3 m^{(t)}_v) 
        & \leq \sum_{v \in V} w^{(t)}_v (1 - \covereps^3 m^{(t)}_v + (\covereps^3 m^{(t)}_v)^2) \\
        & \leq \sum_{v \in V} w^{(t)}_v (1 - \covereps^3 m^{(t)}_v + \covereps^4) \\
        & = \sum_{v \in V} w^{(t)}_v (1 + \covereps^4) - \sum_{v \in V} w^{(t)}_v \covereps^3 m^{(t)}_v \\
        & = \Phi^{(t)} (1 + \covereps^4) - \sum_{v \in V} \Phi^{(t)} p^{(t)}_v \covereps^3 m^{(t)}_v \\
        & = \Phi^{(t)} \left (1 + \covereps^4 - \covereps^3 \langle p^{(t)}, m^{(t)} \rangle \right)\\
        & \leq \Phi^{(t)} \exp\left(\covereps^4 - \covereps^3 \langle p^{(t)}, m^{(t)} \rangle \right).
    \end{align*} 
    The first inequality follows from the fact that $e^x \leq (1 + x + x^2)$ for $x \in [-1,1]$.
    \end{proof}

\begin{claim}
\label{claim:uncovered}
    Let $U$ be the set of vertices that are uncovered by $\hat z$ (i.e., $U=\{v \in V \mid \sum_{S \ni v} \hat z_S \leq 1 - 2\covereps \}$). Then, $$\sum_{v \in U} \dc(v) \leq 2\covereps~\dc(V).$$
\end{claim}
\begin{proof}
    Remember that 
    \[
    \Phi^{(\tmwu+1)} = \sum_{v \in V} w^{(1)}_v \cdot \exp(-\covereps^3 \sum_{t \leq \tmwu} m^{(t)}_v).
    \]
   Fix a vertex $u \in U$. Since $u$ is uncovered by $\hat z$, we have 
   \[
        \tmwu \cdot 2\covereps \leq \tmwu \cdot \left (1 -  \sum_{S \ni u} \hat z_S \right) = \sum_{t \leq \tmwu} \left (1 - \sum_{S \ni u} z^{(t)}_S \right ) =  -\sum_{t \leq \tmwu} m^{(t)}_u.
   \]
   Using this, we can lower bound the potential,
       \[
    \Phi^{(\tmwu+1)} \geq \sum_{v \in U} w^{(1)}_v \cdot \exp(2\covereps^4 \tmwu).
    \]
    Together with the upper bound from Claim \ref{clm:potential},
    \[
         \sum_{v \in U} w^{(1)}_v \cdot \exp(2\covereps^4 \tmwu) \leq \Phi^{(\tmwu+1)} \leq \Phi^{(1)} \cdot \exp \left (\covereps^4 \tmwu - \covereps^3 \sum_{t \leq \tmwu} \langle p^{(t)}, m^{(t)} \rangle \right).
    \]
    Observe that that $\langle p^{(t)}, m^{(t)} \rangle$ will be non-negative since the points $z^{(t)}$ satisfy the single constraint $\sum_{S} p^{(t)}(S) \cdot z_S \geq 1$.
    Indeed,
    \[
        \langle p^{(t)}, m^{(t)} \rangle = \sum_{v \in V} p^{(t)}_v m^{(t)}_v = \sum_{v \in V} p^{(t)}_v \left ( \sum_{S \ni v} z_S - 1 \right ) = \sum_{S} p^{(t)}(S) \cdot z_S - 1.
    \]
    We can conclude that,    
    \[
         \sum_{v \in U} w^{(1)}_v \cdot \exp(2\covereps^4 \tmwu)\leq \Phi^{(1)} \cdot \exp \left (\covereps^4 \tmwu - \covereps^3 \sum_{t \leq \tmwu} \langle p^{(t)}, m^{(t)} \rangle \right) \leq \sum_{v \in V} w^{(1)}_v \cdot \exp(\covereps^4 \tmwu).  
    \]
    Rearranging,
    \[
        \sum_{v \in U} \dc(v) = \sum_{v \in U} w^{(1)}_v \leq \sum_{v \in V} w_v^{(1)} \cdot \exp((\covereps^4 - 2\covereps^4) \tmwu) = \covereps \cdot \sum_{v \in V} w^{(1)}_v = 2 \covereps \cdot \dc(V).
    \]
The second equality holds by the choice of $\tmwu = \frac{\log{(1/\covereps)}}{\covereps^4}$.
    
\end{proof}
Next, we will bound the objective value of $\zmwu$. 
By the second property of Lemma \ref{lemma:update-point},  we have $\covers(z^{(t)}) \leq (1+5\covereps)~\covers(\opt)$ for each $t \leq \tmwu$. 
Since $\covers$ is a linear function, this upper bound also holds for the average $\frac{1}{\tmwu} \sum_{t=1}^{\tmwu} z^{(t)}$. 
At the end of Algorithm \ref{alg:mw}, we cover all uncovered atoms $U = \{v \in V \mid \sum_{S \ni v} \hat z_S \leq 1 - 2\covereps \}$.
When we set a coordinate \( \hat{z}_{K} \) to \( 1 \), the \( \covers \) objective value increases by at most \( 3/2 \cdot \dc(K) \). We have $\cost(K) = 1/2 \cdot \dc(K)$. The additional \( \dc(K) \) is the term added to \( \cost(K)\) for the~\ref{LP:coverclusterlp}.
Recall that each $z^{(t)}$ does not split atoms. Hence, if $v \in U$ then $K(v) \subseteq U$. By Claim \ref{claim:uncovered}, the cost increases by at most,
\[
    \sum_{v \in U} \dc(v) \leq 2 \covereps \cdot \dc(V) \leq 2\covereps \cdot \covers(\opt).
    \]\
Thus, the objective of $z$ can be bounded by,
\[
    \covers(z) \leq \frac{1 + 7 \covereps}{1-2\covereps} \covers(\opt) \leq (1+10\covereps)\covers(\opt). 
\]
\end{proof}

%% file: familyclusters-poly.tex
In this section we will show how to find the point $z^{(t)}$ in Line \ref{line:findPoint} of
the MWU Algorithm \ref{alg:mw}.

\begin{restatable}{lemma}{updatePoint}
\label{lemma:update-point}
Given vertex weights $p_v > 0$ for all $v \in V$, we can construct a point $z \in \left[0, 1/{\covereps}\right]^{2^{|V|}}$ that,
\begin{enumerate}
    \item satisfies the single constraint $\sum_{S} p(S) \cdot z_S \geq 1$,
    \item has objective $\covers(z) \leq (1+5\covereps)~\covers(\opt)$,
    \item does not split atoms (i.e., if $z_S > 0$ then $K(v) \subseteq S$ for all vertices $v \in S$),
    \item has disjoint support (i.e., if $z_S, z_T > 0$ for two distinct $S,T \subseteq V$, then $S \cap T = \emptyset$).
\end{enumerate}
\end{restatable}

Throughout this section, we assume we are given a distribution $p$ over vertices. We have $p_v > 0$ for each vertex $v \in V$ and $\sum_{v \in V} p_v = 1$. For a set $S \subseteq V$ of vertices, we write $p(S) = \sum_{v \in S} p_v$. 
Similarly, we will abuse notation and write $p(\mathcal{F}) = \sum_{S \in \mathcal{F}} p(S)$ for a family $\mathcal{F} \subseteq 2^V$ of subsets of vertices.

In order to construct the point $z$, we will find a partial clustering $\mathcal{F} = \{S_1, S_2,..., S_l \mid S_i \subseteq V \}$ that achieves a small ratio $\frac{\covers(\mathcal{F})}{p(\mathcal{F})}$. 
Each vertex $v$ is contained in at most one set $S \in \mathcal{F}$. Moreover, $\mathcal{F}$ is a partial clustering, meaning that $\mathcal{F}$ might not cover all the vertices in $V$. However, we require that $\mathcal{F}$ covers at least a constant fraction of the probability mass of the vertices (i.e., $p(\mathcal{F}) \geq \covereps$). We show that Algorithm \ref{alg:disjointfamily-poly} will find such a partial clustering.

\paragraph{Description of Algorithm \ref{alg:disjointfamily-poly}.} Algorithm \ref{alg:disjointfamily-poly} relies on Lemma \ref{lem:goodratio}, which we will prove in Section \ref{sec:goodratio}. Lemma \ref{lem:goodratio} states that given a vertex $r$ and the optimal value $R \approx \covers(\opt)$, we can efficiently find a single cluster $C_r \ni r$ with a small ratio $\frac{\covers(C_r)}{p(C_r)} \leq R$ if there exists such a cluster and in particular, if the cluster $C \in \opt$ that contains $r$ achieves the ratio $R$. There has to exist a cluster $C \in \opt$ that achieves the ratio $R$ since $p$ is a probability distribution. Thus, we can find a cluster $C_r$ for each vertex $r \in V$ and add the cluster with the smallest ratio to the partial clustering $\mathcal{F}$; this ratio will be at most $R$. If the cluster $S$ we find that way has small vertex probability mass $p(S) < \covereps$, we have to repeat this process. Note that when we remove the vertices in $S$, we only remove $p(S) < \covereps$ of probability mass. We argue that in this case, there will exist a cluster in $\opt$ with ratio at most $(1+2\covereps)\covers(\opt)$ and we can find a cluster disjoint from $S$ that achieves ratio $(1+2\covereps)\covers(\opt)$. Hence, we can repeat until we cover constant probability mass with $\mathcal{F}$.    
For now, the algorithm will have a runtime of $\widetilde O(n^4)$. In Section \ref{sec:nearlylinear}, we will show how we can modify Algorithm \ref{alg:disjointfamily-poly} to achieve a nearly linear runtime.  

\begin{algorithm}[t]
	\caption{Algorithm to find the family $\mathcal{F}$}
	\label{alg:disjointfamily-poly}
	\begin{algorithmic}[1]
\STATE Let $R$ be the guess for $\covers(\opt)$ such that $R \in [\covers(\opt) , (1 +\covereps) \covers(\opt))$.
    \STATE $\hat p \gets p, \mathcal{F} \gets \emptyset, \hat{V} \gets V$
    \FORALL{$v \in V$}
        \STATE Find a small ratio cluster $C_v$ where $K(v) \subseteq C_v \subseteq \Nadm(v)$ with vertex weights $\hat p > 0$ and target ratio $(1+3\covereps) R$ (Lemma \ref{lem:goodratio}).
    \ENDFOR
    \WHILE{$p(\mathcal{F}) \leq \covereps$} \label{line:loop_family-poly}
        \STATE Choose $C$ with the smallest ratio $\frac{\covers(C)}{\hat p(C)}$ among clusters $\{C_v\}_{v \in \hat{V}}$. 
        \STATE Add $C$ to $\mathcal{F}$, set $\hat p_v$ to $0$ for all $v \in C$ and remove vertices in $C$ from $\hat{V}$.
        \FORALL{$v \in \hat{V}$} 
            \STATE \COMMENT{Update $C_v$ if some node in $C_v$ is added to $\mathcal{F}$}
            \IF{$C_v \cap C \neq \emptyset$} 
            \STATE Find a new small ratio cluster $C_v$  with vertex weights $\hat p > 0$ and target ratio $(1+3\covereps)R$ (Lemma \ref{lem:goodratio}).
            \ENDIF
        \ENDFOR
    \ENDWHILE
    \RETURN $\mathcal{F}$
	\end{algorithmic}
\end{algorithm}

\begin{lemma}
\label{lem:family-poly}
Given vertex weights $p_v > 0$ for all $v \in V$, Algorithm \ref{alg:disjointfamily-poly} finds a family $\mathcal{F} = \{S_1, S_2,..., S_l \mid S_i \subseteq V \}$ such that,
\begin{enumerate}
    \item for any distinct $S, T \in \mathcal{F}$, $S \cap T = \emptyset$,
    \item $\frac{\covers(\mathcal{F})}{p(\mathcal{F})} \leq (1+5\covereps)\covers(\opt)$,
    \item $p(\mathcal{F})$ is at least $\covereps$,
    \item no $S \in \mathcal{F}$ splits an atom (i.e., $K(v) \subseteq S$, for all vertices $v \in S$).
\end{enumerate}
\end{lemma}
\begin{proof}
First, observe that Property 3. holds by the condition of the {\bf while} loop. Property 4. holds since the all small ratio clusters $\{C_v\}_{v \in V}$ do not split atoms by Lemma \ref{lem:goodratio}.
Next, we observe that Property 1. holds:  the family $\mathcal{F}$ consists of disjoint sets. Observe that after we added a cluster $C \subseteq V$ to $\mathcal{F}$ we set the weight $\hat p_v$ to $0$ for all $v \in C$. Note that if the weight $\hat p_v$ is set to $0$ it remains $0$ throughout the execution of the algorithm. 
Moreover, when we add a cluster $C$ to $\mathcal{F}$, it contains only vertices such that $\hat{p}_v > 0$.
Thus, when we add $C$ to $\mathcal{F}$ we have $\hat p_v > 0$ for all $v \in C$ and $\hat p_v = 0$ for all $v \in S, S \in \mathcal{F}$. 
 
Now we address Property 2. Based on our assumption, we have $R$ such that $\covers(\opt) \leq R < (1+\covereps)\covers(\opt)$. It remains to prove the bound on the ratio $\frac{\covers(\mathcal{F})}{p(\mathcal{F})}$.
Let $C \in \opt$ and $v \in C$ such that if $K$ is a non-singleton atom in $C$ then $K=K(v)$.  (If $C$ contains no such atom, then $v$ can be any vertex in $C$.)
As long as, 
\begin{align}
    \frac{\covers(C) + \covereps^2 \dadm(C)}{\hat p(C)} \leq (1+3\covereps) R, \label{C-assumption}
\end{align}
we maintain the following invariant:
\begin{align}
\label{eq:inv}
    \frac{\covers(C_v)}{\hat p(C_v)} \leq (1+3\covereps) R.
\end{align}
Before starting the {\bf while} loop on Line \ref{line:loop_family-poly}, we have by Lemma \ref{lem:goodratio},
\[
    \frac{\covers(C_v)}{p(C_v)} \leq(1+3\covereps)R.
\]
since $C$ satisfies the requirements of Lemma \ref{lem:goodratio} for the vertex $v$ and ratio $(1+3\covereps)R$ by assumption.
Consider one iteration of the {\bf while} loop on Line \ref{line:loop_family-poly}. 
If we did not remove a vertex from $C_v$ then the invariant remains true. 
Again, if $C_v$ was updated, we have by Lemma \ref{lem:goodratio},
\[
    \frac{\covers(C_v)}{\hat p(C_v)} \leq (1+3\covereps) R.
\]
\newline

Let $C^\star$ be the cluster in the optimal clustering $\opt$ with the best ratio $\frac{\covers(C^\star)+\covereps^2 \dadm(C^\star)}{\hat p(C^\star)}$.
By the condition of the {\bf while} loop on Line \ref{line:loop_family-poly}, $\hat p(V) = p(V) - p(\mathcal{F}) \geq 1 - \covereps$.
Note that,
\begin{align*}
    (1+3\covereps)R &\geq \frac{1+ \covereps}{1-\covereps} R \geq \frac{(1+\covereps)\covers(\opt)}{\hat p(V)} \geq \frac{\covers(\opt) + \covereps^2 |E_{\adm}|}{\hat p(V)} \\
    &= \frac{\sum_{C \in \opt}(\covers(C) + \covereps^2 \dadm(C))}{\sum_{C \in \opt} \hat p(C)} \geq \frac{\covers(C^\star)+\covereps^2 \dadm(C^\star)}{\hat p(C^\star)}.
\end{align*}
For the second inequality, we use that $\gamma~\covers(\opt) \geq \gamma~\clcost(\opt) \geq \gamma^2 |E_{\adm}|$ because of the preclustering (Theorem \ref{thm:preclustering-proc}). 
The last inequality holds by the definition of $C^\star$. Let $v \in C^\star$ such that if $K$ is a non-singleton atom in $C^\star$ then $K=K(v)$. 
By (\ref{eq:inv}),
\[
    \frac{\covers(C_v)}{\hat p(C_v)} \leq (1+3\covereps) R.
\]
Since we choose $C_r$ as the cluster with the best ratio among clusters $\{C_u\}_{u \in V}$,
\[
    \frac{\covers(C_r)}{\hat p(C_r)} \leq \frac{\covers(C_v)}{\hat p(C_v)} \leq (1+3\covereps)R.
\]
We can conclude that,
\[
     \frac{\covers(\mathcal{F})}{p(\mathcal{F})} = \frac{\sum_{S \in \mathcal{F}}\covers(S)}{\sum_{S \in \mathcal{F}}p(S)} \leq (1+3\covereps) R \leq (1+5\covereps) \covers(\opt).
\]
\end{proof}
We can use the family $\mathcal{F}$ to construct the point $z^{(t)}$ that we will use in one iteration of the MWU algorithm. 
From Lemma \ref{lem:family-poly} we can readily derive Lemma \ref{lemma:update-point}.
\updatePoint*
\begin{proof}
    For each $S \in \mathcal{F}$, set $z_S = \frac{1}{p(\mathcal{F})} \geq 1$.
    First, we prove properties 3. and 4.    
    Observe that the support of $z$ is equal to $\mathcal{F}$ which contains only disjoint sets. Similarly, $z$ does not split atoms since $\mathcal{F}$ does not. 
    Property 1. is satisfied since,
        \[
            \sum_{S} p(S) \cdot z_S = \sum_{S \in \mathcal{F}} \frac{p(S)}{p(\mathcal{F})} = 1. 
        \]
   To finish the proof, we can bound the objective as follows,
    \[
            \covers(z) = \sum_{S \in \mathcal{F}} \covers(S) z_S = \frac{\covers(\mathcal{F})}{p(\mathcal{F})} \leq (1+5\gamma) \covers(\opt).
    \]
\end{proof}

\begin{lemma}
    The runtime of Algorithm \ref{alg:disjointfamily-poly} is $\widetilde{O}(n^4)$.
\end{lemma}

\begin{proof}
    Finding 
    the small ratio cluster $C_v$ for a single vertex $v$ takes time $\widetilde O(d^2(v))$. So overall, we need time $\widetilde O(n^3)$ to find all $\{C_v\}_{v\in V}$. There are at most $n$ iterations of the {\bf while} loop on Line \ref{line:loop_family-poly} since we set $\hat p_v = 0$ for at least one vertex $v$. This can happen at most $n$ times. Indeed, once a vertex has weight $\hat p_v = 0$, the weight will remain $0$. During one iteration of the {\bf while} loop on Line \ref{line:loop_family-poly}, we might have to update all $\{C_v\}_{v\in V}$. Again, this takes time at most $\widetilde O(n^3)$. 
\end{proof}

%% file: findgoodcluster.tex
In this section, we show how to find a small ratio cluster $C_r$, which is a key subroutine in Algorithm \ref{alg:disjointfamily-poly}. 
We define the candidate set of $r$, $\Ncand(r)$, to be the vertices that could possibly belong to the small ratio cluster $C_r$.
\begin{align*}
    \Ncand(r) = \begin{cases}
    \Nadm(r) \setminus{V_{\calK}} & \text{if } r \text{ is a singleton atom,} \\
    K(r) \cup \left( \bigcap_{u \in K(r)} \Nadm(u) \right) & \text{if } r \text{ belongs to a non-singleton atom}.
\end{cases}
\end{align*}
Intuitively, if $r$ belongs to some non-singleton atom, then we only need to consider all admissible neighbors plus the nodes that are in the same atom. If $r$ is a singleton atom, we do not need to consider any neighboring non-singleton atoms; we only need to consider admissible neighbors that are singleton atoms.  Notice that the definition of $\Ncand$ is not symmetric (i.e., it might be the case that $u \in \Ncand(v)$, but $v \notin \Ncand(u)$).

In this section, we want to show the following Lemma.
\begin{lemma}
\label{lem:goodratio}
Suppose we are given a graph $G = (V, E)$, vertex weights $\hatp$, a target ratio $R$, a vertex $r$ and the set of vertices $\Nadm(r)$. 
\begin{itemize}
    \item[(i)] Assume there exists a cluster $C$ be an \epslarge cluster with $K(r) \subseteq C \subseteq \Ncand(r)$ such that $\covers(C) + \covereps^2  \dadm(C) \leq R \cdot \hatp(C)$.
    \item[(ii)] Assume that $\covers(\{ v \}) > R \cdot \hatp(\{ v \})$ for all $v \in C$.
\end{itemize}
Then, with high probability, in time $\widetilde O(d^2(r))$, we can find
a cluster $C_r \subseteq \Ncand(r)$ such that,
\[
    \covers(C_r) \leq R \cdot \hatp(C_r).
\]
Moreover, $C_r$ does not split atoms and contains exactly one non-singleton atom $K(r) \subseteq C_r$ iff $K(r)$ is a non-singleton atom.
\end{lemma}

Lemma~\ref{lem:goodratio} may return a cluster \( C_r \) that contains some nodes with \(\hat{p}_v = 0\). In such cases, we can simply remove these nodes, as the monotonicity of \(\covers\) ensures that removing nodes with \(\hat{p} = 0\) will only decrease the ratio.

The algorithm is parameterized by $\covereps, \eta, \eps$. 
Recall that $\covereps = O(\eps^{13})$
is a small enough constant and $\eta = \Omega(\covereps^{-2}\eps^{-8}) =\Omega(\eps^{-34})$ is a sufficiently large constant. We also assume that $|\Ncand(r)| = \Omega(\eta^2)$, otherwise, we can enumerate all possible subsets of $\Ncand(r)$, which takes $O(2^{\poly(1/\eps)})$ time.

\subsection{Overview of the Algorithm}

In this section, we aim to solve the following optimization problem: 
Given a preclustered \cc instance and vertex weights $w$, the goal is to find a set $T$ of vertices such that 
\begin{align*}
    \clcost(T) \leq \sum_{v \in T} w(v). 
\end{align*} 

Setting $w(v) := R \cdot p(v) - \dc(v)$, if we can find a set $T$ such that 
\begin{align*}
    \clcost(T) \leq \sum_{v \in T} \big(  R p(v) - \dc(v) \big) 
    \leq \sum_{v \in T}  R p(v) - \sum_{v \in T} \dc(v), 
\end{align*}
then $T$ is a small ratio cluster, since
\begin{align*}
\covers(T) &= \clcost(T) + \sum_{v \in T} \dc(v) \leq R \sum_{v \in T}p(v) = R \cdot p(T).
\end{align*}

Define $D(v) := \Ncand(v) \setminus{K(v)}$ for any $v \in V$, and $\Delta(T) := \clcost(T) - w(T)$ for any $T \subseteq V$. 
As is the case of $\Nadm(v)$ (i.e., see Definition \ref{def:epsgoodpreclustering}), we can also bound the size of $\Ncand(v)$, which we do by the following lemma.
\begin{lemma}
\label{lem:sizeofcandidateset}
For any $r \in V$, if $v \in \Ncand(r)$, then $|\Ncand(r)| = O(\eps^{-4}d(v))$.
\end{lemma}
\begin{proof}
    We have by definition of $\Ncand$ that $|\Ncand(r)| \leq |K(r)| + |\Nadm(r)| = O(\epsilon^{-3} d(r))$. 
    The last equality is true since the preclustering is $\epsilon$-similar. Note that for any vertex $v \in \Ncand(r)$ the edge $(r,v)$ is admissible. Thus, $d(r) \leq \epsilon^{-1} d(v)$ and $|\Ncand(r)| = O(\eps^{-4}d(v))|$.  
\end{proof}

For a set of vertices $T$ and a vertex $v$, we want to compare the cost of including $v$ in $T$ to the cost of not including $v$ in $T$, define the marginal value $v$ with respect to $T$ as,
\begin{align*} 
    \marginal(T, v) &= \clcost(T \cup \{ v \}) - \clcost(T \setminus \{v \}) \nonumber\\
    &= d^-(v,T) + \frac{1}{2} d^+(v, V\setminus (T \cup \{ v \} )) - \frac{1}{2} d^+(v, T \setminus \{v \}) \nonumber\\
      &= |T|-d^+(v, T) + \frac{1}{2}(d^+(v)-1) - d^+(v, T \setminus \{v \}) \nonumber \\
    &= \frac{d^+(v) - 1}{2} + |T|  - 2d^+(v,T) + \mathbb{1}(v \in T). \label{val-v-in-T}
\end{align*}
where \( \mathbb{1} \) is the indicator function for the event \( v \in T \). Recall that we assume \( uu \) is also a \pedge in \( G \). In the above calculation, the indicator variable corresponding to whether \( v \) belongs to \( T \) is used to adjust for the fact that the calculation differs slightly depending on whether \( v \) is in \( T \) or not.
Since we do not know whether or not $v \in T$, we need this extra term.

One of the most important properties for function $\marginal$ is the following.
\begin{claim}
\label{claim:marginalchangesinsmallchunk}
    For any non-empty set $T$ and $T'$ and any vertex $v$, we have 
    \begin{align*}
        \marginal(T, v) - \marginal(T', v) \leq 2 |T \oplus T'|
    \end{align*}
    where $\oplus$ denotes the symmetric difference of two sets.
\end{claim}
\begin{proof}
Using the definition of $\marginal$, we have 
\begin{align*}
    \marginal(T, v) - \marginal(T', v) &= |T \cup \{v\}| - |T' \cup \{v\}| - 2(d^+(v, T) - d^+(v, T')) \\
    &\leq |T \setminus T'| - |T' \setminus T| + 1 - 2(d^+(v, T \setminus T') - d^+(v, T' \setminus T)) \\
    &\leq |T \setminus T'| - |T' \setminus T| + 1 + 2 d^+(v, T' \setminus T)) \\
    &\leq |T \oplus T'| + 1 \leq 2|T \oplus T'|.
\end{align*}
\end{proof}

Claim \ref{claim:marginalchangesinsmallchunk} implies that we can estimate the value $\marginal(T,v)$ by the value $\marginal(T',v)$ if $T'$ is almost the same as $T$.
In each round, we try to decide whether or not a node $v$ should be in $T$ using $\marginal(T, v)$. However, we do not know $T$, so we will instead sample a small set $S$ from $T$ 
and estimate $\marginal(T, v)$. To do this estimation, given a small sample set $S$, an integer guess $t$ for the size of $T$ and a vertex $v$, define 
\[\rmEstCost(S,t,v) := 
\frac{\dpos(v) - 1}{2} + t  - 2\frac{d(v, S)}{|S|}t.\]
We will sample a constant number of nodes, so with constant probability, we can ensure $\rmEstCost$ is close to $\marginal$.

\paragraph{Description of Algorithms
\ref{alg:generateclusterBySampling} and \ref{alg:GenerateCluster}
.} Given the definitions of \( \marginal(T, v) \) and \( \rmEstCost(S, t, v) \), we can now describe Algorithm~\ref{alg:generateclusterBySampling} and Algorithm~\ref{alg:GenerateCluster}. Conceptually, we begin by guessing \( C^* \subset \Ncand(r) \), where $C^*$ is defined to be the cluster for which \( \cost(C^*) - w(C^*) \) is minimized.
Due to the optimality of \( C^* \), adding a node to \( C^* \) or removing a node from \( C^* \) will increase this value,
which implies that \( \marginal(C^*, v) - w(v) \leq 0 \) for all \( v \in C^* \) and \( \marginal(C^*, v) - w(v) \geq 0 \) for all \( v \not\in C^* \). Based on this observation, we should add any node \( v \) if and only if \( \marginal(C^*, v) - w(v) \leq 0 \).

The algorithm does not know $C^*$, so to compute the value of \( \marginal(C^*, v) \) for any node $v$, Algorithm~\ref{alg:generateclusterBySampling} attempts to sample enough nodes from \( C^* \). We will later show that \( |C^*| = \Omega(\covereps^2 \epsilon^{-8} |\Ncand(r)|) \), ensuring that we always obtain some sampled vertices from \( C^* \).

In this process, instead of sampling once, we sample \( \eta \) different sets \( A_i \). This is because we need new sampled nodes each time we add vertices to our final set \( \hatt \). The algorithm also tries to guess the size of \( C^* \). Since \( C^* \) is an \epslarge cluster, we know the size of \( C^* \) will be within \( [\epsilon d^{+}(r), |\Ncand(r)|] \). We do not need the exact size, so we will enumerate possible sizes with different granularities, choosing values from the following set:
\begin{align*}
    L(r) = \left\{ \left( 1 + \frac{1}{\eta} \right)^j \in [\epsilon d^{+}(r), \epsilon^{-4} d^{+}(r)] \mid j \text{ is an integer} \right\}.
\end{align*}

\begin{algorithm}[H]
\caption{\genclusbysam($G, \Nadm(r), w, r, R$)}
\label{alg:generateclusterBySampling}
\begin{algorithmic}[1]
\STATE \textbf{Input:} The graph $G$, $K(r)$, $\Nadm(r)$, $\Ncand(r), w$, $r$, ratio $R$.
\STATE \textbf{Output:} A small ratio cluster $\hatt$ if $r$ satisfies Assumption (i) from Lemma \ref{lem:goodratio}. 
\STATE Repeat the following steps $O(\log n)$ times.
\FOR{$i$ from $1$ to $\eta$}
\STATE Uniformly sample $\Theta(\eta^4  \covereps^{-2} \eps^{-8})$ vertices from $\Ncand(r)$ with replacement
\STATE Let the sample set be $A_i$. \COMMENT{$A_i$ may contain some element multiple times.}
\ENDFOR
\STATE $D(r) \gets \Ncand(r)\setminus{K(r)}$
\FOR{every $(S^1, S^2,...,S^{\eta}) \subset (A_1, A_2,...,A_{\eta})$ such that $|S^i| \leq \eta$, where $i \in [\eta]$} 
    \FOR{ every $(\tilt_1, \tilt_2, ..., \tilt_{\eta}) \in (L(r), L(r),...L(r))$, where $\tilt_j \in L(r)$ for $j \in [\eta]$ }
    \STATE $T \gets $ GenerateCluster($r, D(r), S^1,\ldots,S^\eta, \tilt_1,\ldots, \tilt_\eta$)
    \IF{$\cost(T) \leq w(T)$}
        \RETURN $T$
        \ENDIF
\ENDFOR
\ENDFOR
\RETURN $\emptyset$
\end{algorithmic}
\end{algorithm}

After the sampling step, we proceed to add vertices to \( \hatt \). We must be careful regarding \( C^* \) because, once we add vertices, we may introduce errors. Using the same sampled nodes repeatedly could make the estimation of \( \marginal \) inaccurate. To address this, we divide \( \Ncand(r) \) into \( \eta \) "chunks", each of size \( |\Ncand(r)| / \eta \). For each chunk, we use the estimated \( \marginal \) to decide whether to add a node to \( \hatt \). Once we finish processing a chunk, we update our sampled set to align with the choices we have already made. This process ensures that we return a good cluster with constant probability. By repeating it \( O(\log n) \) times, we obtain our final small ratio cluster.

\begin{algorithm}[H]
\caption{\genclus($r, D(r), S^1,\ldots,S^\eta, \tilt_1,\ldots, \tilt_\eta$)}
\label{alg:GenerateCluster}
\begin{algorithmic}[1]
\STATE $T \gets K(r), \hatt_1 \gets K(r)$
\STATE Let $D_r^1, \ldots, D_r^{\eta}$ be an arbitrary partition of the vertices of $D(r)$ into equal-size parts
\FORALL{$i=1,\ldots,\eta$} \label{loop:for-gc}
\FORALL{$v \in D_r^i$} 
\IF{$\rmEstCost(S^i, \tilt_i, v$)$ + 6\eta^{-1} |\Ncand(r)| \leq w(v)$}
\STATE $T \gets T \cup \{ v \}$ 
\ENDIF
\ENDFOR 
\STATE $\hatt_{i+1} \gets T$
\ENDFOR
\end{algorithmic}
\end{algorithm}

Let $C^*$ be the cluster such that $K(r) \subseteq C^* \subseteq \Ncand(r)$ and $\Delta(C^*) = \clcost(C^*) - w(C^*)$ is minimized.  For $C$, which exists by Assumption (i) in Lemma \ref{lem:goodratio}, we have $\Delta(C^*) \leq \Delta(C) \leq - \covereps^2 \dadm(C)$.  
To analyze the algorithm, recall $\hatt_{i}$
is the set of vertices in set $\hatt$ 
at the beginning of the $i$-th iteration of the {\bf for} loop on Line \ref{loop:for-gc} of Algorithm \ref{alg:GenerateCluster}.  Notice that $\hatt_{i-1} \subseteq \hatt_i$.
We define 
\[ C^*_i = \hatt_{i} \cup \mathrm{argmin} \left\{ \clcost(\hatt_{i} \cup B) - w(\hatt_{i} \cup B) ~\Big|~B \subseteq \bigcup_{j = i}^{\eta}D_r^j\right\}.\]
Assume now that $S^i$ is a uniform sample of the cluster $C^*_i$. Let $t_i = |C^*_i|$ be the size of $|C^*_i|$, and assume that $\tilt_i \in [t_i,(1+\frac{1}{\eta})t_i]$. Our main lemma regarding Algorithm~\ref{alg:GenerateCluster} is given as follows.

\begin{lemma}
\label{lem:goodratio2}
Suppose we are given a graph $G = (V, E)$, vertex weights $w$ and vertex $r$.
\begin{itemize}
    \item Assume there exists a cluster $C$ that is \epslarge cluster with $K(r) \subseteq C \subseteq \Ncand(r)$ such that $\clcost(C) + \covereps^2 \dadm(C) \leq w(C)$.
\end{itemize}
Let $C^*_i$ be the set defined above and $S^i$ be a uniform sample with size $\eta_0 = \Omega(\eta^3)$ from $C^*_i$, $\tilt_i \in [t_i,(1+\frac{1}{\eta})t_i]$ is the guess for the size of $C^*_i$. Then with probability $1 - 2\eta \exp(-2\eta)$,

\genclus($r, D(r), S^1,\ldots,S^\eta, \tilt_1,\ldots, \tilt_\eta$) produces a cluster $\hat{T}$ such that $\clcost(\hatt) \leq  w(\hatt)$.
\end{lemma}

\begin{proof}
We will later show the following claim,

\begin{claim}
\label{claim:invariant-4.2}
With probability at least $1 - 2\exp(-2\eta)$,

\begin{align*}
    \Delta(C^*_{i + 1}) - \Delta(C^*_i) = O( \eta^{-2} \eps^{-8} \dadm(C) ).
\end{align*}
\end{claim}

Assuming the claim is true, a union bound implies that with probability at least $1 - 2\eta \exp(-2\eta)$, 
this claim holds for all $i = 1, \dots, \eta$. 
Lemma \ref{lem:goodratio2} follows by observing that $\hatt = C^*_{\eta+1}$ and $C^*_1 = C^*$. 
Recall, by definition, $\Delta(C^*_i) = \clcost(C^*_i) - w(C^*_i)$. For any $C^*_i$, where $i \in [\eta]$, we have
\begin{align*}
\Delta(C^*_i) & = \Delta(C^*_1) + \sum_{j = 1}^{i - 1}(\Delta( C^*_{j+1}) - \Delta( C^*_j)) \\
    & \le \Delta(C^*) + \eta \cdot O( \eta^{-2} \varepsilon^{-8} \dadm(C)) \\
    & \le \Delta(C^*) + \frac{\covereps^2}{2} \dadm(C) \\
    & \leq - \frac{\covereps^2}{2} \dadm(C).
\end{align*}

The next-to-last inequality follows from the assumption $\eta = \Omega(\covereps^{-2}\eps^{-8})$. 
\end{proof}

\begin{claim}
\label{claim:lowerboundofsampleset}
Either $|C^\star_i| = s_i = \Omega(\covereps^2 \eps^8 |\Ncand(r)|)$ for all $i \in [\eta]$ or there exists a vertex $u \in \Ncand(r)$ such that $\clcost(u) \leq  w(u)$.
\end{claim}
\begin{proof}
Note that for any $i \in [\eta]$,  we have,
\begin{align*}
    \Delta(C^*_i) \leq - \frac{\covereps^2}{2} \dadm(C),
\end{align*}
where $C$ is the \epslarge cluster such that $K(r) \subseteq C \subseteq \Ncand(r)$. By Lemma \ref{lem:sizeofcandidateset}, $|\Ncand(r)| = O(\epsilon^{-3} d(r) )$.
If $r$ is in a non-singleton atom $K(r) \subseteq C^\star_i$,
then $|C^\star_i| \geq |K(r)| \geq (1 - O(\epsilon)) d(r) \geq \Omega( \epsilon^3 |\Ncand(r)|)$. $|K(r)| \geq (1 - \epsilon) d(r)$ is due to preclustering, every non-singleton atom has at most $O(\eps)$ fraction $+$neighbors outside of $K(r)$. 

Otherwise, all vertices in $\Ncand(r)$ are in singleton-atoms. Assume that $|C^\star_i| \leq \covereps^2 \epsilon^8 |\Ncand(r)| = O( \covereps^2 \epsilon^5 d(r))$. In this case,
\[    - \frac{\covereps^2}{2} \dadm(C) \geq \Delta(C^\star_i) = \cost(C^\star_i) - \sum_{v \in C^\star_i} w(v) \geq\frac{1}{2} \sum_{v\in C^\star_i} d(v) - |C^\star_i|^2 - \sum_{v \in C^\star_i} w(v). \]

The last inequality uses $\cost(C') + |C'|^2 \geq \frac{1}{2} \sum_{v \in C'} d(v)$, which holds for any cluster $C'$.  (The right side counts the \pedges incident on $C'$. Either such an edge leaves $C'$ and contributes $1/2$ to the $\cost(C')$ or it is inside $C'$. There are at most $|C'|^2/2$ edges inside $C'$.)
Since $C \subseteq \Ncand(r)$, $C$ only contains vertices in singleton atoms. Since $C$ is $\epsilon$-large, all edges in $C$ are admissible.  Thus, $\dadm(C) \geq \frac{1}{2} |C|^2 \geq \frac{\epsilon^2}{2} d^2(r)$. Hence, $|C^\star_i|^2 - \frac{\covereps^2}{2} \dadm(C) \leq \epsilon^{10}\covereps^6 d^2(r) - \frac{\epsilon^2 \covereps^2}{4} d^2(r) \leq 0$.
In particular, 
\[
    \frac{1}{2} \sum_{v\in C^\star_i} d(v) - \sum_{v \in C^\star_i} w(v) \leq |C^\star_i|^2 - \frac{\covereps^2}{2} \dadm(C) \leq 0. 
\]
There has to exist a vertex $v \in C^\star_i$ such that $\frac{1}{2} d(v) - w(v) \leq 0$.
\end{proof}

Now, we are able to show the main lemma~\ref{lem:goodratio} of this section.

\begin{proof}[Proof of Lemma~\ref{lem:goodratio}]
Lemma~\ref{lem:goodratio2} outputs \( \hatt \) such that \( \cost(\hatt) \leq w(\hatt) \). However, we still need to satisfy the input conditions for Lemma~\ref{lem:goodratio2}. 

This is provided by Claim~\ref{claim:lowerboundofsampleset}, noting that each \( |C^*_i| \) has size \( \Omega(\covereps^2 \epsilon^8 |\Ncand(r)|) \). In Algorithm~\ref{alg:generateclusterBySampling}, we uniformly sample \( \Theta(\eta^4 \covereps^{-2} \epsilon^{-8}) \) vertices from \( \Ncand(r) \). The expected number of nodes we hit in \( C^*_i \) is \( \Theta(\eta^4) \). With probability at least \( 1 - \exp(-\eta) \), we obtain \( \eta_0 = \eta^3 \) sampled nodes, allowing the algorithm to enumerate all subsets of \( A_i \); at least one run will contain all sampled nodes from \( C^*_i \). By a union bound, with probability at least \( 1 - \eta \exp(-\eta) \), a call to~\genclus~will satisfy Lemma~\ref{lem:goodratio2}. Once we make the correct call to~\genclus, with probability at least \( 1 - 2\eta \exp(-2\eta) \), ~\genclus~outputs a correct answer. The high-probability guarantee comes from the fact that Algorithm \ref{alg:generateclusterBySampling}
repeats the whole process \( \log n \) times.

To argue about the runtime for finding a small ratio cluster, 
we use the statements in Theorem \ref{thm:preclustering-proc} about deciding admissibility. 
We will compute $\Ncand(r)$ as follows. We can find $\Nadm(r)$ in time $\widetilde O (d(r))$. Then we can iterate through all vertices in $\Nadm(r)$ and check if they are in $\Ncand(r)$. This takes time $\widetilde O(d^2(r))$.
Since $\eta$ and $\epsilon$ are constants, the runtime of GenerateClusterBySampling, Algorithm~\ref{alg:generateclusterBySampling}, is $\log(n)$ times the runtime of GenerateCluster, Algorithm~\ref{alg:GenerateCluster}. Here, we compute $\rmEstCost(S, t, v)$ for all $v \in \Ncand$ for a constant number of constant sized sets $S$. 
This takes at most $\tilde O(d(r))$ since we only need to iterate over the neighbors of vertices in $S$. We can compute $\clcost(T)$ in time $O(|T| \cdot d(r)) = O(d^2(r))$. Overall, we spend at most $\widetilde O(d^2(r))$. 
\end{proof}

\subsection{Bounding $\Delta(C^*_{i + 1}) - \Delta(C^*_i)$}

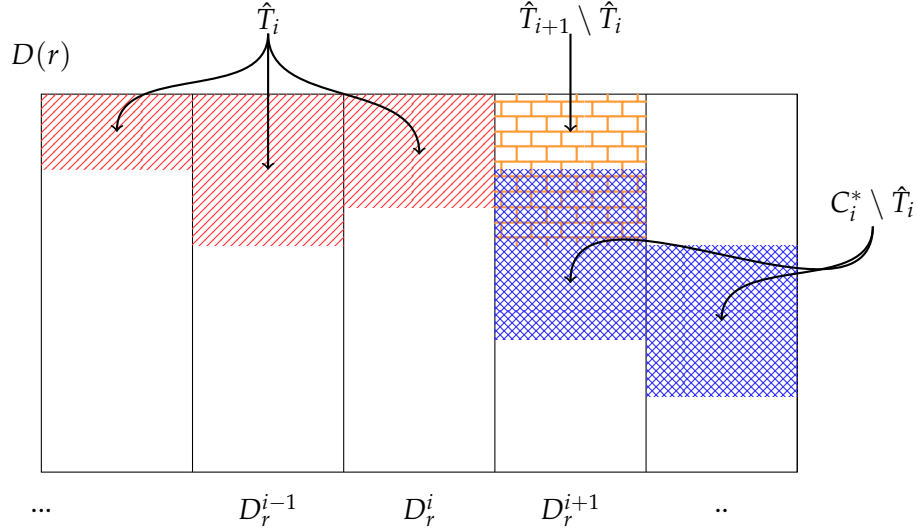
\begin{figure}[t]
    \centering
    \begin{tikzpicture}
        \def\n{5}
        \def\rectwidth{10}
        \def\rectheight{5}
        
        \draw[fill=white!10] (0,0) rectangle (\rectwidth,\rectheight);

        \foreach \i in {1,...,\n} {
            \draw (\i*\rectwidth/\n,0) -- (\i*\rectwidth/\n,\rectheight);
        }

        \node at (0, -0.5) {$...$};
        \node at (1*\rectwidth/\n + 0.5*\rectwidth/\n, -0.5) {$D^{i-1}_r$};
        \node at (2*\rectwidth/\n + 0.5*\rectwidth/\n, -0.5) {$D^i_r$};
        \node at (3*\rectwidth/\n + 0.5*\rectwidth/\n, -0.5) {$D^{i+1}_r$};
        \node at (4*\rectwidth/\n + 0.5*\rectwidth/\n, -0.5) {$..$};

        \fill[pattern=north east lines,  pattern color= red!80] (0*\rectwidth/\n, 5*\rectheight/\n) rectangle (1*\rectwidth/\n, 4*\rectheight/\n); 
        \fill[pattern=north east lines,  pattern color= red!80] (1*\rectwidth/\n, 5*\rectheight/\n) rectangle (2*\rectwidth/\n, 3*\rectheight/\n); 
        \fill[pattern=north east lines,  pattern color= red!80] (2*\rectwidth/\n, 5*\rectheight/\n) rectangle (3*\rectwidth/\n, 3.5*\rectheight/\n); 

        \fill[pattern=bricks,  pattern color= orange!80] (3*\rectwidth/\n, 5*\rectheight/\n) rectangle (4*\rectwidth/\n, 3*\rectheight/\n); 

        \fill[pattern=crosshatch,  pattern color= blue!80] (3*\rectwidth/\n, 4*\rectheight/\n) rectangle (4*\rectwidth/\n, 1.75*\rectheight/\n); 

        \fill[pattern=crosshatch,  pattern color= blue!80] (4*\rectwidth/\n, 3*\rectheight/\n) rectangle (5*\rectwidth/\n, 1*\rectheight/\n); 

        \node at (1.5*\rectwidth/\n, \rectheight + 1) {$\hatt_i$};
        \draw[->, thick, rounded corners=1pt] (1.5*\rectwidth/\n, \rectheight + 0.8) 
            to[out=-90, in=90] (0.5*\rectwidth/\n, 4.5*\rectheight/\n);
        \draw[->, thick, rounded corners=1pt] (1.5*\rectwidth/\n, \rectheight + 0.8) 
            to[out=-90, in=90] (1.5*\rectwidth/\n, 4*\rectheight/\n);
        \draw[->, thick, rounded corners=1pt] (1.5*\rectwidth/\n, \rectheight + 0.8) 
            to[out=-90, in=90] (2.5*\rectwidth/\n, 4.2*\rectheight/\n);

        \node at (3.5*\rectwidth/\n, \rectheight + 1) {$\hatt_{i+1} \setminus \hatt_{i}$};
        \draw[->, thick, rounded corners=1pt] (3.5*\rectwidth/\n, \rectheight + 0.8) 
            to[out=-90, in=90] (3.5*\rectwidth/\n, 4.5*\rectheight/\n);

        \node at (5.5*\rectwidth/\n, 3.5 \rectheight ) {$C^{*}_i \setminus \hatt_{i}$};
        \draw[->, thick, rounded corners=1pt] (5.5*\rectwidth/\n, 3.2 \rectheight) 
            to[out=-90, in=90] (3.5*\rectwidth/\n, 2.5*\rectheight/\n);
        \draw[->, thick, rounded corners=1pt] (5.5*\rectwidth/\n, 3.2 \rectheight) 
            to[out=-90, in=90] (4.5*\rectwidth/\n, 2*\rectheight/\n);

        \node at (0, \rectheight + 0.5) {$D(r)$};
    \end{tikzpicture}
    \caption{Illustration of the sets \( \hatt_i \), \( C^*_i \), and \( Q_i \). The rectangle represents \( D(r) \), divided into \(\eta\) parts. The red region denotes \( \hatt_i \), containing all vertices already added to \( \hatt \). The set \( C^*_i \) includes both the red and blue regions. In \( D^{i+1}_r \), the algorithm attempts to include as many vertices as possible in \( C^*_i \); the yellow region represents the newly added vertices in \( \hatt \). Claim~\ref{claim:invariant-4.2} states that the yellow and blue regions have significant overlap.}

\end{figure}

In this section, we will prove Claim \ref{claim:invariant-4.2}.



\begin{proof}[Proof of Claim~\ref{claim:invariant-4.2}]
Let 
\[ Q_{i} := \hatt_{i+1} \cup \left( C_i^\ast \cap \bigcup_{j=i+1}^\eta D_r^j\right).\]
Recall that 
\[ C^*_i = \hatt_{i} \cup \mathrm{argmin} \left\{ \clcost(\hatt_{i} \cup B) - w(\hatt_{i} \cup B) ~\Big|~B \subseteq \bigcup_{j = i}^{\eta}D_r^i\right\}.\]
By the optimality of $C^*_{i + 1}$, we know that
$\Delta(C^*_{i+1}) \leq \Delta(Q_i)$.
To bound $\Delta( C^*_{i+1}) - \Delta( C^*_i)$, it is sufficient to prove that 
\[\Delta( Q_{i}) - \Delta( C^*_i) = O( \eta^{-2} \eps^{-8} \dadm(C) ).\]
Note that 
\begin{align*}
 \Delta(Q_{i}) - \Delta( C^*_i) &= 
 \clcost(Q_i) - w(Q_i) - \clcost(C^*_i) + w(C^*_i) \\
 &= \clcost(Q_i) - \clcost(C^*_i) - w(Q_i \setminus C^*_i) + w(C^*_i \setminus Q_i) 
\end{align*}

Intuitively, removing one vertex from $Q_i$ will incur $\marginal(Q_i, v)$ to $\clcost(Q_i) - \clcost(C^*_i)$, so we can use $\marginal$ value to bound $\clcost(Q_i) - \clcost(C^*_i)$. More precisely, consider the process that we change $Q_i$ to $C^*_i$. we first remove nodes from $Q_i$ one by one, at the end, we get the set $Q_i \cap C^*_i$, then we try to add nodes to this set and make the final set $C^*_i$, $\clcost(Q_i) - \clcost(C^*_i)$ is bounded by the marginal in each step. 
Let's consider some moment in this process, we first consider the removal process. assume that at step $j$, our set is $Q_{i,j}$, then we choose an arbitrary element from $v \in Q_{i, j} \setminus C^*_i$, and remove $v$ from $Q_{i, j}$, the new set is $Q_{i, j + 1}$. At the beginning, we have $Q_{i, 1} = Q_i$. By Claim \ref{claim:marginalchangesinsmallchunk}, we know that 
\begin{align*}
    \marginal(Q_{i, j}, v) - \marginal(C^*_i, v)\leq 2|C^*_i \oplus Q_{i, j}| \leq 2|C^*_i \oplus Q_i|
\end{align*}
so the cost difference to change $Q_{i, j}$ to $Q_{i, j+1}$ is at most
\begin{align*}
    \clcost(Q_{i, j}) - \clcost(Q_{i, j+1}) = \marginal(Q_{i, j}, v) \leq \marginal(C^*_i, v) + 2|Q_{i} \oplus C^*_i|.
\end{align*}

Similarly, in the process of changing $C^*_i \cup Q_i$ to $C^*_i$, the cost change in each step is bounded by $\marginal(C^*_i, v) - 2|Q_{i} \oplus C^*_i|$. So, in order to change $Q_i$ to $C^*_i$, the total marginal changes 
\begin{align*}
\clcost(Q_i) - \clcost(C^*_i) &\leq  \sum_{v \in Q_{i} \setminus C^*_i}( \marginal(C^*_i, v) + 2|Q_{i} \oplus C^*_i|) - \sum_{v \in C^*_i \setminus Q_{i}}( \marginal(C^*_i, v) - 2|C^*_i \oplus Q_{i} |) \\ 
&\leq  \sum_{v \in Q_{i} \setminus C^*_i} \marginal(C^*_i, v)- \sum_{v \in C^*_i \setminus Q_{i}}\marginal(C^*_i, v) + 2 |Q_{i} \oplus C^*_i|^2  \\
\end{align*}

$\Delta(Q_{i}) - \Delta( C^*_i) $ is bounded by
\begin{align*}
 \Delta(Q_{i}) - \Delta( C^*_i) 
 &= \clcost(Q_i) - \clcost(C^*_i) - w(Q_i \setminus C^*_i) + w(C^*_i \setminus Q_i) \\
 &\leq \sum_{v \in Q_{i} \setminus C^*_i} \left(\marginal(C^*_i, v) - w(v) \right)- \sum_{v \in C^*_i \setminus Q_{i}} \left( \marginal(C^*_i, v) - w(v) \right)+ 2|Q_{i} \oplus C^*_i|^2.
\end{align*}
Now we will bound these three parts one by one. 

\paragraph{Bound for $|Q_{i} \oplus C^*_i|^2$.} We start by bounding $|Q_{i} \oplus C^*_i|^2$. The only difference between $C^*_i$ and $Q_i$ is the vertices in $D^{i+1}_r$, the algorithm replace the elements in $C^{*}_i \cap D^{i+1}_r$ with $\hatt_{i+1} \setminus \hatt_i$. So
\begin{align*}
    |Q_{i} \oplus C^*_i|^2 \leq |D^{i+1}_r| ^ 2 \leq \eta^{-2} |D(r)|^2.
\end{align*}

\paragraph{Bound for $\sum_{v \in Q_{i} \setminus C^*_i} \left(\marginal(C^*_i, v) - w(v) \right)$.}
By definition of $C^*_i$, for any 
vertex $v \in D^{i+1}_r \setminus C^*_i$, adding $v$ to $C^*_i$ does not yield 
a strictly better cluster. In particular, for any $v \in D^{i+1}_r \setminus C^*_i$, we have $\marginal(C^*_i,v) - w(v) \geq 0$. 
The difficulty of the proof is that $\marginal(C^*_i,v) - w(v)$ could be very large.  Let
\begin{align*}
    \ell(v) = \frac{\marginal(C^*_i,v) - w(v)}{\eta^{-1}|\Ncand(r)|}
\end{align*}
be the ratio of contribution, so if we add $v$ to $Q_i$, $v$ contributes $\ell(v) \eta^{-1}|\Ncand(r)|$ to $\Delta(Q_i) - \Delta(C^*_i)$.

We need to bound the estimation error between \( \rmEstCost \) and \( \marginal \). We provide the following lemma regarding this bound, and its proof is given at the end of this section.

\begin{lemma}
\label{lem:concentration-weighted}
Let $\eta_0 = \eta^3$ and $\ell \geq 1$. Consider a vertex $v$ and
an arbitrary set of vertices $T$ of size $t = \Omega(\eta)$. Let $S$ be a set consisting of $\eta_0$ random samples (with repetition) from $T$.
Let $\tilt \in [t,(1+\frac{1}{\eta})t]$
Then, with probability at least  $1 - 2\exp\left(-2\ell^2\eta\right)$. we have the following
inequality holds:
\begin{equation}
    \label{eq:coststay-concentration}
    \marginal(T,v) - \frac{(4+\ell)t}{\eta} \leq \rmEstCost(S,\tilt,v) \leq \marginal(T,v) + \frac{(4+\ell)t}{\eta}
\end{equation}
\end{lemma}

According to Lemma~\ref{lem:concentration-weighted}, with probability at least  
$1 - 2\exp\left(-2(\frac{\ell(v)}{2} + 1)^2\eta\right)$, we have
\[
\rmEstCost(S^i, \tilt_i, v) - \marginal(C^*,v) \geq -(4 + 2(\frac{\ell(v)}{2} + 1))\eta^{-1} |C^*_i| \geq -(6+\ell(v))\eta^{-1} |\Ncand(r)|,
\]
and 
\[
\rmEstCost(S^i, \tilt_i, v) + 6\eta^{-1} |\Ncand(r)| \geq \marginal(C^*,v) - \ell \eta^{-1} |\Ncand(r)| \ge w(v),
\]
and Algorithm~\ref{alg:GenerateCluster} will not add $v$ to $\hatt$.
Therefore $v \in Q_i$ with probability at most $2\exp\left(-2(\frac{\ell(v)}{2} + 1)^2\eta\right)$, and the expected contribution of $v$ to $\sum_{v \in Q_{i} \setminus C^*_i} \left(\marginal(C^*_i, v) - w(v) \right)$ is at most 
\begin{align*}
&2\exp\left(-2(\frac{\ell(v)}{2} + 1)^2\eta\right) \cdot \left(\marginal(C^*_i, v) - w(v) \right) \\ 
&\leq 2\ell(v) \exp\left(-2(\frac{\ell(v)}{2} + 1)^2\eta\right) \eta^{-1} |\Ncand(r)| \\
&\leq 2\exp(-2\eta) \eta^{-1} |\Ncand(r)|
\end{align*}

and 

\begin{align*}
    \EX[\sum_{v \in Q_{i} \setminus C^*_i} \left(\marginal(C^\star_i, v) - w(v) \right) ]&\leq 2 \exp(-2\eta) \eta^{-1} |\Ncand(r)| \eta^{-1}|D(r)|.
\end{align*}
By Markov inequality, with probability at most $\exp(-2\eta)$, we have
\begin{align*}
    \sum_{v \in Q_{i} \setminus C^*_i} \left(\marginal(C^\star_i, v) - w(v) \right) \geq 2 \eta^{-2} |\Ncand(r)||D(r)|.
\end{align*}

\paragraph{Bound for $\sum_{v \in C^*_i \setminus Q_i} \left(\marginal(C^\star_{i}, v) - w(v) \right)$.}
If $v\in C^*_i \setminus Q_i$, we know that $\marginal(C^*_i,v) - w(v) \geq 0$.
We will distinguish two cases.
\begin{enumerate}
    \item $\marginal(C^\star_i, v) - w(v) \leq  -12\eta^{-1} |\Ncand(r)|$ \newline
    Again, we use the same strategy of bounding $\sum_{v \in Q_{i} \setminus C^*_i} \left(\marginal(C^*_i, v) - w(v) \right)$. 
    Let
\begin{align*}
    \ell(v) = \frac{w(v) - \marginal(C^*_i,v)}{\eta^{-1}|\Ncand(r)|}
\end{align*}
be the ratio of contribution, so if we decide not to add $v$ to $Q_i$, $v$ contributes $\ell(v) \eta^{-1}|\Ncand(r)|$ to $\Delta(Q_i) - \Delta(C^*_i)$. Note that $\ell(v) \geq 12$.
According to Lemma~\ref{lem:concentration-weighted}, with probability at least  
$1 - 2\exp\left(-2(\frac{\ell(v)}{2} - 5)^2\eta\right)$, we have
\[
\rmEstCost(S^i, \tilt_i, v) - \marginal(C^*,v) \leq (4 + 2(\frac{\ell(v)}{2} - 5))\eta^{-1} |C^*_i| \leq (\ell(v) - 6)\eta^{-1} |\Ncand(r)|.
\]
and 

\[
\rmEstCost(S^i, \tilt_i, v) + 6\eta^{-1} |\Ncand(r)| \leq \marginal(C^*,v) + \ell \eta^{-1} |\Ncand(r)| \leq w(v),
\]
and Algorithm~\ref{alg:GenerateCluster} will add $v$ to $\hatt$.
Therefore $v \not\in Q_i$ with probability at most $2\exp\left(-2(\frac{\ell(v)}{2} - 5)^2\eta\right)$, and the expected contribution of $v$ to $\sum_{v \in Q_{i} \setminus C^*_i} \left(w(v) - \marginal(C^*_i, v) \right)$ is at most 
\begin{align*}
&2\exp\left(-2(\frac{\ell(v)}{2} - 5)^2\eta\right) \cdot \left(w(v) - \marginal(C^*_i, v) \right) \\ 
&\leq 2\ell(v) \exp\left(-2(\frac{\ell(v)}{2} - 5)^2\eta\right) \eta^{-1} |\Ncand(r)| \\
&\leq 24\exp(-2\eta) \eta^{-1} |\Ncand(r)|
\end{align*}

and 

\begin{align*}
    \EX[\sum_{v \in Q_{i} \setminus C^*_i} \left(w(v) - \marginal(C^*_i, v) \right)]&\leq 24 \exp(-2\eta) \eta^{-1} |\Ncand(r)| \eta^{-1}|D(r)|.
\end{align*}
By Markov inequality, with probability at most $\exp(-2\eta)$, we have
\begin{align*}
    \sum_{v \in Q_{i} \setminus C^*_i} \left(w(v) - \marginal(C^*_i, v) \right) \geq 24 \eta^{-2} |\Ncand(r)| \eta^{-1}|D(r)|.
\end{align*}
With probability at least $1 -\exp(-2\eta) $, we have 
\begin{align*}
    \sum_{v \in Q_{i} \setminus C^*_i} \left(\marginal(C^*_i, v) - w(v) \right) \geq - 24 \eta^{-2} |\Ncand(r)| \eta^{-1}|D(r)|
\end{align*}

    \item $ \marginal(C^\star_i, v) - w(v) \ge -12\eta^{-1} |\Ncand(r)|.$  Recall that $C^*_i \setminus Q_i \subseteq D^{i+1}_r$. We have,
    \[
        \sum_{v \in Q_{i} \setminus C^*_i} \left(\marginal(Q_{i}, v) - w(v) \right) \geq -12\eta^{-1} |\Ncand(r)| \cdot |D^{i+1}_r| = -12\eta^{-1} |\Ncand(r)| \cdot \eta^{-1}|D(r)|.
    \]
\end{enumerate}
Now, we are ready to give the final bound on $\Delta(Q_{i}) - \Delta( C^*_i)$.  Using the bounds from the three cases above and we have, with probability at least $1 - 2\exp(-2\eta) $, we have 
\begin{align*}
     \Delta(Q_{i}) - \Delta( C^*_i) 
     \leq&  \sum_{v \in Q_{i} \setminus C^*_i} \left(\marginal(Q_{i}, v) - w(v) \right)- \sum_{v \in C^*_i \setminus Q_{i}} \left( \marginal(C^*_i, v) - w(v) \right)+ |Q_{i} \oplus C^*_i|^2 \\
     \leq & 2 \eta^{-2} |\Ncand| \cdot |D(r)| + 24 \eta^{-2} |\Ncand| \cdot |D(r)| + 12 \eta^{-2} |\Ncand| \cdot |D(r)| + 2 \eta^{-2} |D(r)|^2 \\
     =& O(\eta^{-2} |\Ncand|\cdot |D(r)|)
\end{align*}
We still have to bound $|\Ncand|\cdot |D(r)|$, which actually is the number of two hop candidates. Fortunately, after preclustering, the number of two hop candidates $|D(v)| \cdot |\Ncand(v)|$ can be bounded by following lemma, which is also used in~\cite{cohen2024combinatorial}.
\begin{lemma}[Lemma 34 of~\cite{cohen2024combinatorial}]
\label{lem:boundtwohopsneighbor}
For any vertex $r$ and any \epslarge cluster $C$ such that $K(r) \subseteq C \subseteq \Ncand(r)$, we have  
\begin{align*}
    |D(r)| \cdot |\Ncand(r)| = O( \eps^{-8} \dadm(C)). 
\end{align*}
\end{lemma}
Using Lemma \ref{lem:boundtwohopsneighbor}, we obtain the claimed bound,
\[
    \Delta(Q_{i}) - \Delta( C^*_i) = O(\eta^{-2} |\Ncand|\cdot |D(r)|) = O(\eta^{-2}\eps^{-8} \dadm(C)). 
\]
By the union bound, all bounds hold simultaneously with probability at least $1 - 2\exp\left(-2\eta\right)$.

\end{proof}

\begin{proof}[Proof of Lemma~\ref{lem:concentration-weighted}]
Recall that $\tilt$ is the guess for size of $T$.
\begin{align*} 
    \marginal(T, v) = \frac{d^+(v) - 1}{2} + |T|  - 2d^+(v,T) + \mathbb{1}(v \in T). 
\end{align*}

and 
\[\rmEstCost(S,\tilt,v) := 
\frac{\dpos(v) - 1}{2} + \tilt  - 2\frac{\dpos(v, S)}{|S|} \tilt\] 

Consider the $i$-th sampled node $u$ in $S$, let $X_i = \frac{t}{\eta_0} \cdot \mathbb{1}(uv \in E)$ be a random variable, so $X_i \in \{0, \frac{t}{\eta_0}\}$.  
\begin{align*}
    \EX [~\sum_{i = 1}^{\eta_0} X_i~] = \frac{\eta_0}{|T|}  \cdot \frac{t}{\eta_0} \cdot \dpos(v, T) = \dpos(v, T).
\end{align*}
By  Hoeffding’s inequality, we have,
\begin{align*}
    \Pr[\left|~\sum_{i = 1}^{\eta_0}(X_i-\EX[X_i])~\right| \geq \frac{\ell t}{\eta}] \leq 2 \exp \left( -\frac{2(\ell t / \eta)^2}{\eta_0 \cdot (t / \eta_0)^2} \right) \leq 2\exp \left( -2 \ell^2 \eta \right).
\end{align*}

Now, consider the difference of $\marginal$ and $\rmEstCost$, we have 
\begin{align*}
&\quad \left|\rmEstCost(S,\tilt,v) - \marginal(T, v)~\right| \\ 
&\leq \left|\tilt - 2\frac{\dpos(v, S)}{|S|}\tilt - |T|  + 2d^+(v,T) - \mathbb{1}(v \in T)~\right| \\
&\leq \left|\tilt - |T|\right| + 2\frac{\dpos(v, S) }{|S|}|\tilt - t| + 2\left|~\sum_{i = 1}^{\eta_0}(X_i-\EX[X_i])~\right| + 1\\
&\leq \frac{t}{\eta} + \frac{2t}{\eta} + \frac{2\ell t}{\eta}  + 1\\
&\leq \frac{(4 + 2\ell) t}{\eta}.
\end{align*}
The last inequality is because $t = \Omega(\eta)$.

\end{proof}

%% file: nearlylineartime.tex
In this section, we present an algorithm to solve~\ref{LP:clusterlp2} in nearly linear time (i.e., \(\widetilde{O}(m)\)), where \(m\) is the number of \(\text{\pedges}\). 
In a later section, we will present a sublinear-time algorithm. The main theorem regarding the nearly linear algorithm is stated as follows.

\begin{restatable}[Nearly Linear Time \ref{LP:clusterlp}]{theorem}{thmsolveclusterLPsequentialinlineartime} \label{thm:solving-cluster-LP-sequential-nearlyLinear}
Let $\eps> 0$ be a sufficiently small constant and let $\cost(\opt)$ be the cost of the optimum solution to the given \cc instance. Then there is a small $\delta = \poly(\epsilon)$ such that the following statement holds: One can output a solution $(z_S)_{S \subseteq V}$ to the cluster LP, described using a list of non-zero coordinates, with $\obj(x) \leq (1+\eps)\cost(\opt)$ in expectation such that each coordinate of $z$ is either $0$ or at least $\delta$. The running time to compute $z$ is $\widetilde{O}(2^{\poly(1/\eps )}m)$.
\end{restatable}

We need to address two problems when aiming for nearly linear time. First, we cannot afford to compute all the clusters $C_r$, one for each vertex $r \in V$. Instead, we will sample a subset of vertices $U \subseteq V$ such that $U$ is not too large and we can compute a cluster $C_r$ for each vertex $r \in U$. In particular, we will include each vertex $v$ in $U$ with probability $\log(n)/d(v)$.
Since we assume that the optimal solution is \epslarge, we will hit every cluster $C \in \opt$ with a vertex (i.e., $U \cap C \neq \emptyset$) with high probability. Because of this, there will still be a cluster among $\{C_r\}_{r \in U}$ that achieves the ratio $R$. 

Second, we cannot update a cluster $C_r$ as soon as it loses one vertex. However, we can wait until a cluster $C_r$ has lost a constant fraction $\covereps \cdot p(C_r)$ of its probability mass before we update it. If a cluster $C_r$ did not lose this constant fraction, the ratio did not change by too much. Moreover, we can show that if the probability distribution $p$ is a distribution from the MWU Algorithm \ref{alg:mw}, then $\Ncand(r)$ has to lose at least a constant fraction of vertices in order for $C_r$ to lose a constant fraction of probability mass. Thus, we only need to update a cluster $C_r$ at most a constant number of times. 


\begin{algorithm}[h]
    \caption{(Nearly Linear) Algorithm to find the family $\mathcal{F}$}
	\label{alg:disjointfamily}
	\begin{algorithmic}[1]
\STATE Let $R$ be the guess for $\covers(\opt)$ such that $R \in [\covers(\opt) , (1 +\covereps) \covers(\opt))$.
    \STATE $\hat p \gets p, \mathcal{F} \gets \emptyset$
    \FOR{$t=1,\dots,\log(n)/\eps^2$}
        \STATE Add each vertex $v$ with probability $\frac{1}{d(v)}$ to $U$. 
    \ENDFOR
    \FORALL{$v \in V$}
     \STATE If $\frac{\covers(K(v))}{p(K(v))} \leq (1+6\covereps)~R$, add $K(v)$ to $\mathcal{F}$ and set $\hat p_w = 0$ for all $w \in K(v)$. \label{line:addlargepnodes}
    \STATE If $p_v \leq \frac{\covereps \dc(v)}{4\dc(V)}$, set $\hat p_w = 0$ for all $w \in K(v)$. \label{line:removesmallpnodes}
    \ENDFOR
    \FORALL{$u \in U$} \label{line:loop}
        \STATE Find a small ratio cluster $C_u$ such that $K(u) \subseteq C_u \subseteq \Nadm(u)$ with vertex weights $\hat p > 0$ and target ratio $(1+3\covereps)R$ (Lemma \ref{lem:goodratio}).
    \ENDFOR
    \WHILE{$ p(\mathcal{F}) \leq \covereps$} \label{line:loop_family}
        \STATE Choose $C$ with the smallest ratio $\frac{\covers(C)}{\hat p(C)}$ among clusters $\{C_v\}_{v \in U}$.
        \STATE Remove all $w$ such that $\hat p_w = 0$ from $C$, add the new $C$ to $\mathcal{F}$.
        \STATE Set $\hat p_v$ to $0$ for all $v \in C$.
        \FORALL{$v \in U$}
            \IF{$\hat p(C_v) \leq (1-\covereps) p(C_v)$} 
            \STATE Find a new small ratio cluster $C_v$  with vertex weights $\hat p > 0$ and target ratio $(1+3\covereps)R$ (Lemma \ref{lem:goodratio}).
            \ENDIF
        \ENDFOR
    \ENDWHILE
    \RETURN $\mathcal{F}$
	\end{algorithmic}
\end{algorithm}

\subsection{Approximate Ratio of Algorithm~\ref{alg:disjointfamily}}

\begin{lemma}
\label{lem:disjointfamilyapproximateratio}
Given vertex weights $p_v >0$, Algorithm \ref{alg:disjointfamily} finds a family $\mathcal{F} = \{S_1, S_2,..., S_l \mid S_i \subseteq V \}$ such that,
\begin{enumerate}
    \item for any distinct $S, T \in \mathcal{F}$, $S \cap T = \emptyset$,
    \item $\frac{\covers(\mathcal{F})}{p(\mathcal{F})} \leq (1+8\covereps)\covers(\opt)$,
    \item $p(\mathcal{F})$ is at least $\covereps$,
    \item no $S \in \mathcal{F}$ splits an atom, i.e. $K(v) \subseteq S$ for all vertices $v \in S$.
\end{enumerate}
\end{lemma}
\begin{proof}
First, observe that Property 3. holds by the condition in the while loop. Property 2. holds since all small ratio clusters $\{C_v\}_{v\in U}$ do not split atoms.  Next, we will prove the family $\mathcal{F}$ consists of disjoint sets. Observe that after we added a cluster $C \subseteq V$ to $\mathcal{F}$ we set the weight $\hat p_v$ to $0$ for all $v \in C$. Note that if the weight $\hat p_v$ is set to $0$ it remains $0$ throughout the execution of the algorithm. 
 Moreover, before we add a cluster $C \subseteq V$ to $\mathcal{F}$ we remove all vertices $v \in C$ with weight $p_v = 0$ from $C$. Thus, when we add $C$ to $\mathcal{F}$ we have $\hat p_v > 0$ for all $v \in C$ and $\hat p_v = 0$ for all $v \in S, S \in \mathcal{F}$. 
Assume that the guess $R$ satisfies $\covers(\opt) \leq R \leq (1+\covereps)\covers(\opt)$. It remains to prove the bound on the ratio $\frac{\covers(\mathcal{F})}{p(\mathcal{F})}$. To do so, we will make use of the following claim.
\begin{claim}
    \label{clm:hitopt}
    For each cluster $C \in \opt$ with $|C| \geq 2$, $C \cap U \neq \emptyset$ with high probability. Moreover, if there exists a non-singleton atom $K \subseteq C$ then $K \cap U \neq \emptyset$ with high probability. 
\end{claim}
\begin{cproof}
    We can assume that the optimal clustering is \epslarge with respect to our preclustering. Thus, for any non-singleton atom $K$ we have $|K| \geq O(\epsilon) d(v)$ for each $v \in K$. 
    Otherwise, $v$ has more than $O(\epsilon |K|)$ neighbors outside of $K$.
    The probability that we don't include a vertex from $K$ in $U$ in one iteration is
    \[
        \prod_{v \in K} \left(1 - \frac{1}{d(v)} \right) \leq \exp \left (-\sum_{v \in K } \frac{1}{d(v)} \right ) \leq \exp\left(-O(\epsilon)\right).
    \]
    
    Thus, after $\log(n)/\epsilon^2$ rounds, we include a vertex from $K$ in $U$ with high probability.
    If $|C| > 1$, then $|C| \geq \epsilon d(v)$. By the same arguments as above, we will include a vertex from $C$ in $U$ with high probability. 
\end{cproof}
Let \(C \in \opt\). By Claim~\ref{clm:hitopt}, there exists a vertex \(v \in C \cap U\) such that \(K = K(v)\) if \(K\) is a non-singleton atom in \(C\). Now, we are able to analyze the approximation ratio of Algorithm~\ref{alg:disjointfamily}. We first show the existence of a cluster that satisfies the conditions of Lemma~\ref{lem:goodratio}.

\begin{lemma}
\label{lem:goodratioclusterexists}
In the algorithm~\ref{alg:disjointfamily}, if $p(\mathcal F) \leq \covereps$, then there always be a cluster $C^* \in \opt$ such that 
\begin{align*}
    \frac{\covers(C^\star)+\covereps^2 \dadm(C^\star)}{\hat p(C^\star)} \leq (1 + 3\covereps) R.
\end{align*}
\end{lemma}
\begin{proof}
Let $C^\star$ be the cluster among clusters in $\opt$ with the smallest ratio $\frac{\covers(C^\star)}{\hat p(C^\star)}$. Notice that the algorithm will set some $\hat p_v$ to $0$ without adding $u$ to $\mathcal{F}$, let 
\[
    \mathrm{Small} = \{ v \in \hat p_v \text{ is set to  0 at line~\ref{line:removesmallpnodes} in Algorithm~\ref{alg:disjointfamily}} \}
\]
We have $ p(\mathrm{Small}) \leq \sum_{v \in V} \frac{\covereps \dc(v)}{16\dc(V)} \leq \covereps / 4$. By the condition in loop \ref{line:loop_family}, we have $\hat p(V) = p(V) - p(\mathcal{F}) - p(\mathrm{Small}) \geq 1 - \covereps - \covereps/4 \geq 1 - 1.5\covereps$.
Thus,
\begin{align*}
    (1+3\covereps)R \geq \frac{1+ \covereps}{1-1.5\covereps} R \geq \frac{(1+\covereps)\covers(\opt)}{\hat p(V)} \geq \frac{\covers(\opt) + \covereps^2 |E_{\adm}|}{\hat p(V)} \\= \frac{\sum_{C \in \opt}(\covers(C) + \covereps^2 \dadm(C))}{\sum_{C \in \opt} \hat p(C)} \geq \frac{\covers(C^\star)+\covereps^2 \dadm(C^\star)}{\hat p(C^\star)}.
\end{align*}
For the second inequality, we use that $\gamma~\covers(\opt) \geq \gamma~\clcost(\opt) \geq \gamma^2 |E_{\adm}|$ because of the preclustering (Theorem \ref{thm:preclustering-proc}).
The last inequality holds by the definition of $C^\star$.
\end{proof}

Based on Lemma~\ref{lem:goodratioclusterexists}, we know that if we use \((1+3\covereps)R\) as the ratio and apply Lemma~\ref{lem:goodratio}, we can always find a cluster \(C_v\) with ratio at most \((1+3\covereps)R\). One issue remains: we might set the \(\hat{p}\) value of some nodes in \(C_v\) to 0. This is captured by the following invariant.

\begin{lemma}
\label{lem:clustercvratio}
For any \(C_v\) added to \(\mathcal{F}\) in Algorithm~\ref{alg:disjointfamily}, we have 
\begin{align*}
    \frac{\covers(C_v)}{\hat{p}(C_v)} \leq (1+5\covereps) R.
\end{align*}
\end{lemma}

\begin{proof}
When \(C_v\) is created, we have 
\begin{align*}
    \frac{\covers(C_v)}{\hat{p}(C_v)} = \frac{\covers(C_v)}{p(C_v)} \leq (1+3\covereps) R.
\end{align*}
Note that at this moment, \(C_v\) only contains nodes with \(\hat{p}\) values greater than 0. Later, the algorithm may add other clusters to \(\mathcal{F}\) and set some \(\hat{p}\) values to 0, which will increase the ratio of \(C_v\). However, whenever at least a \(\covereps\)-fraction of the \(\hat{p}\) value in \(C_v\) is decreased to 0, we renew \(C_v\). Thus, we always maintain the bound:
\begin{align*}
    \frac{\covers(C_v)}{\hat{p}(C_v)} \leq \frac{\covers(C_v)}{(1 - \covereps)p(C_v)} \leq \frac{1+3\covereps}{1-\covereps} R \leq (1+5\covereps) R.
\end{align*}
\end{proof}

Now, we only need to show that whenever \(p(\mathcal{F}) \leq \covereps\), there always exists a non-empty cluster \(C\) among the clusters \(\{ C_v \}_{v \in U}\) that we can consider. We prove this by contradiction. 

Assume that at some iteration of Line~\ref{line:loop_family}, we have \(p(\mathcal{F}) \leq \covereps\) and \(C_v = \emptyset\) for all \(v \in U\). By Lemma~\ref{lem:goodratioclusterexists}, there exists a cluster \(C^*\) such that
\[
\covers(C^*) + \covereps^2 \dadm(C^*) \leq (1 + 3\covereps) R \hat{p}(C^*).
\]
By Claim~\ref{clm:hitopt}, there exists a node \(u \in U\) such that \(u \in C^*\). Moreover, if \(C^*\) contains a non-singleton atom, then \(u\) is in the non-singleton atom. Since \(\hat{p}\) is non-increasing, we can always find a cluster \(C_u\) such that the ratio is at most \((1 + 3\covereps) R\).

Remember that we might add single vertices with ratio $(1+6\covereps)R$. We can conclude that,
\begin{align*}
    \frac{\covers(\mathcal{F})}{p(\mathcal{F})} = \frac{\sum_{S \in \mathcal{F}}\covers(S)}{\sum_{S \in \mathcal{F}}p(S)} \leq \frac{\covers(C_v)}{\hat p(C_v)} \leq (1+6\covereps) R \leq (1+8\covereps) ~\covers(\opt)
\end{align*}

\end{proof}

\subsection{Runtime of Algorithm~\ref{alg:disjointfamily}}
We now prove that Algorithm~\ref{alg:disjointfamily} runs in nearly linear time.

\begin{lemma}
\label{lem:disjointfamilylinear}
    For each vertex $v$, let $p_v$ be the normalized weight $w^{(t)}_v$ during a round $t = 1, \dots, T$ of the MWU Algorithm \ref{alg:mw} for $T = \frac{-\log(\covereps)}{\covereps^2}$.
    Then, the expected runtime of Algorithm \ref{alg:disjointfamily} is $\widetilde O(m).$
\end{lemma}
\begin{proof}
We add each vertex to $U$ with probability at most $\frac{1}{d(v)} \log(n)$. 
By Lemma \ref{lem:goodratio}, computing one $C_u$ for $u \in U$ takes time $\widetilde O(d^2(u))$. Computing all $\{C_u\}_{u \in U}$ takes expected time,
\[
    \sum_{v \in V} \frac{\widetilde O(d^2(v))}{d(v)} \log(n) = \widetilde O(m).
\]
We maintain a priority queue of the clusters $\{C_u\}_{u \in U}$. We can pick the cluster $C$ with the best ratio and remove it from the queue in time $\log(n)$. If we update a $C_u$, we can insert it into the queue in time $\log(n)$. When we set the weight of a vertex to zero we have to remove this vertex from each cluster $C_u$. In the following, we will show that $v$ is contained in at most $\log(n)$ clusters $C_u$ in expectation. Thus, this takes time at most $\log(n)$. Moreover, we set the weight of a vertex to zero at most once.
\begin{claim}
\label{clm:overlap}
With high probability, each vertex is contained in at most $O(\log(n))$ clusters from $\{C_u\}_{u \in U}$.
\end{claim}
\begin{cproof}
    If a vertex $v$ is included in the cluster $C_u$, the $v$ and $u$ have to be connected by an admissible edge. Thus the number of clusters $C_u$ with $v \in C_u$ is bounded by the cardinality of $|U \cap \Nadm(v)|$. Since the preclustering is \epssimilar, we have that $|\Nadm(v)| = \dadm(v) \leq 2\epsilon^{-3} d(v)$ and $d(u) \geq \frac{1}{2} \epsilon d(v)$ for each vertex $u \in \Nadm(v)$. Thus, we have for the expectation,
    \begin{align*}
        \EX[|U \cap \Nadm(v)|] = \sum_{u \in \Nadm(v)} \frac{1}{d(u)} \log(n) \leq 2\epsilon^{-1} \sum_{u \in \Nadm(v)} \frac{1}{d(v)} \log(n) \leq 4 \epsilon^{-4} \log(n).
    \end{align*}
    The high-probability argument follows from a standard Chernoff bound and the union bound.
\end{cproof}
It is left to show that we don't have to update the $C_u$'s too often. To do so, we first proof that each $C_u$ has size proportional to the degree $d(u)$.
\begin{claim}
    \label{clm:size}
    For any cluster $C_u \in \Ncand(u)$ such that 
    $\frac{\covers(C_u)}{\hat p(C_u)} \leq (1+5\covereps) R$, then 
    \begin{itemize}
        \item either $|C_u \cap \{ v \in C_u \mid \hat p_v > 0 \}| \geq \covereps^3 d(u)$,
        \item or $\frac{\covers(\{v\})}{ p_v} \leq (1+6\covereps)R$ for some $v \in C_u$.
    \end{itemize}
\end{claim}
\begin{cproof}
    Note that $\covers$ is monotonic, so we can always remove all vertices with $\hat p = 0$ at the beginning and this will only decrease the ratio. We will show that if a cluster $C_u$ has size $|C_u \cap \{ v \in C_u \mid \hat p_v > 0 \}| \leq \covereps^3 d(u)$ then there exists a vertex $v \in C_u$ with
    \[
     \frac{\covers(\{v\})}{ p_v} \leq (1+6\covereps)R
    \]
    Remember that the edge $(u,v)$ is admissible for all $v \in C_u$. If $C_u$ contains a non-singleton atom $K(u) \subseteq C_u$, note that $|K(u)| \geq d(u)/2 \geq |C_u|$, this  contradicts to $K(v) \subseteq C_u$. so we can assume that $C_u$ does not contain any non-singleton atom and
    \begin{align*}
        \covers(C_u) & \geq \sum_{v \in C_u} \left( \frac{1}{2} (d(v) - |C_u|) + \dc(v) \right)\\
        \geq & \sum_{v \in C_u} \left( \frac{1}{2}(d(v) - \covereps^2 d(v)) + \dc(v) \right) \\
        & \geq (1-\covereps^2) \sum_{v \in C_u} \covers(\{v\}).
    \end{align*}
    The second inequality holds since we assumed that $|C_u| \leq \covereps^3 d(u) \leq \covereps^3 \epsilon^{-1} d(v) \leq \covereps^2 d(v)$. Here, we used that the preclustering is \epssimilar and $(u,v)$ is admissible.
    If $C_u$ is considered in the algorithm, then the ratio of $C_u$ is bounded by $(1+5\covereps)R$. One can safely remove $u$ from $U$ otherwise.
    \[
        (1+5\covereps) R \geq \frac{\covers(C_u)}{\hat p(C_u)} \geq (1-\covereps^2) \frac{\sum_{v \in C_u} \covers(\{v\})}{\sum_{v \in C_u} \hat p_v}.
    \]
    However, than there has to exist a vertex $v \in C_u$ that achieves the ratio 
    \[
         \frac{1+5\covereps}{1-\covereps^2} R \leq (1+6\covereps)R
    \]
     This vertex would have been added to $\mathcal{F}$ before $\{C_u\}_{u \in U}$ are computed. Then, $\hat p_v = 0$, a contradiction.
\end{cproof}

The next lemma will give us an upper and lower bound of $p$ value each round, this can help us bound the number of updates for $C_u$.
\begin{lemma}
\label{lemma:pvaluehasarange}
Invariant: In Algorithm~\ref{alg:mw}, for any $t \in [1, \tmwu]$ and any node $u \in V$, we have 
\[
\frac{\covereps \dc(u)}{16\dc(V)} \leq p^{(t)}_u \leq \frac{16 \dc(u)}{\dc(V)}
\]
\end{lemma}
\begin{proof}
We prove the statement by induction. The base case \( t = 1 \) holds automatically. 

Now, assume that the inequality holds at round \( t \). First, note that if a node \( u \) is not added to \( \mathcal{F} \), then \( m^{(t)}_u = -1 \); otherwise, we have \( 1 \leq m^{(t)}_u \leq 1/\covereps \). Therefore, the value \( w^{(t)}_u \) will increase by at most \( e^{-\covereps^3 m^{(t)}_u} = e^{\covereps^3} \leq 2 \) if \( u \) is not added to \( \mathcal{F} \), and decrease by at most \( e^{-\covereps^2} \geq 1/2 \) if \( u \) is added to \( \mathcal{F} \). Consequently, we obtain:
\[
p^{(t)}_u \leq p^{(t+1)}_u \leq 4 p^{(t)}_u, \quad \text{if } u \text{ is not added to } \mathcal{F},
\]
and
\[
\frac{p^{(t)}_u}{4} \leq p^{(t+1)}_u \leq p^{(t)}_u, \quad \text{if } u \text{ is added to } \mathcal{F}.
\]

It remains to show that: \( u \) is never added to \( \mathcal{F} \) if \( p^{(t)}_u \leq \frac{\covereps \dc(u)}{4\dc(V)} \), and \( u \) is always added to \( \mathcal{F} \) if \( p^{(t)}_u \geq \frac{ 4 \dc(u)}{\dc(V)} \). For the first argument, by Line~\ref{line:removesmallpnodes}, we always set \( \hat{p}_u = 0 \) if 
\( p^{(t)}_u \leq \frac{\covereps \dc(u)}{4\dc(V)} \), and we never add a node with \( \hat{p}_u = 0 \).

The second argument is slightly more involved. A node \( u \) is added to \( \mathcal{F} \) at Line~\ref{line:addlargepnodes}. Note that \( R \geq \covers(\opt) \geq \dc(V) \). Whenever \( p^{(t)}_u \geq \frac{4 \dc(u)}{\dc(V)} \), we have 
\begin{align*}
\frac{\covers(K(u))}{p(K(u))} &\leq \frac{\dc(K(u))}{4 |K(u)| \cdot \dc(u) / \dc(V)} \\
&\leq \frac{|K(u)| \cdot \dc(u)}{4 |K(u)| \cdot \dc(u) / \dc(V)} \leq \frac{\dc(V)}{2} \leq R.
\end{align*}
Thus, we will add \( K(u) \) to \( \mathcal{F} \), which completes the proof.
\end{proof}

\nairen{check it}

Now, using Lemma~\ref{lemma:pvaluehasarange}, we can give a lower bound and upper bound of the $\hat p$ value.

\begin{lemma}
\label{lemma:pvaluebound}
Let \(\Ncand(u) \subseteq K(u) \cup \left( \bigcup_{w \in K(u)} \Nadm(w) \right) \) be the candidate set considered in Lemma~\ref{lem:goodratio}, and define \( D(u) = \Ncand(u) \setminus K(u) \). If \( |D(u)| \leq \frac{\dadm(K(u))}{|K(u)|} \), then for any $v \in \Ncand(u)$
\begin{align*}
    p_v = \Omega\left(\frac{\covereps \eps^{4} |D(u)|}{\dc(V)}\right), \quad  p_v = O\left(\frac{\eps^{-1} d(u)}{\dc(V)}\right)
\end{align*}
\end{lemma}
\begin{proof}
For any node \(v \in K(u)\), we have  
\begin{align*}
    p_v &= p_u \geq \frac{\covereps \dc(u)}{16\dc(V)} \geq \frac{\covereps \dc(K(u))}{16|K(u)|\dc(V)} = \Omega\left( \frac{\covereps \eps^{3} \dadm(K(u))}{|K(u)| \dc(V)} \right) = \Omega\left(\frac{\covereps \eps^{3} |D(u)|}{\dc(V)}\right).
\end{align*}

The first inequality follows from Lemma~\ref{lemma:pvaluehasarange}, the second inequality is based on the definition of \(\dc(u)\), and the first equality follows from Lemma~\ref{lem:dc}.  

For any node \(v \in D(u)\), we have  
\begin{align*}
    p_v &\geq \frac{\covereps \dc(v)}{16\dc(V)} \geq \frac{\covereps d(v)}{32\dc(V)} = \Omega\left(\frac{\covereps \eps d(w)}{\dc(V)}\right) (\text{for some $w \in K(u)$})\\
    &= \Omega\left(\frac{\covereps \eps \dc(w)}{\dc(V)}\right) = \Omega\left(\frac{\covereps \eps \dc(u)}{\dc(V)}\right) = \Omega\left(\frac{\covereps \eps^{4} |D(u)|}{\dc(V)}\right).
\end{align*}
The first inequality follows from Lemma~\ref{lemma:pvaluehasarange}, the second inequality is based on the definition of \(\dc(v)\). The first equality follows from \(\eps\)-similar preclustering. Note that we do not require the edge \(uv\) to be an admissible edge; instead, we only require that there exists some node \(w\) such that \(vw\) is admissible. The second equality follows from \(\eps\)-similar preclustering and the definition of \(\dc(u)\).

Next, we focus on the upper bound. The upper bound for $v \in K(u)$ holds based on the definition of $\dc(v)$. From Lemma~\ref{lemma:pvaluehasarange}, each node \(v \in D(u)\) satisfies  
\begin{align*}
    \hat{p}_v \leq \frac{16\dc(v)}{\dc(V)} = O\left(\frac{d(v)}{\dc(V)}\right) = O\left(\frac{\eps^{-1} d(w)}{\dc(V)}\right) = O\left(\frac{\eps^{-1} d(u)}{\dc(V)}\right).
\end{align*}

The first equality follows from the definition of \(\dc(v)\), the second equality follows from the fact that there exists some \(w \in K(u)\) such that \(vw\) is admissible and \(\eps\)-similar preclustering, and the last equality holds due to \(\eps\)-similar preclustering.
\end{proof}

We remark that in the sublinear model, it is impossible to compute \(\bigcap_{v \in K(u)} \Nadm(v)\) exactly. Therefore, we must relax the definition of \(\Ncand\) in the sublinear model. That's the main reason that we only require \( |D(u)| \leq \frac{\dadm(K(u))}{|K(u)|} \) in Lemma~\ref{lemma:pvaluebound}.

Now, we are able to argue that we do not need to update \(C_u\) too many times. By Lemma~\ref{lem:goodratio}, we know that for each \(u\), we only consider nodes in \(\Ncand(u)\) for inclusion in \(C_u\).

\begin{lemma}
\label{lem:cudoesnotupdateoften}
Consider Algorihtm~\ref{alg:disjointfamily}, for each $u \in U$, \(C_u\) is updated at most \(O(\covereps^{-5}\eps^{-5})\) times.
\end{lemma}

\begin{proof}
Since all nodes in \(K(u)\) are always chosen simultaneously, at most one round is required to set the \(\hat{p}\) value of nodes in \(K(u)\) to zero. Our goal is to bound the number of rounds needed to remove all nodes from \(D(u)\). To do this, we derive a lower bound on \(\hat{p}(C_u)\).  

The first thing is the condition for Lemma~\ref{lemma:pvaluebound}. Note that we set 
\[
\Ncand(u) = K(u) \cup \left( \bigcap_{w \in K(u)} \Nadm(w) \right)
\]
if \(u\) is not a singleton. The candidate set is the intersection of all admissible neighbors, so we must have 
\[
|D(u)| \leq \frac{\dadm(K(u))}{|K(u)|}.
\]

so from Lemma~\ref{lemma:pvaluebound}, we know that $\hat p_v = \Omega\left(\frac{\covereps \eps^{4} |D(u)|}{\dc(V)}\right)$. From Claim~\ref{clm:size}, we know that \( |C_u| \geq \covereps^3 d(u) \). To remove a \(\covereps\)-fraction of the \(\hat{p}\) value of \(C_u\), we must remove a total weight of at least  
\begin{align*}
    \Omega\left(\covereps \cdot \covereps^3 d(u) \cdot \frac{\covereps \eps^{4} |D(u)|}{\dc(V)}\right) = \Omega\left(\frac{\covereps^5 \eps^{4} d(u) |D(u)|}{\dc(V)}\right).
\end{align*}

On the other hand, from Lemma~\ref{lemma:pvaluebound}, we know that $\hat p_v = O\left(\frac{\eps^{-1} d(u)}{\dc(V)}\right)$. If we do not set the \(\hat{p}\) value of nodes in \(K(u)\) to zero, we must set at least  
\begin{align*}
    \Omega\left(\frac{\covereps^5 \eps^{4} d(u) |D(u)| / \dc(V)}{\eps^{-1} d(u) / \dc(V)}\right) = \Omega(\covereps^5 \eps^{5} |D(u)|)
\end{align*}
nodes from \(D(u)\) to \(\hat{p} = 0\) to reduce $\covereps$ fractional value of $\hat p(C_u)$. Therefore, after at most \(O(\covereps^{-5}\eps^{-5})\) rounds, the algorithm will have set all nodes in \(D(u)\) to zero.
\end{proof}

\subsection{Wrap-Up: Proof of Nearly Linear Time Algorithm for \ref{LP:clusterlp}}

We are now ready to prove Theorem~\ref{thm:solving-cluster-LP-sequential-nearlyLinear}.

\begin{proof}[Proof of Theorem~\ref{thm:solving-cluster-LP-sequential-nearlyLinear}]
We can obtain the preclustering $(\calK, E_{\adm})$ in time $\widetilde O(n)$ by Theorem \ref{thm:preclustering-proc}.
    By Lemma \ref{lem:mw}, we can obtain a solution $z$ to the \ref{LP:coverclusterlp} after constant rounds $\tmwu = O(\poly(1/\epsilon))$ of Algorithm \ref{alg:mw} such that $\covers(z) \leq (1+O(\covereps))~\covers(\opt)$. Moreover, $z_S$ is at least $\frac{1}{\tmwu}$ for each $S \in \supp(z)$ since each coordinate of the points $z^{(t)}$ is either $0$ or at least $1$. The solution $z$ is simply the average of all the points $z^{(t)}$. Similar, the solution $z$ does not split atoms since each $z^{(t)}$ does not split atoms by Lemma~\ref{lem:cover-to-cluster}. Furthermore, each vertex $v$ is contained in at most $T$ sets $S \in \supp(z)$ since the solutions $z^{(t)}$ have disjoint support by Lemma~\ref{lem:cover-to-cluster}. 

    By Lemma \ref{lem:cover-to-cluster}, in time $O(n)$, we can convert $z$ into a solution $\widetilde z$ to the \ref{LP:clusterlp} such that $\covers(\widetilde z) \leq \covers(z)$ and $\widetilde z_S \geq \frac{1}{c \tmwu}$ for all $S \in \supp(z)$.
    Observe that for a solution $z$ to the \ref{LP:clusterlp}, we have $\covers(z) = \cost(z) + \dc(V)$. Thus,
    \begin{align*}
        \cost(\widetilde z) & \leq \covers(\widetilde z) - \dc(V) \\
        & \leq (1+O(\gamma)) \covers(z) - \dc(V) \\
        & \leq (1+O(\gamma))\covers(\opt) - \dc(V) \\
        & = (1+O(\gamma))\clcost(\opt) + O(\gamma) \dc(V) \\
        & \leq (1+O(\gamma))\clcost(\opt) + O(\epsilon) \clcost(\opt) \\
        & \leq (1+O(\epsilon)) \clcost(\opt).
    \end{align*}
    Here, we used that $\dc(V) = O(\epsilon^{-12})\cost(\opt)$ by Lemma \ref{lem:dc}. In each round of Algorithm \ref{alg:mw}, we have to construct the point $z^{(t)}$. To do so, we will find the family $\mathcal{F}$ from Lemma \ref{lem:disjointfamilyapproximateratio} in expected time $\widetilde O(m)$ by Lemma \ref{lem:disjointfamilylinear}. 
\end{proof}

\end{proof}

%% file: familyclusters.tex
In order to achieve sublinear runtime, we will have to address the following problems. 
\begin{enumerate}
    \item First, it is not clear how to compute $\dc(v)$ for vertices in non-singleton atoms. However, we can estimate $\dc(v)$ in sublinear time for all vertices in atoms.
    \item We need to implement the algorithm that finds one good ratio cluster, GenerateClusterBySampling (Algorithm \ref{alg:generateclusterBySampling}) in time $\widetilde O(d(r))$ instead of $\widetilde O(d^2(r))$. More specifically, 
    \begin{itemize}
        \item We need to be able to find $\Ncand(r)$,
        \item We need to estimate the cost of the cluster $T$. 
    \end{itemize}
    \item We need to determine the best solution $z$ since we compute one for each guess $R$ of the optimal cost.
\end{enumerate}
Now, we address each problem one by one.

\subsection{Compute $\dc(v)$}
Actually, it is impossible to compute \( \dc(v) \) in sublinear time if \( \dc(v) \) is very small. Consider the following two scenarios, where we have two graphs:

\begin{itemize}
    \item The first graph consists of two cliques, each containing \( n \) nodes.
    \item The second graph also consists of two cliques, but with a slight modification: we remove one random edge from each clique and then add two edges crossing the cliques, connecting the vertices from which we removed edges.
\end{itemize}

Suppose we are given one of these two graphs with probability \( 1/2 \). If we could compute \( \dc(v) \) for these graphs in sublinear time, we would be able to distinguish which graph we were given. However, it is not difficult to show that distinguishing between these two graphs requires \( \Omega(n^2) \) queries for any (randomized) algorithm. See~\cite{DBLP:conf/innovations/Assadi022} B.1 for similar reduction.

The issue in the above scenario arises when \( \dc(v) \) is too small, making it impossible to compute in sublinear time. To address this, whenever \( \dc(v) \) is too small, we set \( K(v) \) as a cluster. For large \( \dc(v) \), instead of computing it exactly, we estimate \( \dc(v) \) within a \( (1 + \esteps) \) approximation, for any small enough constant $\esteps = O(\eps^3)$. 

This is formally captured by the following lemma.

\begin{lemma} 
Let $\esteps > 0$ be a small enough constant. In time $\widetilde O(|K| )$, we can find, with high probability, either an estimate $\hatdc(K)$ such that
\[
    (1 - \esteps) \dc(K) \leq \hatdc(K) \leq (1+\esteps) \dc(K)
\]
or a certificate that $K$ is a cluster in $\opt$.
\label{lem:estimate-dc}
\end{lemma}
\begin{proof}
Note that $\dc(K) = 2 \cdot \cost(K) = |E^+(K,V\setminus K)| + 2 |E^-(K)|$ by definition. We will estimate both terms in the sum individually. 
\begin{claim}
    If $|E^+(K,V\setminus K)| \geq \esteps^2 |K|$, then we can find an estimate $X$ in time $\widetilde O( |K|)$ such that 
    \[
        |X - |E^+(K,V\setminus K)|| \leq \frac{\esteps}{4} \cdot |E^+(K,V\setminus K)|
    \]
    with high probability.
    \label{clm:est-plus-cost}
\end{claim}
\begin{cproof}
    We will sample $s = 32 \esteps^{-4} |K| \log n$ edges incident on $K$ uniformly at random. To sample one such edge we can sample a vertex proportional to his degree, then an edge adjacent to the vertex uniformly at random and if the edge is inside $K$ we will discard it with probability $1/2$. Let $m_K$ be the number of edges incident on $K$.
    We set the estimate to be 
    \[
        X = \frac{m_K}{s} \sum_{i = 1}^s X_i, 
    \]
    where $X_i$ is a random variable indicating whether the $i$-the edge leaves the atom $K$. Observe that the estimator is unbiased,
    \[
        \EX[X] = \frac{m_K}{s} \sum_{i = 1}^s \EX[X_i] = \frac{m_K}{s} \sum_{i = 1}^s \frac{|E^+(K,V\setminus K)|}{m_K} = |E^+(K,V\setminus K)|.
    \]
    By our preclusetring, there are at most $\frac{1}{2}|K|^2$ edges leaving the atom $K$. Thus, $m_K \leq |K|^2$. By Chernoff,
    \begin{align*}
        \Pr\left[\left|\frac{s}{m_K}~X - \frac{s}{m_K}~|E^+(K,V\setminus K)|\right| \geq \frac{\esteps}{4} \cdot \frac{s}{m_K}~|E^+(K,V\setminus K)|\right] & \leq 2 \exp\left(-\frac{\esteps^2}{16} \cdot \frac{s}{m_K}~|E^+(K,V\setminus K)|\right)\\ & =  O\left(\frac{1}{n^2} \right).
    \end{align*}
    The last inequality holds since $m_K \leq |K|^2$ and $|E^+(K,V\setminus K)| \geq \esteps^2 |K|$ by assumption. 
\end{cproof}
We can estimate \( E^-(K) \) in a similar manner. The only difference is that we sample \( s = 32\esteps^{-4} |K| \log n \) pairs \( uv \) and check whether \( uv \in E^+ \). 

Recall that in the sublinear model, we can determine whether \( uv \in E^+ \) in \( O(1) \) time. The following lemma states the result, and we omit the proof for brevity.

\begin{claim}
    If $|E^-(K,V\setminus K)| \geq \esteps^2 |K|$, then we can find an estimate $Y$ in time $\widetilde O(|K|)$ such that 
    \[
        |Y - |E^-(K)|| \leq \frac{\esteps}{4} \cdot |E^-(K)|
    \]
    with high probability.
    \label{clm:est-minus-cost}
\end{claim}
To deal with the case where $|E^+(K,V\setminus K)|$ is small, we will repeat the estimation process $\log(n)$ times to obtain $X^{(1)},\dots,X^{(\log n)}$. The final estimate $X$ will be the median of the repetitions $X^{(1)},\dots,X^{(\log n)}$. Now, if $|E^+(K,V\setminus K)| \leq \esteps^2 |K|$ then $X^{(i)} \leq 3 \esteps^2 |K|$ with probability at least $1/3$ by Markov's inequality. Thus, $X \leq 3 \esteps^2 |K|$ with high probability. We proceed similar for $Y$. Our estimate for $\dc(K)$ will be $\hatdc(K) = X + 2 Y$.
\begin{claim}
    If $\dc(K) \geq \esteps |K|$, the the estimate $\hatdc(K) = X + 2 Y$ satisfies
    \[
        |\hatdc(K) - \dc(K)| \leq \esteps \cdot \dc(K),
    \]
    with high probability.
\end{claim}
\begin{cproof}
    Since $\dc(K) = |E^+(K,V\setminus K)| + 2|E^-(K)| \geq \esteps |K|$, we know that $|E^+(K,V\setminus K)| \geq \frac{1}{2}\esteps |K|$ or $|E^-(K)| \geq \frac{1}{4} \esteps |K|$. We will assume that $|E^+(K,V\setminus K)| \geq \frac{1}{2}\esteps |K|$. The case where $|E^-(K)| \geq \frac{1}{4} \esteps |K|$ is symmetric. If $|E^-(K)| \geq \esteps^2 |K|$, then the bound follows from Claim \ref{clm:est-plus-cost} and Claim \ref{clm:est-minus-cost}. Otherwise, with high probability, 
    \begin{align*}
        |\hatdc(K) - \dc(K)| & \leq |X - |E^+(K,V\setminus K)|| + 2|Y - |E^-(K)|| 
        \\ & \leq \frac{\esteps}{4} |E^+(K,V\setminus K)| + 8 \esteps^2 |K| 
        \\ &\leq  \esteps \cdot \dc(K).
    \end{align*}
\end{cproof}

\begin{claim}
    If $\dc(K) \leq \esteps |K|$ then $\hatdc(K) \leq 6 \esteps |K|$ with hight probability.
\end{claim}
\begin{cproof}
    Since $\dc(K) = |E^+(K,V\setminus K)| + 2|E^-(K)| \leq \esteps |K|$, we know that $|E^+(K,V\setminus K)| \leq \esteps |K|$ and $|E^-(K)| \leq \frac{1}{2} \esteps |K|$. Thus, with high probability, $X \leq 3 \esteps |K|$ and $Y \leq \frac{3}{2} \esteps |K|$. Hence, $\hatdc(K) = X + 2Y \leq 6 \esteps |K|$
\end{cproof}
So far we have shown that we can obtain an estimate $\hatdc(K)$ such that
\begin{itemize}
    \item if $\dc(K) \geq \esteps |K|$ then $\hatdc(K)$ concentrates with high probability,
    \item and if $\dc(K) \leq \esteps |K|$ then $\hatdc(K) \leq 6\esteps |K|$ with high probability. 
\end{itemize}
To finish the proof, we need to show that if the estimate is small then $K$ is a cluster in the optimal solution and on the other hand, if the estimate is large then it concentrates. 
\begin{claim}
    If $\hatdc(K) \leq 6 \esteps |K|$, then $K$ is a cluster in the optimal solution.
\end{claim}
\begin{cproof}
    In the case where $\dc(K) \geq 12 \esteps |K|$, we know that $\hatdc(K) \geq (1-\esteps) 12 \esteps |K| > 6 \esteps |K|$ with high probability. Thus, $\dc(K) < 12 \esteps |K|$ with high probability since we assume $\hatdc(K) \leq 6 \esteps |K|$. In that case, however, by Lemma \ref{lem:dc},
\[
   12 \esteps |K| > \dc(K) \geq \Omega(\eps^3 \dadm(K)). 
\]
In particular, for a small enough $\esteps$, $\dadm(K) < |K|$. This implies that $\bigcap_{v \in K} \Nadm(v) = \emptyset$. Thus, $K$ is a cluster in $\opt$.
\end{cproof}

\begin{claim}
    If $\hatdc(K) > 6 \esteps |K|$, then with hight probability
     \[
        (1-\esteps) \dc(K) \leq \hatdc(K) \leq (1+\esteps) \dc(K). 
     \]
\end{claim}
\begin{cproof}
In the case where $\dc(K) \leq \esteps |K|$, we know that $\hatdc(K) \leq 6 \esteps |K|$ with high probability. Thus, $\dc(K) \geq \esteps |K|$ with high probability since we assume $\hatdc(K) > 6 \esteps |K|$. To finish the proof, remember that 
\[
    (1-\esteps) \dc(K) \leq \hatdc(K) \leq (1+\esteps) \dc(K)
\]
with high probability if $\dc(K) \geq \esteps |K|$.
\end{cproof}

\end{proof}

Note that, based on Lemma~\ref{lem:estimate-dc}, if \( \dc(v) \) is too small, then we make \( K(v) \) a cluster and never consider it in~\ref{LP:clusterlp}. For any \( K(v) \) that we do consider, we can assume that \( \dc(K) = \Omega(\poly(1/\eps) |K|) \).

\subsection{Approximate $\Ncand$}

The next question is the construction of \( \Ncand \). In Section~\ref{sec:goodratio}, we defined the candidate set as the intersection of all admissible neighbors. However, it is again impossible to compute such a candidate set exactly in $\tilde O(n)$ time. Instead, we must relax the definition of the candidate set in a way that does not affect the validity of other proofs.

\begin{lemma}
    Given a vertex $r$ that is part of a non-singleton atom $K(r)$. We can find a set $\hat D(r) \subseteq \Nadm(r)$ in time $\widetilde O(d(r))$ such that 
    \[
        |\hat D(r)| \leq \frac{2 \cdot \dadm(K)}{|K|}
    \]
    with hight probability.
    \label{lem:estimateN}
\end{lemma}
\begin{proof}
    We will sample $s = \log n$ vertices $v_1, \dots, v_s$ from $K(r)$ uniformly at random. Let $u$ be the vertex that minimizes $\dadm(u)$ among the vertices $v_1, \dots v_s$. We set $\hat D(r) = \Nadm(u) \cap \Nadm(r)$.
    Observe that 
    \[
        \EX_{v \sim K} \left[\dadm(v)\right] = \frac{\dadm(K)}{K}.
    \]
    Thus, by Markov's inequality, $\dadm(v_1) \leq \frac{2 \cdot \dadm(K)}{K}$ with probability at least $1/2$. Thus, after $s = \log n$ repetitions there will be a vertex $v$ among $v_1, \dots v_s$ such that $\dadm(v) \leq \frac{2 \cdot \dadm(K)}{K}$ with high probability. In particular, 
    \[
        |\hat D(r)| \leq \dadm(v) \leq \frac{2 \cdot \dadm(K)}{|K|}. 
    \]
\end{proof}

Lemma \ref{lem:sizeofcandidateset} and Lemma \ref{lem:boundtwohopsneighbor} still hold true for the approximate $\hatNcand \coloneqq K(r) \cup \hat D(r)$ where $\hat D(r)$ is the estimate from Lemma \ref{lem:estimateN}. 

\begin{lemma}
    \label{lem:sizeofcandidateset-approx}
    For any $r \in V$, if $v \in \hatNcand(r)$, then $|\hatNcand(r)| = O(\eps^{-4}d(v))$.
\end{lemma}
The proof for Lemma \ref{lem:sizeofcandidateset-approx} is the same as for Lemma \ref{lem:sizeofcandidateset} since $\hatNcand \subseteq \Nadm(r)$ remains true for the estimate.

\begin{lemma}
\label{lem:boundtwohopsneighbor-approx}
For any vertex $r$ and any \epslarge cluster $C$ such that $K(r) \subseteq C \subseteq \hatNcand(r)$, we have  
\begin{align*}
    |\hat D(r)| \cdot |\hatNcand(r)| = O( \eps^{-8} \dadm(C)). 
\end{align*}
\end{lemma}
\begin{proof}
Fix $\alpha = \eps^8$. We will distinguish two cases.
\begin{enumerate}
    \item $|K(r)| \geq \alpha |\hatNcand(r)|$. In this case, we can show the required inequality immediately because the estimate satisfies $|K(r)| \cdot |\hat D(r)| \leq 2 \dadm(K(r))$. In particular, 
    \begin{align*}
        \dadm(C) \geq \dadm(K(r)) \geq \frac{1}{2} |K(r)| \cdot |\hat D(r)| \geq \frac{\alpha}{2} |\hatNcand(r)| \cdot |\hat D(r)|.
    \end{align*}
    \item $|K(r)| \leq \alpha |\hatNcand(r)|$. We have by definition of $\hatNcand$ that $|\hatNcand(r)| \leq |K(r)| + |\Nadm(r)| = O(\epsilon^{-3} d(r))$. Since $C$ is $\eps$-large, $|C| \geq \eps d(r) \geq \Omega(\epsilon^4 |\Ncand(r)|)$
    \begin{align*}
        \dadm(C) & \geq \dadm(C \setminus K(r)) 
        \\ & \geq (|C| - |K(r)|) \cdot |C| \\ & \geq \Omega((\epsilon^4 |\hatNcand(r) - |K(r)|) \cdot \epsilon^4 \hatNcand(r))
        \\ & \geq \Omega((\epsilon^4 (\epsilon^4 - \alpha) |\hatNcand(r)|^2))
        \\ & \geq \Omega(\epsilon^8 |\hatNcand(r)|^2).
    \end{align*}
\end{enumerate}

\end{proof}

\subsection{Estimate the cost of a cluster $T$}
We are given a vertex $r$ and a cluster $T \subseteq \hatNcand(r)$ that contains at most one non-singleton atom $K \coloneqq K(r)$.
We want to find a good estimate of the $\cost(T)$. If $T$ doesn't contains a non-singleton atom $K$, we will define $K = \emptyset$. We can write the cost of the cluster $T$ as  
    \[
        \cost(T) = \cost(K) + \Delta(T).
    \]
    Here, $\Delta(T)$ is the change in cost when adding the missing vertices from $T \setminus K$ to $K$. We can write $\Delta(T)$ as the sum over the individual contributions of the vertices in $T\setminus K$, 
    \[
        \Delta(T) = \sum_{v \in T \setminus K} \Delta(T, v).
    \]
    where $\Delta(T, v)$ is given as
    \begin{align*}
        \Delta(T, v) = \frac{d^+(v) + d^-(v, T) + d^-(v, K) - d^+(v, T) - d^+(v, K)}{2}
    \end{align*}
    We already have a good estimate for $\cost(K) = \dc(K)$ by Lemma \ref{lem:estimate-dc}. We can also obtain a good estimate for $\Delta(T)$.
    \begin{lemma}
    For any small enough constant $\esteps > 0$, given a vertex $r$ and a cluster $T \subseteq \hatNcand(r)$ that contains at most one non-singleton atom $K \coloneqq K(r)$. Assume there exists an \epslarge cluster $C$ with $K(r) \subseteq C \subseteq \hatNcand(r)$. We can find an estimate $\hat \Delta(T)$ in time $\widetilde O(d(r) \eps^{20} / \esteps)$ such that,
    \[
        |\hat \Delta(T) - \Delta(T)| \leq \esteps \cdot \min\{\dadm(C), \dc(T\setminus K)\}
    \]
    with hight probability. 
    \label{lem:estimatedelta}
\end{lemma}
\begin{proof}
     To estimate $\Delta(T)$ we sample $s = \epsilon^{-20}\esteps^{-2} \log n$ vertices $v_1, \dots, v_s$ from $T \setminus K$ uniformly at random. We set the estimate to be
    \[
        \hat \Delta(T) = \frac{|T \setminus K|}{s} \sum_{i=1}^s \Delta(T, v_i).
    \]
     Note that the estimator is unbiased,
    \[
        \EX [\hat \Delta(T)] =  \frac{|T\setminus K|}{s} \sum_{i=1}^s \EX [\Delta(T, v_i)] = \sum_{v \in T \setminus K} \Delta(T, v) = \Delta(T).
    \]
    Observe that $|\Delta(T, v)| \leq d(v) + |T| \leq d(v) + |\hatNcand(r)| = O(\epsilon^{-2} |\hatNcand(r)|)$. Here, we used that the \epslarge cluster $C$ is a subset of $\hatNcand(r)$, in particular, $|\hatNcand(r)| \geq |C| \geq \epsilon d(r)$. By Hoeffdings inequality,
    \begin{align*}
        \Pr\left[\left|s \cdot \hat \Delta(T)  - s \cdot \Delta(T) \right | \geq s \cdot \epsilon^{8}\esteps \cdot |\hatNcand(r)| \cdot |\hat D(r)| \right] & \leq 2 \exp\left(- \Omega\left(\frac{s^2 \epsilon^{16} \esteps^2 \cdot |\hatNcand(r)|^2 \cdot |\hat D(r)|^2}{s \cdot \epsilon^{-4} |\hatNcand(r)|^2 \cdot |C\setminus K|^2}\right) \right)
        \\ & \leq 2 \exp\left(-\Omega(s \cdot \epsilon^{20}\esteps^2)\right)
        \\ & = O\left(\frac{1}{n^2}\right)
    \end{align*}
    The last inequality holds since $T \setminus K \subseteq \hat D(r)$. Remember that $|\hatNcand(r)| \cdot |\hat D(r)| = O( \eps^{-8} \dadm(C))$ by Lemma \ref{lem:boundtwohopsneighbor-approx}.
    \newline Similar, $|\Delta(C, v)| \leq d(v) + |C| \leq d(v) + |\hatNcand(r)| = O(\epsilon^{-3} d(r))$ by Lemma \ref{lem:sizeofcandidateset-approx}. Note that $\dc(v) = d(v) = \Omega(\epsilon d(r))$ for all $v \in T \setminus K$. Thus, $\dc(T\setminus K) = \Omega(|T\setminus K| \cdot \epsilon d(r)) $ By Hoeffdings inequality,
    \begin{align*}
        \Pr\left[\left|s \cdot \hat \Delta(T)  - s \cdot \Delta(T) \right | \geq s \cdot \esteps \cdot \dc(T \setminus K)\right] & \leq 2 \exp\left(-\frac{s^2 \esteps^2 \dc^2(T \setminus K)}{s \cdot \epsilon^{-6} \cdot d^2(r) \cdot |T\setminus K|^2} \right)
        \\ & \leq 2 \exp\left(- \Omega\left(\frac{s\cdot |T\setminus K|^2 \cdot \esteps^2 \cdot \epsilon^2 \cdot d^2(r)}{|T\setminus K|^2 \cdot \epsilon^{-6} d^2(r)}\right)\right)
        \\ & = O\left(\frac{1}{n^2}\right).
    \end{align*}
\end{proof}
Now, we are ready to show that we can combine $\hatdc(K)$ and $\hat \Delta(T)$ to get a good estimate for $\hatcosts(T)$.
\begin{lemma}
\label{lem:estimate-cost}
For any small enough constant $\esteps > 0$, given a vertex $r$ and a cluster $T \subseteq \hatNcand(r)$ that contains at most one non-singleton atom $K \coloneqq K(r)$, we can find an estimate $\hatcosts(T)$ in time $\widetilde O(d(r))$ such that,
    \[
        |\hatcosts(T) - \cost(T)| \leq \esteps \cdot \dc(T)
    \]
    with hight probability. 
\end{lemma}
\begin{proof}
    By Lemma \ref{lem:estimate-dc} and Lemma \ref{lem:estimatedelta}, we can obtain estimates $\hatdc(K)$ and $\hat \Delta(T)$ respectively. We will set $\hat \cost(T) = \hatdc(K) + \hat \Delta(T)$. We have by the guarantees provided, 
    \begin{align*}
        |\hatcosts(T) - \cost(T)| & \leq |\hatdc(K) - \dc(K)| + |\hat \Delta(T) - \Delta(T)| \\ & \leq \esteps \cdot (\dc(K) + \dc(T \setminus K)) \\ & = \esteps \cdot \dc(T).
    \end{align*}
\end{proof}

\subsection{Finding one small ratio cluster in sublinear time}
\begin{lemma}
\label{lem:goodratio-sublinear}
Suppose we are given a graph $G = (V, E)$, vertex weights $\hatp$, a target ratio $R$, a vertex $r$ and the set of vertices $\Nadm(r)$. 
\begin{itemize}
    \item[(i)] Assume there exists a cluster $C$ be an \epslarge cluster with $K(r) \subseteq C \subseteq \Ncand(r)$ such that $\covers(C) + \covereps^2  \dadm(C) \leq R \cdot \hatp(C)$.
    \item[(ii)] Assume that $\covers(\{ v \}) > R \cdot \hatp(\{ v \})$ for all $v \in C$.
\end{itemize}
Then, with high probability, in time $\widetilde O(d(r))$, we can find
a cluster $C_r \subseteq \Ncand(r)$ such that,
\[
    \covers(C_r) \leq R \cdot \hatp(C_r).
\]
Moreover, $C_r$ does not split atoms and contains exactly one non-singleton atom $K(r) \subseteq C_r$ iff $K(r)$ is a non-singleton atom.
\end{lemma}

\begin{algorithm}[H]
\caption{\genclusbysam($G, \Nadm(r),\hatNcand(r), w, r, R$)}
\label{alg:generateclusterBySamplingsublinear}
\begin{algorithmic}[1]
\STATE \textbf{Input:} The graph $G$, $K(r)$, $\Nadm(r)$, $\hatNcand(r), w$, $r$, ratio $R$.
\STATE \textbf{Output:} A small ratio cluster $\hatt$ if $r$ satisfies Assumption (i) from Lemma \ref{lem:goodratio}. 
\STATE Repeat the following steps $O(\log n)$ times.
\FOR{$i$ from $1$ to $\eta$}
\STATE Uniformly sample $\Theta(\eta^4  \covereps^{-2} \eps^{-8})$ vertices from $\hatNcand(r)$ with replacement.
\STATE Let the sample set be $A_i$. \COMMENT{$A_i$ may contain some element multiple times.}
\ENDFOR
\STATE $\hat D(r) \gets \hatNcand(r)\setminus{K(r)}$
\STATE $\mathcal{T} \gets \emptyset$ 
\FOR{every $(S^1, S^2,...,S^{\eta}) \subset (A_1, A_2,...,A_{\eta})$ such that $|S^i| \leq \eta$, where $i \in [\eta]$} 
    \FOR{ every $(\tilt_1, \tilt_2, ..., \tilt_{\eta}) \in (L(r), L(r),...L(r))$, where $\tilt_j \in L(r)$ for $j \in [\eta]$ }
    \STATE Add $T \gets$ GenerateCluster($r, D(r), S^1,\ldots,S^\eta, \tilt_1,\ldots, \tilt_\eta$) to $\mathcal{T}$.
\ENDFOR
\ENDFOR
\STATE Compute the estimate $\hat \Delta(T)$ for each $T \in \mathcal{T}$ according to Lemma \ref{lem:estimatedelta}.
\RETURN $T \in \mathcal{T}$ that minimizes $\hat \Delta(T)$.
\end{algorithmic}
\end{algorithm}

\begin{algorithm}[H]
\caption{\genclus($r, \hat D(r), S^1,\ldots,S^\eta, \tilt_1,\ldots, \tilt_\eta$)}
\label{alg:GenerateClustersublinear}
\begin{algorithmic}[1]
\STATE $T \gets K(r), \hatt_1 \gets K(r)$
\STATE Let $D_r^1, \ldots, D_r^{\eta}$ be an arbitrary partition of the vertices of $\hat D(r)$ into equal-size parts
\FORALL{$i=1,\ldots,\eta$} \label{loop:for-gc}
\FORALL{$v \in D_r^i$} 
\IF{$\rmEstCost(S^i, \tilt_i, v$)$ + 6\eta^{-1} |\hatNcand(r)| \leq w(v)$}
\STATE $T \gets T \cup \{ v \}$ 
\ENDIF
\ENDFOR 
\STATE $\hatt_{i+1} \gets T$
\ENDFOR
\end{algorithmic}
\end{algorithm}

\begin{proof}
The proof is almost identical to that of Lemma~\ref{lem:goodratio}. The only difference is that we use \( \hat{\Delta}(T) \) to choose \( T \). From Lemma~\ref{lem:estimatedelta}, we know that the estimation introduces at most an error of \( O(\beta \cdot \dadm(C)) \). 

The first requirement of the lemma provides a slackness of \( \covereps^2 \dadm(C) \). As long as we set \( \beta = O(\covereps^2) \), the estimation error remains within the allowed slackness.

\end{proof}

\nairen{check it}
\subsection{Determine the correct guess \( R \) of the optimal cost }
Our final challenge is to determine the best solution $z$ among the solutions we computed for each guess $R$.
Since we iterate over different granularities of \( R \) in the range \( [1, n^2] \), we know that one of these values satisfies \( R \in [\covers(\opt), (1 + \covereps) \covers(\opt)] \). Thus, we know that for this guess $R$, we will compute a solution $z$ with small cost. To determine the this solution (or any other with a smaller cost), we will use Lemma~\ref{lem:estimate-cost} to estimate the cost. After applying Algorithm~\ref{alg:solveclusterlpframework}, we obtain a solution \( \{ z_S \} \), where each nonzero \( z_S \) is at least some constant. We can estimate \( \cost(S) \) for each nonzero \( z_S \), and since each node appears in only a constant number of sets \( S \), the total running time for estimating all \( \cost(S) \) values is \( \tilde{O}(n) \). 

On the other hand, each node contributes to the estimation error of \(\beta \cdot \dc(v) \), leading to a total error of \( \beta \cdot \dc(V) \). As long as \( \beta \) is a sufficiently small constant, we can determine a solution $z$ that is a good approximation.

%% file: mpcalgorithm.tex
In this section, we present an algorithm to solve~\ref{LP:clusterlp} in the MPC model. Our goal is to establish the following theorem for the MPC model.

\thmsolveclusterLPsequential*

While Algorithm~\ref{alg:mw} is already well-parallelized and runs in \(O\left(\frac{\log (1/\covereps)}{\covereps^4}\right)\) rounds, the algorithm for finding a disjoint family of clusters is not well-parallelized. In Algorithm~\ref{alg:disjointfamily}, clusters are identified sequentially and added one by one to the final output set \(\mathcal{F}\). While the running time is efficient, the approach lacks parallelism. To address this, we need a method that selects multiple good ratio clusters \(C_u\) simultaneously. The MPC algorithm is presented in Algorithm~\ref{alg:disjointfamilympc}.

\textbf{Algorithm Description}  
Algorithm~\ref{alg:disjointfamilympc} first removes all nodes whose ratio is either too large or too small. For nodes with a large ratio, we add them to \(\mathcal{F}\) (Line~\ref{line:addlargepnodesmpc}), and for nodes with a small ratio, we remove them by setting their \(\hat{q}\) values to 0 (Line~\ref{line:removesmallpnodesmpc}). 

Next, we attempt to find a disjoint set of clusters \(C_u\) such that, in each round, the sum of their \( p \)-values is sufficiently large. This set can be found with constant probability, so we repeat this process for \(\Theta(\log n)\) rounds (Line~\ref{line:oneroundendmpc}). 

In each round, we add a cluster center to \( U \) with probability proportional to its degree. Thus, for each cluster in the optimal solution, at least one node is chosen into \( U \). We then consider the candidate sets of these chosen nodes: if a node appears in the candidate sets of two chosen nodes, we remove it in this round by setting its \(\hat{p}\) value to 0. For each chosen node, as long as it retains enough fractional modes that have not been removed, we apply Lemma~\ref{lem:goodratio} to find a good ratio cluster. We set the parameter as \(\covereps' = \Theta(\covereps^4 \eps^5)\).

\begin{algorithm}[h]
	\caption{MPC Algorithm to find the family $\mathcal{F}$}
	\label{alg:disjointfamilympc}
	\begin{algorithmic}[1]
    \STATE Let $R$ be the guess for $\covers(\opt)$ such that $R \in [\covers(\opt) , (1 + \frac{\covereps}{2}) \covers(\opt))$.
    \STATE $\hat q \gets p, \mathcal{F} \gets \emptyset$
    \FORALL{$v \in V$}
        \label{line:addlargepnodesmpc} \STATE If $\frac{\covers(K(v))}{p(K(v))} \leq (1+6\covereps)~R$, add $K(v)$ to $\mathcal{F}$ and set $\hat q_w = 0$ for all $w \in K(v)$.
    \label{line:removesmallpnodesmpc} \STATE If $p_v \leq \frac{\covereps \dc(v)}{4\dc(V)}$, set $\hat q_w = 0$ for all $w \in K(v)$.
    \ENDFOR
    \label{line:whileloopmpc}\WHILE{$p(\mathcal{F}) \leq \covereps$}
        \label{line:oneroundstartmpc}\FOR{$t=1,\dots, \Theta(\log n / (\eps^5 \covereps'))$}
        \STATE $\hat p \gets \hat q, \mathcal{F}_t \gets \emptyset$
        \STATE $\hat p_v \gets 0$, for all nodes $v$ in $\mathcal{F}$
        \STATE \label{line:addvertextou} Add each vertex $v$ with probability $\frac{\eps^4\covereps'}{24d(v)}$ to $U$. 
        \STATE Mark all nodes in $D(u)$, for every $u \in U$.
        \FORALL{$u \in U$} \label{line:loop}
        \STATE Let $Remove(u) = \{ w \in D(u) \mid w \text{ gets more than one mark} \}$
        \label{line:removevertexfromu}\STATE If $\hat q(Remove(u)) \geq \gamma' \hat q(D(u))$, then remove $u$ from $U$.
        \STATE set $\hat p_w \gets 0$ for all $w \in Remove(u)$.
    \ENDFOR
    \STATE Let $\widetilde U$ be the $U$ set.
    \FOR{$u \in \widetilde U$}
        \STATE Find a good ratio cluster $C_u$ such that $K(u) \subseteq C_u \subseteq \Nadm(u)$ with vertex weights $\hat p > 0$ and target ratio $(1+5\covereps)R$ (Lemma \ref{lem:goodratio}).
        \STATE add $C_u$ to $\mathcal{F}_t$
    \ENDFOR
    \ENDFOR
    \STATE $\tmpc \gets \mathrm{argmax}_{t} p(\mathcal F_t)$
    \label{line:oneroundendmpc} \STATE add $\mathcal F_{\tmpc}$ to $\mathcal F$ and sets $\hat q_v = 0$ for all $v \in \mathcal F_{\tmpc}$
    \ENDWHILE
    \RETURN $\mathcal{F}$
	\end{algorithmic}
\end{algorithm}

\subsection{Approximate Ratio of Algorithm~\ref{alg:disjointfamilympc}}

Our main lemma regarding approximate ratio is given as follows,

\begin{lemma}
\label{lem:familympc}
Given vertex weights $p_v >0$, Algorithm \ref{alg:disjointfamilympc} finds a family $\mathcal{F} = \{S_1, S_2,..., S_l \mid S_i \subseteq V \}$ such that,
\begin{enumerate}
    \item for any distinct $S, T \in \mathcal{F}$, $S \cap T = \emptyset$,
    \item $\frac{\covers(\mathcal{F})}{p(\mathcal{F})} \leq (1+8\covereps)\covers(\opt)$,
    \item $p(\mathcal{F})$ is at least $\covereps$,
    \item no $S \in \mathcal{F}$ splits an atom, i.e. $K(v) \subseteq S$ for all vertices $v \in S$.
\end{enumerate}
\end{lemma}

\begin{proof}
First, observe that Property 3 holds due to the condition in the while loop. Property 2 holds because all small-ratio clusters \(\{C_v\}_{v\in U}\) do not split atoms. 

Next, we prove that the family \(\mathcal{F}\) consists of disjoint sets. Note that in each round, when we consider the candidate set, if a node appears in two candidate sets, we set its \(\hat{p}\) value to zero, ensuring that it is not added to the final cluster. 

It remains to bound the ratio \(\frac{\covers(\mathcal{F})}{p(\mathcal{F})}\), which follows directly from Lemma~\ref{lem:goodratio}.
\end{proof}

Now, we only need to show that Algorithm~\ref{alg:disjointfamilympc} terminates.

\subsection{Number of Iterations of Algorithm~\ref{alg:disjointfamilympc}}

In this section, we show that the for loop (Line~\ref{line:oneroundstartmpc} to Line~\ref{line:oneroundendmpc}) terminates in \( O(\poly(1/\epsilon)) \) rounds. More precisely, we prove the following:

\begin{lemma}
\label{lem:mpcnumberofrounds}
Consider one round of execution of Algorithm~\ref{alg:disjointfamilympc} from Line~\ref{line:oneroundstartmpc} to Line~\ref{line:oneroundendmpc}. Then, with high probability, we have 
\[
p(\mathcal{F}_{\tmpc}) = \Omega(\covereps^{10} \eps^{14}).
\]
Consequently, the while loop (Line~\ref{line:whileloopmpc}) will be executed at most \( O(1/(\covereps^{10} \eps^{14})) \) rounds with high probability.
\end{lemma}

We start our proof by showing that each round for any cluster $C \in \opt$, with constant probability, we will add at least one node to $U$ and if $C$ contains a non-singleton atom, we will include at least one node from the non-singleton atom with constant probability.

\begin{claim}
    \label{clm:hitoptmpc}
    For each cluster $C \in \opt$ with $|C| \geq 2$, we have $C \cap \widetilde U \neq \emptyset$ with probability at least $\Omega(\epsilon^5 \covereps')$. Moreover, if there exists a non-singleton atom $K \subseteq C$ then $K \cap \widetilde U \neq \emptyset$ with probability at least $\Omega(\epsilon^5 \covereps' )$. 
\end{claim}
\begin{cproof}
    Recall the $(\mathcal K, \eadm)$ is \epssimilar preclustering. Thus, for any non-singleton atom $K$ we have $|K| \geq d(v) / 2$ for each $v \in K$. On the other hand, recall that $\opt$ is \epslarge cluster, so $|C| \geq \epsilon d(v)$, and $|C| \leq \dadm(v) + |K(v)| \leq 3\epsilon^{-3}d(v)$.

    At line~\ref{line:addvertextou}, we will add node $u$ to $U$ with probability $\frac{\eps^4\covereps'}{24d(u)}$. Then, at line~\ref{line:removevertexfromu}, we will remove the node from $U$ if $\covereps'$ fractional of $\hat q(D(u))$ can be considered in other candidate sets. For a given node $u \in C$, let $X_u$ be the indicator variable that $u$ is the only node in $C$ that is added to $U$ at Line~\ref{line:addvertextou} and is still in $\widetilde U$ after line~\ref{line:removevertexfromu}. Let $A_u$ be the indicator variable that $u$ is the only node in $C$ that is added to $U$ at Line~\ref{line:addvertextou} and $B_u$ be the total $\hat p$ value that are in the $Remove(u)$ set. so
    \begin{align*}
        \Pr[X_u = 1] &= \Pr[A_u = 1 \cap B_u < \covereps' \hat p(D(u)) ] \\
        &= \Pr[B_u < \covereps' \hat p(D(u)) \mid A_u = 1] \cdot \Pr[A_u = 1]
    \end{align*}

    For $A_u$, we know that only $u$ is added to $U$, so 
    \begin{align*}
        \Pr[A_u = 1] &= \frac{\eps^4\covereps'}{24d(u)} \prod_{v \in C \setminus \{ u \}} \left(1 - \frac{\eps^4\covereps'}{24d(v)} \right) \\
        &\geq \frac{\eps^4\covereps'}{24 \epsilon^{-1}|C|} \left(1 - \frac{\eps^4\covereps'}{8|C|\epsilon^{3}} \right)^{|C|} \geq \frac{\eps^5\covereps'}{48 |C|}
    \end{align*}

    The first inequality is because $\epsilon^{3}|C| / 3 \leq d(v) \leq \epsilon^{-1}|C|$. Now condition on $A_u = 1$, we calculate the expected value of $B_u$, 
    \begin{align*}
        \EX[B_u \mid A_u = 1] &\leq \sum_{v \in D(u)} \hat p_v \sum_{w \in \Nadm(v)} \frac{\eps^4\covereps'}{24d(w)} \leq \sum_{v \in D(u)}\hat p_v \sum_{w \in \Nadm(v)} \frac{\eps^4\covereps'}{12 \eps d(v)} \\
        &\leq \sum_{v \in D(u)} \hat p_v \cdot \dadm(v) \frac{\eps^4\covereps'}{12 \eps d(v)} \leq \sum_{v \in D(u)} \hat p_v \cdot 2\eps^{-3} d(v) \frac{\eps^4\covereps'}{12 \eps d(v)} \\
        &\leq \sum_{v \in D(u)} \frac{\covereps' \hat p_v }{6} \leq \covereps' \hat p(D(u)) / 3
    \end{align*}

    Then by markov's inequality, we have 
    \begin{align*}
        \Pr[B_u \geq \covereps' \hat p(D(u)) \mid A_u = 1] \leq \frac{1}{2}
    \end{align*}

    Combining these two points together, we have 
    \begin{align*}
        \Pr[X_u = 1] \geq \frac{\eps^5\covereps'}{96 |C|}
    \end{align*}

    and $C \cap U \not= \emptyset$ with probability at least 
    \begin{align*}
        \sum_{u \in C} \Pr[X_u = 1] \geq |C| \cdot \frac{\eps^5\covereps'}{96 |C|} = \Omega(\eps^5\covereps')
    \end{align*}

The proof for \( K \cap \widetilde{U} \neq \emptyset \) is almost identical, so we omit the details here. However, we note that whenever \( K \cap \widetilde{U} \neq \emptyset \), we always have \( |K \cap \widetilde{U}| = 1 \). 

This follows from the fact that if \( |K \cap U| \geq 2 \), then all nodes in \( D(K) \) will receive two marks, causing \( D(K) \) to be added to the \( Remove \) set, which in turn removes \( K \cap U \) from \( U \). 

\end{cproof}

Consider one iteration Line~\ref{line:oneroundstartmpc} of Algorithm~\ref{alg:disjointfamilympc}, let 
\begin{align*}
\mathcal C^*_{\leq (1+4\covereps)R} = \{C^* \in \opt \mid     \frac{\covers(C^*)+\covereps^2 \dadm(C^*)}{\hat q(C^*)} \leq (1 + 4\covereps) R \}
\end{align*}

be the set of clusters from optimal clustering such that the ratio is at most $(1 + 4\covereps) R$. We will first show that $\mathcal C^*_{\leq (1+4\covereps)R}$ has a large $\hat q$ value.

\begin{lemma} 
\label{lem:boundgoodclusterratiompc}
Consider one iteration of Algorithm~\ref{alg:disjointfamilympc}, if $p (\mathcal{ F})\leq \covereps$, then we have 
\begin{align*}
    \hat q(\mathcal C^*_{\leq (1+4\covereps)R} ) \geq \covereps
\end{align*}
\end{lemma}
\begin{proof}
Notice that the algorithm will set some $\hat q_v$ to $0$ without adding $u$ to $\mathcal{F}$, let 
\[
\mathrm{Small} = \{ v \in \hat q_v \text{ is set to  0 at line~\ref{line:removesmallpnodesmpc} in Algorithm~\ref{alg:disjointfamilympc}} \}
\]
We have $ p(\mathrm{Small}) \leq \sum_{v \in V} \frac{\covereps \dc(v)}{16\dc(V)} \leq \covereps / 4$. Based on the assumption of $p(\mathcal F) \leq \covereps$, we have $\hat q(V) = p(V) - p(\mathcal{F}) - p(\mathrm{Small}) \geq 1 - 1.25\covereps $.
Thus,
\begin{align*}
    (1+2.5\covereps)R &\geq \frac{1+ \covereps}{1-1.25\covereps} R \geq \frac{(1+\covereps)\covers(\opt)}{\hat q(V)} \\
    &\geq \frac{\covers(\opt) + \covereps^2 |E_{\adm}|}{\hat q(V)} = \frac{\sum_{C \in \opt}(\covers(C) + \covereps^2 \dadm(C))}{\sum_{C \in \opt} \hat q(C)}.
\end{align*}
For the second inequality, we use that $\gamma~\covers(\opt) \geq \gamma~\clcost(\opt) \geq \gamma^2 |E_{\adm}|$ because of the preclustering (Theorem \ref{thm:preclustering-proc}). On the other hand, if $\hat q(\mathcal C_{\leq (1+4\covereps)R} ) < \covereps$, then consider $C \in \opt$ whose ratio is at least $(1+4\covereps)R$, we have 
\begin{align*}
&\quad \frac{\sum_{C \in \opt}(\covers(C) + \covereps^2 \dadm(C))}{\sum_{C \in \opt} \hat q(C)} \\ &\geq \frac{\sum_{C \in \opt \setminus C_{\leq (1+4\covereps)R}}(\covers(C) + \covereps^2 \dadm(C))}{\sum_{C \in \opt} \hat q(C)} \\
&\geq \frac{\sum_{C \in \opt \setminus C_{\leq (1+4\covereps)R}}(\covers(C) + \covereps^2 \dadm(C))}{\sum_{C \in \opt} \hat q(C)} \\
&\geq \frac{\sum_{C \in \opt \setminus C_{\leq (1+4\covereps)R}}(\covers(C) + \covereps^2 \dadm(C))}{\sum_{C \in \opt \setminus C_{\leq (1+4\covereps)R}} \hat q(C)} \cdot \frac{\sum_{C \in \opt \setminus C_{\leq (1+4\covereps)R}} \hat q(C)}{\sum_{C \in \opt } \hat q(C)} \\
&\geq (1 + 4\covereps ) R \cdot \big(1 - \frac{\covereps}{1 - 1.5\covereps} \big) > (1 + 2.5 \covereps) R
\end{align*}
which contradicts to the fact that the total ratio is at most $(1 + 2.5 \covereps) R$. So, we conclude that $\hat q(\mathcal C^*_{\leq (1+4\covereps)R} ) \geq \covereps$.
\end{proof}

Consider one iteration Line~\ref{line:oneroundstartmpc} of Algorithm~\ref{alg:disjointfamilympc}, let 
\begin{align*}
\mathcal C_{\leq (1+4\covereps)R} = \{C \in \mathcal C^*_{\leq (1+4\covereps)R} \mid C \cap \widetilde U \not= \emptyset \}
\end{align*}
be the set of good ratio cluster such that at least one of its node is chosen into $\widetilde U$, we can show that

\begin{lemma} 
\label{lem:includeenoughweightsmpc}
For every $C^* \in C_{\leq (1+4\covereps)R}$, we have 
\begin{align*}
    \hat p(C^*) \geq (1 - \frac{\covereps}{2}) \hat q(C^*)
\end{align*}
Moreover, we have 
\begin{align*}
    \Pr[\hat p(\mathcal C_{\leq (1+4\covereps)R}) = \Omega(\eps^5\covereps'\covereps)  ] = \Omega(\eps^5\covereps')
\end{align*}
\end{lemma}
\begin{proof}
when we set up $\hat p$, we first copy $\hat q$ to $\hat p$ value. Then, we might set some $\hat p$ value to $0$ if the nodes have two chosen neighbor. So, to show that $ \hat p(C^*) \geq (1 - \covereps) \hat q(C^*)$, we need to show that the removed nodes contributes to at most $\covereps$ fractional weight to $C^*$.

For each $C^* \in \mathcal C_{\leq (1+4\covereps)R}$, let $u \in C^* \cap \widetilde U$ be the node in $\widetilde U$. Based on Claim~\ref{clm:size}, we know that $|C^*| = 1$ or $|C^*| \geq \covereps^3 d(u)$. The first case is impossible since we will set its $\hat p$ value to 0 and we will never consider it in $C_{\leq (1+4\covereps)R}$. For the second case, note that we remove at most $\covereps'$ fractional $\hat q$ from $D(u)$, we will give an upper bound of this value.

Recall that based on lemma~\ref{lemma:pvaluebound}, we know that for any node $v \in C^* \subset \Ncand(u)$, we have 
\begin{align*}
     \hat q_v = O\left(\frac{\eps^{-1} d(u)}{\dc(V)}\right)
\end{align*}
so the upper bound $\hat q( D(u))$ is at most 
\begin{align*}
    \hat q(D(u)) = O\left(\frac{\eps^{-1} d(u) \cdot |D(u)|}{\dc(V)}\right)
\end{align*}
based on Claim~\ref{clm:size}, we know that $C^*$ contains at least $\covereps^3d(u)$ vertices with positive $\hat p$ value, for each vertex in $C^*$, we have 
\begin{align*}
     \hat q_v = \Omega\left(\frac{\covereps \eps^{4} |D(u)|}{\dc(V)}\right) 
\end{align*}
so, the lower bound $\hat p$ value for $C^*$ is at least
\begin{align*}
    \hat q(C^*) = \Omega\left(\frac{\covereps^4 \eps^{4} |D(u)| \cdot d(u)}{\dc(V)}\right) 
\end{align*}
For any node $u \in \widetilde U$, we remove at most $\covereps'$ fractional of its $\hat q(D(u))$ value, as long as $\covereps' = \Omega(\covereps^5\eps^4)$, we have 
\begin{align*}
    \hat p(C^*) \geq (1 - \frac{\covereps}{2}) \hat q(C^*)
\end{align*}
Now for each $C^* \in \mathcal C^*_{\leq (1+4\covereps)R}$, based on Claim~\ref{clm:hitoptmpc}, we know that with probability at least $\Omega(\eps^5\covereps')$, $C^* \cap \widetilde U \not= \emptyset$ and $C^* \in \mathcal C_{\leq (1+4\covereps)R}$, so we have 
\begin{align*}
    \EX [\hat q(\mathcal C_{\leq (1+4\covereps)R})] = \Omega(\eps^5\covereps' \hat q(\mathcal C^*_{\leq (1+4\covereps)R}))
\end{align*}
Note that $\hat q(\mathcal C_{\leq (1+4\covereps)R}) \leq \hat q(\mathcal C^*_{\leq (1+4\covereps)R})$, by the reverse markov inequality, we have 
\begin{align*}
    \Pr[\hat q(\mathcal C_{\leq (1+4\covereps)R}) = \Omega(\eps^5\covereps' q(\mathcal C^*_{\leq (1+4\covereps)R})) ] = \Omega(\eps^5\covereps')
\end{align*}
By Lemma~\ref{lem:boundtwohopsneighbor}, we know that $\hat q(C^*_{\leq (1+4\covereps)R}) \geq \covereps$, so we have 
\begin{align*}
    \Pr[\hat q(\mathcal C_{\leq (1+4\covereps)R}) = \Omega(\eps^5\covereps'\covereps) ] = \Omega(\eps^5\covereps')
\end{align*}
For each $C^* \in \mathcal C_{\leq (1+4\covereps)R}$, we have $\hat p(C^*) \geq (1 - \covereps) \hat q(C^*)$, so 
\begin{align*}
    \Pr[\hat p(\mathcal C_{\leq (1+4\covereps)R}) = \Omega(\eps^5\covereps'\covereps)  ] = \Omega(\eps^5\covereps')
\end{align*}
\end{proof}

From the above lemma, we know that for any \( C^* \in C_{\leq (1+4\covereps)R} \), at least one node \( u \in C^* \) will be fed into Lemma~\ref{lem:goodratio}. Our final goal is to show that in each round, we cover a set of good clusters with sufficient \( p \)-value. 

Next, we show that the cluster returned by Lemma~\ref{lem:goodratio} covers a constant fraction of the \(\hat{p}\)-value. This is formally stated in the following lemma.

\begin{lemma}
\label{lemma:eachcuislargeenoughmpc}
For every $C^* \in C_{\leq (1+4\covereps)R}$, let $u \in C^* \cap \widetilde U$ be the node in $\widetilde U$ and $C_u$ be the cluster returned by Lemma~\ref{lem:goodratio}, then we have 
\begin{align*}
    \hat p(C_u) = \Omega( \covereps^4 \eps^5 \hat p(C^*)) 
\end{align*}
\end{lemma}
\begin{proof}
    Note that for any $C^* \in C_{\leq (1+4\covereps)R}$, we have 
\begin{align*}
    \frac{\covers(C^*)+\covereps^2 \dadm(C^*)}{\hat p(C^*)} &\leq \frac{\covers(C^*)+\covereps^2 \dadm(C^*)}{(1 -\covereps) \hat q(C^*)} \\
    &\leq \frac{1 + 4\covereps}{1-\covereps /2} R \leq (1 + 5\covereps) R
\end{align*}
This satisfy the input of Lemma~\ref{lem:goodratio}. So we can find a cluster $C_u$ with ratio at most $(1+5\covereps) R$.

If $\hat p(D(u)) \leq \hat p(\Ncand(u)) / 2$, based on Lemma~\ref{lem:goodratio}, we know that $K(u) \subset C_u$ is alwasy added to $C_u$, so $\hat p(C_u) \geq \hat p(K(u)) \geq \hat p(\Ncand(u)) / 2 \geq \hat p(C^*) / 2$.

On the other hand, if $\hat p(D(u)) \geq \hat p(\Ncand(u)) / 2$. Based on Lemma~\ref{lemma:pvaluebound}, the upper bound $\hat p( D(u))$ is at most 
\begin{align*}
    \hat p(D(u)) = O\left(\frac{\eps^{-1} d(u) \cdot |D(u)|}{\dc(V)}\right)
\end{align*}
and the upper bound for $\hat p(C^*)$ is 
\begin{align*}
    \hat p(C^*) \leq  p(\Ncand(u)) \leq 2 \hat p(D(u)) = O\left(\frac{\eps^{-1} d(u) \cdot |D(u)|}{\dc(V)}\right)
\end{align*}
By Claim~\ref{clm:size}, we know that $C_u$ contains at least $\covereps^3 d(u)$ vertices with $\hat p > 0$. Based on Lemma~\ref{lemma:pvaluebound}, for each vertex in $C_u$, we have 
\begin{align*}
     \hat p_v = \Omega\left(\frac{\covereps \eps^{4} |D(u)|}{\dc(V)}\right) 
\end{align*}
so, the lower bound $\hat p$ value for $C_u$ is at least
\begin{align*}
    \hat p(C_u) = \Omega\left(\frac{\covereps^4 \eps^{4} |D(u)| \cdot d(u)}{\dc(V)}\right) 
\end{align*}
so we have $\hat p(C_u) = \Omega( \covereps^4 \eps^5 \hat p(C^*)) $
\end{proof}

\begin{proof}[Proof of Lemma~\ref{lem:mpcnumberofrounds}]
Note that, based on Lemma~\ref{lem:includeenoughweightsmpc}, for each iteration of the for loop at Line~\ref{line:oneroundstartmpc}, we include the set \(\mathcal{C}_{\leq (1+4\covereps)R}\) with total $\hat p$ value at least \(\Omega(\eps^5 \covereps' \covereps) = \Omega(\eps^9 \covereps^6)\) with probability at least \(\Omega(\eps^5 \covereps')\).

We repeat the for loop for \(\Theta(\log n/(\eps^5 \covereps'))\) iterations. With high probability, we will find a set \(\mathcal{C}_{\leq (1+4\covereps)R}\) with total $p$ value at least \(\Omega(\eps^5 \covereps' \covereps) = \Omega(\eps^9 \covereps^6)\). 

Now, from Lemma~\ref{lemma:eachcuislargeenoughmpc}, for each \(C^* \in \mathcal{C}_{\leq (1+4\covereps)R}\), we can find a sufficiently large set \(C_u\). Thus, \(\mathcal{F}_{\tmpc}\) can cover at least \(\Omega(\covereps^{10} \eps^{14})\) of the \(p\) value.

\end{proof}

\subsection{Wrap-Up: Proof of MPC Algorithm for Theorem~\ref{thm:solving-cluster-LP-sequential}}
Now we are ready to show the key theorem for this section.

\begin{proof}[Proof of Theorem~\ref{thm:solving-cluster-LP-sequential}]

We can obtain the preclustering \( (\calK, E_{\adm}) \) in time \( \widetilde{O}(n) \) by Theorem~\ref{thm:preclustering-proc}. We now argue that this algorithm can be implemented in the MPC model.

In Algorithm~\ref{alg:solveclusterlpframework}, we first compute \( \dc \). This can be done given \( \mathcal{K} \) by appending the cluster label to each edge to indicate whether the endpoints belong to the same cluster in the preclustering. Then, computing \( \dc \) only requires counting the edges that cross different clusters in \( \mathcal{K} \). This step can be implemented using sorting. We will discuss how to solve~\ref{LP:coverclusterlp} later. 

Finally, we need to apply Lemma~\ref{lem:cover-to-cluster} to compute the solution for~\ref{LP:clusterlp}. Implementing Lemma~\ref{lem:cover-to-cluster} in the MPC model is somewhat tricky, but the key observation is that each \( z_S \geq \frac{1}{\tmwu} \). Therefore, for each node \( u \), at most \( O(\tmwu) \) clusters with positive \( z_S \) values will cover \( u \). For each node, we collect all clusters with positive \( z_S \) values. If the final sum \( \sum_{S \ni u} z_S > 1 \), we sort \( z_S \) for \( u \) in any order and mark those \( z_S \) as "greater than 1 for \( u \)" for later clusters. Each \( S \) then collects these marked clusters and removes them from itself. In this process, for each non-singleton atom, we perform this operation only once for a representative, allowing \( S \) to remove the non-singleton atoms. Since at most \( O(\tmwu) \) clusters cover \( u \), this process requires at most \( O(\tmwu \cdot n) \) total space.

Algorithm~\ref{alg:mw} and Algorithm~\ref{alg:disjointfamilympc} are naturally parallelized, requiring only \( O(\poly(1/\eps)) \) rounds for execution. From Claim~\ref{clm:overlap}, we know that with high probability, each vertex is contained in at most \( O(\log n) \) clusters from \( \{C_u\}_{u \in U} \), so storing the candidate set \( \Ncand \) for \( U \) requires at most \( O(n \log n) \) space. We can then distribute all edge information to execute Lemma~\ref{lem:goodratio} for each \( u \in U \).

The more challenging part is applying Lemma~\ref{lem:goodratio}. For each \( \Ncand(u) \), we need to sample \( O(\poly(1/\eps)) \) nodes and enumerate all possible subsets of the sampled nodes. There are two methods for performing this enumeration:

\begin{enumerate}
    \item Store all subsets of the sampled nodes. For each \( \Ncand(u) \), this requires \( \tilde{O}(2^{\poly(1/\eps)}|\Ncand(u)|) \) space to store the edge relations between the sampled nodes and all other nodes. Once these relations are recorded, Algorithm~\ref{alg:GenerateCluster} can determine whether to add a node to \( T \). This method requires only \( O(1) \) rounds but consumes a total space of \( \tilde{O}(2^{\poly(1/\eps)} m) \).
    
    \item Enumerate all possible subsets one by one. This avoids the need to record edge information for the sampled nodes, but it requires an additional \( O(2^{\poly(1/\eps)}) \) rounds for each iteration.
\end{enumerate}

Combining all these points, we conclude that we can either:
- Spend \( 2^{\poly (1/\eps)} \) rounds with \( O(n^{\delta}) \) memory per machine and total memory \( \tilde{O}(\poly (1/\eps) m) \), or
- Spend \( \poly (1/\eps) \) rounds with \( O(n^{\delta}) \) memory per machine and total memory \( \tilde{O}(2^{\poly (1/\eps)} m) \) 
to solve~\ref{LP:clusterlp}.

\end{proof}

%% file: rounding-algorithms.tex
In this section, we present fast rounding algorithms for the cluster
LP. We begin by providing intuition behind the rounding algorithms and
by explaining the rounding techniques used
in~\cite{cao2024understanding}.  In Section \ref{sec:roundingnearlin},
we show how to implement the rounding in nearly linear time. Then, in
Section~\ref{sec:roundingsublinear}, we demonstrate how to perform
rounding in the sublinear model.

Note that we have already obtained a solution
to~\ref{LP:clusterlp}. In~\cite{cao2024understanding}, the authors
showed how to round a~\ref{LP:clusterlp} solution to an integral
solution. The approximation ratio is given by the following theorem.

\begin{theorem} \cite{cao2024understanding}
\label{thm:rounding}
There exists an algorithm that, given a feasible solution for the cluster LP, produces a clustering whose objective value is at most $\pureclusterlpratio$ times that of the cluster LP solution. 
\label{thm:main:algo:computer}
\end{theorem}

To be more precise, the algorithm from Theorem \ref{thm:rounding} consists of two different rounding algorithms, one is executed with probability $p = \frac{\pureclusterlpratio}{2}$ the other with probability $1-p$. The two algorithms in question are the cluster-based rounding (Algorithm \ref{alg:cluster-based}) and the pivot-based rounding (Algorithm \ref{alg:correlated-rounding}).  We need to show that we can implement both algorithms in nearly linear time.

\paragraph{Algorithm Description.} 
The cluster-based rounding algorithm (Algorithm~\ref{alg:cluster-based}) selects clusters \( S \) based on their \( z_S \) values. In each round, we randomly choose a cluster \( S \) with probability proportional to \( z_S \). We then add all nodes from \( S \) that are still in the graph to a new cluster and remove those nodes from the graph. This process is repeated until the graph is empty.

The pivot-based algorithm (Algorithm~\ref{alg:correlated-rounding})
follows a different approach. Instead of selecting a set directly, at
each round, we randomly choose a node \( u \) that is still in the
graph, referred to as the pivot. We then form a cluster using this
pivot as follows.  For each small \pedge \( (u, v) \), we include \( v \) in the cluster. For the remaining \pedges, we apply correlated rounding, where we randomly choose a set \( S \) containing \( u \) and include all remaining \pedges if they are in \( S \). For \medges, we use independent rounding, where the probability of adding a node to the cluster is determined by the distance metric:
\[
1 - x_{uv} = \sum_{S \supseteq \{u, v\}} z_S.
\]
The algorithm then removes all nodes that have been clustered in the current round and repeats the process until the graph is empty.

\begin{algorithm}[h]
    \caption{Cluster-Based Rounding}
        \label{alg:cluster-based}
    \begin{algorithmic}
        \STATE $\calC \gets \emptyset, V' \gets V$
        \WHILE{$V' \neq \emptyset$}
            \STATE randomly choose a cluster $S\subseteq V$, with probabilities $\frac{z_S}{\sum_{S'} z_{S'}}$
            \IF{$V' \cap S \neq \emptyset$} 
            \STATE $\calC \gets \calC \cup \{V' \cap S\}$, $V' \gets V' \setminus S$ \ENDIF
        \ENDWHILE
        \STATE return $\calC$
    \end{algorithmic}    
\end{algorithm}

\begin{algorithm}
    \caption{Pivot-Based Rounding (with correlated mid-range for \(+\) edges)}
    \label{alg:correlated-rounding}
    \begin{algorithmic}[1]
        \STATE \(\mathcal{C} \gets \emptyset\), \(V' \gets V\)
        \WHILE{\(V' \neq \emptyset\)}
            \STATE choose a pivot \(u \in V'\) uniformly at random
            \STATE \(C \gets \{u\}\)
            \FOR{each \(v \in V' \cap N^-(u)\)}
                \STATE independently add \(v\) to \(C\) with probability \(1 - x_{uv}^2\)
            \ENDFOR
            \STATE \(C_{\ell} \gets \{v \in V' \cap N^+(u): x_{uv} \le 0.40\}\)
            \STATE \(C_{m} \gets \{v \in V' \cap N^+(u): 0.40 < x_{uv} \le 0.57\}\)
            \STATE \(C_{h} \gets \{v \in V' \cap N^+(u): x_{uv} > 0.57\}\)
            \FOR{each \(v \in C_h\)}
                \STATE independently add \(v\) to \(C\) with probability \(1 - x_{uv}\)
            \ENDFOR
            \STATE sample a random set \(S \ni u\) according to the set distribution \(\{z_S\}_{S \ni u}\) 
            \STATE \(C \gets C \cup C_\ell \cup \bigl(C_m \cap S\bigr)\)
            \STATE \(\mathcal{C} \gets \mathcal{C} \cup \{C\}\), \(V' \gets V' \setminus C\)
        \ENDWHILE
        \STATE \textbf{return} \(\mathcal{C}\)
    \end{algorithmic}
\end{algorithm}


\subsection{Nearly Linear Time Rounding
    Algorithm}\label{sec:roundingnearlin}

\paragraph{Implementation of Cluster-Based Rounding.}
We are given a solution $z$ to the cluster LP with $\cost(z) \leq
    (1+\epsilon) \optc$ and $|\supp(z)| = O(n)$.  Moreover, by
    Lemma \ref{lem:cover-to-cluster}, for each $S \in \supp(z)$, we have
    $z_S \geq \Delta$ for some constant $\Delta$.

We will implement Algorithm \ref{alg:cluster-based} as follows.  For
each vertex we sample a random variable $k_v$ from an exponential
distribution with rate $z_S$.  To do this, we sample a value $k_S$ for each set with $z_S > 0$ by sampling a value uniformly at random from $p_S \in (0,1)$ and then setting 
$$k_S = \lfloor \frac{n^c}{z_S} \log{\frac{1}{p_S}}\rfloor,$$
where $c$ is a fixed constant.
Then, for each vertex $v$, let $k_v$ be the smallest $k_S$ among sets $S$ that contain the vertex $v$. All
vertices with the same $k_v$ are assigned to the same cluster.
Algorithm \ref{alg:sublinear-cluster-based}
is equivalent to Algorithm \ref{alg:cluster-based}. Indeed, let
$X_1, \dots, X_l$ be exponentially distributed random variables each
with rate $r_i$. The probability that $X_i$ is the minimum among
$X_1, \dots, X_1$ is equal to
\[
    Pr[X_i = \min\{X_1,\dots,X_l\}] = \frac{r_i}{\sum_{j} r_j}.
\]

We emphasize that, given the solution from Theorem~\ref{thm:solving-cluster-LP-sequential}, Algorithm~\ref{alg:sublinear-cluster-based} runs in \(\widetilde{O}(n)\) time in the sublinear model. We conclude with the following lemma for the cluster-based rounding algorithm:

\begin{lemma}
\label{lem:clusterrounding}
Given a solution \( \{ z_S \} \) to~\ref{LP:clusterlp} such that for each \( S \in \supp(z) \), we have \( z_S \geq \Delta \) for some constant \( \Delta \). Then, Algorithm~\ref{alg:cluster-based} runs in \( \widetilde{O}(n/\Delta) \) time in the sublinear model.
\end{lemma}

\begin{algorithm}
    \caption{Cluster-Based Rounding}
    \label{alg:sublinear-cluster-based}
    \begin{algorithmic}
        \FOR{$S \in \mathcal{S}$}
         \STATE $k_S
    =  \lfloor\frac{n^c}{z_s}\log \frac{1}{p_s} \rfloor$, where $p_s$
    is uniformly chosen from $(0, 1)$ 
      \ENDFOR
        \FOR{$v \in V$}
        \STATE $k_v = \min \{ k_S \mid S \ni v\}$
        \ENDFOR
        \STATE Put all nodes with same $k_v$ value into same cluster
    \end{algorithmic}    
\end{algorithm}

\paragraph{Implementation of Pivot-Based Rounding.} Now we discuss how to implement Algorithm \ref{alg:correlated-rounding} in nearly linear time. 
Note that each vertex is contained in at most $\frac{1}{\Delta}$ clusters $S \in \supp(z)$.  Thus a fixed pivot $u$ belongs to at most $1/\Delta$ sets $S$ such that $z_S > 0$.  We therefore only need to consider $1/\Delta$ sets to choose one of these sets according to the respective probability distribution.
Visiting the \pedges incident to each pivot requires at most $O(m)$ time over the course of the algorithm.   

For the \medges, the main challenge is listing all \medges from \(
V \cap N^-(u) \) and then performing independent rounding. Our key observation is that \( x_{uv} < 1 \) if and only if there exists some \( S \) such that \( S \) contains both \( u \) and \( v \) and \( z_S > 0 \). Thus, we only need to iterate over all vertices in $\cup_{S \ni u, z_S > 0}~S$. 

Since each \( z_S \geq \Delta \) for \( S \in \supp(z) \), each node \( v \) with \( x_{uv} < 1 \) contributes at least \( 1/\Delta \) to the LP value. Therefore, there are at most \( O(m / \Delta) \) \medges across all \( S \). Since each \medge is visited at most once, listing all \medges takes at most \( O(m / \Delta) \) time.

Once we process the pivot, we need to update \( S \) to remove nodes that have been clustered. Since each nonzero \( z_S > \Delta \), the total update time is at most \( O(n / \Delta) \). 

Combining everything, we obtain the following lemma for pivot-based rounding.





\begin{lemma}
\label{lem:pivotbasedlinear}
Given a solution \( \{ z_S \} \) to~\ref{LP:clusterlp} such that for each \( S \in \supp(z) \), we have \( z_S \geq \Delta \) for some constant \( \Delta \). Then, Algorithm~\ref{alg:correlated-rounding} runs in \( \widetilde{O}(m / \Delta) \) time.
\end{lemma}

\subsection{Rounding in Sublinear Model}
\label{sec:roundingsublinear}

To argue that Algorithm~\ref{alg:correlated-rounding} runs in sublinear time, we make the following observation. For a fixed pivot \( r \), let 
\[
U_r = \bigcup_{S \ni r: z_S > 0} S
\]
i.e., \( U_r \) is the set of all vertices that occur together with \( r \) in some set in the support of \( z \). If \( v \not\in U_r \), then \( x_{uv} = 1 \), meaning Algorithm~\ref{alg:correlated-rounding} never considers \( v \) when \( r \) is chosen as the pivot. Thus, Algorithm~\ref{alg:correlated-rounding} only considers \( U_r \) for one step of iteration.

Moreover, we can iterate over all \( S \ni r \) with \( z_S > 0 \) and all \( v \in S \) in time \( \frac{1}{\Delta} |U_r| \) since each vertex is contained in at most \( \frac{1}{\Delta} \) sets \( S \in \supp(z) \). Hence, we can implement one step of the algorithm in time at most \( \widetilde{O}(|U_r| / \Delta) \).

On the other hand, for each vertex \( v \in U_r \), the probability that it is included in the cluster of \( r \) is at least \( 1 - x_{rv} \geq \Delta \) for $+$edges and \( 1 - x_{rv} \cdot x_{rv} \geq \Delta \) for $-$edges. Thus, in expectation, we remove at least \( \Delta |U_r| \) vertices from the current graph. 

Combining these points, in Algorithm~\ref{alg:correlated-rounding}, we pay \( O(1/\Delta^2) \) for each node in expectation. Therefore, the final running time is given by the following lemma.

\begin{lemma}
\label{lem:pivotbasedsublinear}
Given a solution \( \{ z_S \} \) to~\ref{LP:clusterlp} such that for each \( S \in \supp(z) \), we have \( z_S \geq \Delta \) for some constant \( \Delta \). Then, Algorithm~\ref{alg:correlated-rounding} runs in \( \widetilde{O}(n / \Delta^2) \) time.
\end{lemma}

Now, we can prove the main theorem regarding the rounding algorithm.

\begin{proof}[Proof of Theorem~\ref{thm:rounding-cluster-LP-main}]
The approximation ratio follows from Theorem~\ref{thm:rounding}. The runtime of cluster-based rounding is given in Lemma~\ref{lem:clusterrounding}, and the runtime of pivot-based rounding is given in Lemma~\ref{lem:pivotbasedsublinear}. Combining these results establishes the main theorem.
\end{proof}

\paragraph{Expected Guarantee} 
We highlight one final point regarding general rounding algorithms in the sublinear model. Typically, a rounding algorithm provides an expected guarantee. To achieve a high-probability guarantee, we can run the rounding algorithm multiple times and select the best outcome. However, in the sublinear model, it is unclear how to determine which output is the best, as computing the clustering cost sublinear time is challenging.

For our rounding algorithm, we achieve a high-probability guarantee. This is because our rounding algorithm sets an atom as a cluster if its \( \dc \) value is very small and never splits an atom. We can apply Lemma~\ref{lem:estimate-dc} to estimate the cost of the clustering. The estimation cost is \( \esteps \cdot \dc(V) \). 

Since \( \dc(V) = O(\frac{1}{\epsilon^{12}}) \opt \), the estimation cost is always small enough compared to the optimal solution.

%% file: main.bbl
\newcommand{\etalchar}[1]{$^{#1}$}
\begin{thebibliography}{KCMNT08}

\bibitem[ACN08]{ACN08}
Nir Ailon, Moses Charikar, and Alantha Newman.
\newblock Aggregating inconsistent information: {R}anking and clustering.
\newblock {\em Journal of the ACM}, 55(5):1--27, 2008.

\bibitem[AHK{\etalchar{+}}09]{agrawal2009generating}
Rakesh Agrawal, Alan Halverson, Krishnaram Kenthapadi, Nina Mishra, and Panayiotis Tsaparas.
\newblock Generating labels from clicks.
\newblock In {\em Proceedings of the Second ACM International Conference on Web Search and Data Mining (WSDM)}, pages 172--181, 2009.

\bibitem[AHK12]{AHK}
Sanjeev Arora, Elad Hazan, and Satyen Kale.
\newblock The multiplicative weights update method: {A} meta-algorithm and applications.
\newblock {\em Theory of Computing}, 8(1):121--164, 2012.

\bibitem[ARS09]{arasu2009large}
Arvind Arasu, Christopher R{\'e}, and Dan Suciu.
\newblock Large-scale deduplication with constraints using dedupalog.
\newblock In {\em Proceedings of the 25th IEEE International Conference on Data Engineering (ICDE)}, pages 952--963, 2009.

\bibitem[AW22]{DBLP:conf/innovations/Assadi022}
Sepehr Assadi and Chen Wang.
\newblock Sublinear time and space algorithms for correlation clustering via sparse-dense decompositions.
\newblock In {\em Proceedings of the 13th Conference on Innovations in Theoretical Computer Science (ITCS)}, volume 215 of {\em LIPIcs}, pages 10:1--10:20, 2022.

\bibitem[BBC04]{BBC04}
Nikhil Bansal, Avrim Blum, and Shuchi Chawla.
\newblock Correlation clustering.
\newblock {\em Machine learning}, 56(1):89--113, 2004.

\bibitem[BCC{\etalchar{+}}24]{behnezhad2024fullydynamiccorrelationclustering}
Soheil Behnezhad, Moses Charikar, Vincent {Cohen-Addad}, Alma Ghafari, and Weiyun Ma.
\newblock Fully dynamic correlation clustering: Breaking 3-approximation, 2024.

\bibitem[BCMT22]{DBLP:conf/focs/BehnezhadCMT22}
Soheil Behnezhad, Moses Charikar, Weiyun Ma, and Li{-}Yang Tan.
\newblock Almost 3-approximate correlation clustering in constant rounds.
\newblock In {\em Proceedings of 63rd Annual IEEE Symposium on Foundations of Computer Science, (FOCS)}, pages 720--731, 2022.

\bibitem[BCMT23]{DBLP:conf/soda/BehnezhadCMT23}
Soheil Behnezhad, Moses Charikar, Weiyun Ma, and Li{-}Yang Tan.
\newblock Single-pass streaming algorithms for correlation clustering.
\newblock In {\em Proceedings of the 2023 ACM-SIAM Symposium on Discrete Algorithms (SODA)}, pages 819--849, 2023.

\bibitem[BGU13]{bonchi2013overlapping}
Francesco Bonchi, Aristides Gionis, and Antti Ukkonen.
\newblock Overlapping correlation clustering.
\newblock {\em Knowledge and Information Systems}, 35(1):1--32, 2013.

\bibitem[CCL{\etalchar{+}}24]{cao2024understanding}
Nairen Cao, Vincent {Cohen-Addad}, Euiwoong Lee, Shi Li, Alantha Newman, and Lukas Vogl.
\newblock Understanding the cluster linear program for correlation clustering.
\newblock In {\em Proceedings of the 56th Annual ACM Symposium on Theory of Computing (STOC)}, pages 1605--1616, 2024.

\bibitem[CGW05]{CGW05}
Moses Charikar, Venkatesan Guruswami, and Anthony Wirth.
\newblock Clustering with qualitative information.
\newblock {\em Journal of Computer and System Sciences}, 71(3):360--383, 2005.

\bibitem[CHS24]{cao2023breaking}
Nairen Cao, Shang-En Huang, and Hsin-Hao Su.
\newblock Breaking 3-factor approximation for correlation clustering in polylogarithmic rounds.
\newblock In {\em Proceedings of the 2024 {ACM-SIAM} Symposium on Discrete Algorithms (SODA)}, 2024.

\bibitem[CKP08]{chakrabarti2008graph}
Deepayan Chakrabarti, Ravi Kumar, and Kunal Punera.
\newblock A graph-theoretic approach to webpage segmentation.
\newblock In {\em Proceedings of the 17th International Conference on World Wide Web (WWW)}, pages 377--386, 2008.

\bibitem[CLLN23]{CLLN23}
Vincent {Cohen-Addad}, Euiwoong Lee, Shi Li, and Alantha Newman.
\newblock Handling correlated rounding error via preclustering: {A} 1.73-approximation for correlation clustering.
\newblock In {\em Proceedings of the 64rd Annual IEEE Symposium on Foundations of Computer Science (FOCS)}, pages 1082--1104, 2023.

\bibitem[CLM{\etalchar{+}}21]{CLMNP21}
Vincent Cohen{-}Addad, Silvio Lattanzi, Slobodan Mitrovic, Ashkan Norouzi{-}Fard, Nikos Parotsidis, and Jakub Tarnawski.
\newblock Correlation clustering in constant many parallel rounds.
\newblock In {\em Proceedings of the 38th International Conference on Machine Learning (ICML)}, pages 2069--2078, 2021.

\bibitem[CLN22]{CLN22}
Vincent {Cohen-Addad}, Euiwoong Lee, and Alantha Newman.
\newblock Correlation clustering with {S}herali-{A}dams.
\newblock In {\em Proceedings of 63rd Annual IEEE Symposium on Foundations of Computer Science (FOCS)}, pages 651--661, 2022.

\bibitem[CLP{\etalchar{+}}24]{cohen2024combinatorial}
Vincent {Cohen-Addad}, David~Rasmussen Lolck, Marcin Pilipczuk, Mikkel Thorup, Shuyi Yan, and Hanwen Zhang.
\newblock Combinatorial correlation clustering.
\newblock In {\em Proceedings of the 56th Annual ACM Symposium on Theory of Computing (STOC)}, pages 1617--1628, 2024.

\bibitem[CMSY15]{CMSY15}
Shuchi Chawla, Konstantin Makarychev, Tselil Schramm, and Grigory Yaroslavtsev.
\newblock Near optimal {LP} rounding algorithm for correlation clustering on complete and complete $k$-partite graphs.
\newblock In {\em Proceedings of the 47th Annual ACM Symposium on Theory of Computing (STOC)}, pages 219--228, 2015.

\bibitem[CSX12]{chen2012clustering}
Yudong Chen, Sujay Sanghavi, and Huan Xu.
\newblock Clustering sparse graphs.
\newblock In {\em Advances in Neural Information Processing Systems (Neurips)}, pages 2204--2212, 2012.

\bibitem[DMM24]{prunpivot}
Mina Dalirrooyfard, Konstantin Makarychev, and Slobodan Mitrovi\'{c}.
\newblock Pruned pivot: correlation clustering algorithm for dynamic, parallel, and local computation models.
\newblock In {\em Proceedings of the 41st International Conference on Machine Learning}, ICML'24. JMLR.org, 2024.

\bibitem[KCMNT08]{kalashnikov2008web}
Dmitri~V. Kalashnikov, Zhaoqi Chen, Sharad Mehrotra, and Rabia Nuray-Turan.
\newblock Web people search via connection analysis.
\newblock {\em IEEE Transactions on Knowledge and Data Engineering}, 20(11):1550--1565, 2008.

\bibitem[KYNK14]{kim2014image}
Sungwoong Kim, Chang~D Yoo, Sebastian Nowozin, and Pushmeet Kohli.
\newblock Image segmentation using higher-order correlation clustering.
\newblock {\em IEEE Transactions on Pattern Analysis and Machine Intelligence}, 36(9):1761--1774, 2014.

\bibitem[MC24]{makarychev2024single}
Konstantin Makarychev and Sayak Chakrabarty.
\newblock Single-pass pivot algorithm for correlation clustering. {K}eep it simple!
\newblock {\em Advances in Neural Information Processing Systems (NeurIPS)}, 36, 2024.

\bibitem[PST95]{plotkin1995fast}
Serge~A. Plotkin, David~B. Shmoys, and {\'E}va Tardos.
\newblock Fast approximation algorithms for fractional packing and covering problems.
\newblock {\em Mathematics of Operations Research}, 20(2):257--301, 1995.

\bibitem[YV18]{yaroslavtsev2018massively}
Grigory Yaroslavtsev and Adithya Vadapalli.
\newblock Massively parallel algorithms and hardness for single-linkage clustering under $\ell_p$-distances.
\newblock In {\em Proceedings of 35th International Conference on Machine Learning (ICML)}, page 5596–5605, 2018.

\end{thebibliography}
